\documentclass[
aps,
prx,
amsfonts,amsmath,amssymb,
notitlepage,
reprint,
10pt,
superscriptaddress,
floatfix,
nofootinbib,
noeprint
]{revtex4-2}

\usepackage[english]{babel}
\usepackage[utf8]{inputenc}
\usepackage[T1]{fontenc}
\usepackage{amsthm}
\usepackage{mathrsfs}
\usepackage{bbm}
\usepackage{bm}
\usepackage[shortlabels]{enumitem}
\usepackage{graphicx}
\usepackage[dvipsnames]{xcolor}
\usepackage{float}
\usepackage{subdepth}
\usepackage{fnpct}
\usepackage{booktabs}
\usepackage{multirow}
\usepackage[math]{cellspace}
\usepackage{soul}

\usepackage[colorlinks,
linkcolor=BrickRed,
citecolor=MidnightBlue,
urlcolor=MidnightBlue,
bookmarks=true,        
bookmarksopen=true,    
bookmarksnumbered=true,
]{hyperref}

\newtheorem*{lemma*}{Lemma}
\newtheorem*{theorem*}{Theorem}

\newcommand{\Pf}{\mathrm{Pf}\,}

\newcommand{\Tr}{\mathrm{Tr}}
\renewcommand{\Re}{\mathrm{Re}}
\renewcommand{\Im}{\mathrm{Im}}
\renewcommand{\vec}[1]{\underline{#1}}
\renewcommand{\i}{\mathrm{i}}

\newcommand{\id}{\mathbbm{1}}

\newcommand{\sgn}{\mathrm{sgn}\,}
\newcommand{\scL}{\mathcal{L}}
\newcommand{\scS}{\mathcal{S}}

\newcommand{\scW}{\mathcal{W}}
\newcommand{\scD}{\mathcal{D}}
\newcommand{\scN}{\mathcal{N}}
\newcommand{\scO}{\mathcal{O}}
\newcommand{\scZ}{\mathcal{Z}}
\newcommand{\scw}{\mathfrak{w}}
\newcommand{\scs}{\mathfrak{s}}
\newcommand{\scC}{\mathcal{C}}
\newcommand{\scP}{\mathcal{P}}
\newcommand{\ut}{\tilde{u}}
\newcommand{\bt}{\tilde{b}}
\newcommand{\ct}{\tilde{c}}
\newcommand{\Wt}{\widetilde{W}}
\newcommand{\wt}{\widetilde{w}}
\newcommand{\Gammat}{\widetilde{\Gamma}}
\newcommand{\Dt}{\widetilde{D}}
\newcommand{\ch}{\mathrm{ch}}
\newcommand{\dsL}{\mathbb{L}}
\newcommand{\dsA}{\mathbb{A}}
\newcommand{\dsD}{\mathbb{D}}
\newcommand{\dsVL}{\mathbb{V}_\mathrm{L}}
\newcommand{\dsVR}{\mathbb{V}_\mathrm{R}}
\newcommand{\tA}{\widetilde{\dsA}}
\newcommand{\tL}{\widetilde{\dsL}}
\newcommand{\setV}{\mathfrak{V}}
\newcommand{\setU}{\mathfrak{U}}

\newcommand{\setI}{\mathfrak{I}}
\newcommand{\bX}{\overline{X}}
\newcommand{\bZ}{\overline{Z}}

\newcommand{\bO}{\overline{O}}
\newcommand{\lt}{\tilde{\ell}}

\newcounter{ls}

\begin{document}

\title{Exactly solvable dissipative dynamics and one-form strong-to-weak spontaneous symmetry breaking in interacting two-dimensional spin systems}

\author{Lucas S\'a}
\affiliation{TCM Group, Cavendish Laboratory, Ray Dolby Centre, University of Cambridge, JJ Thomson Avenue, Cambridge, CB3 0US UK}

\author{Benjamin B\'eri}
\affiliation{TCM Group, Cavendish Laboratory, Ray Dolby Centre, University of Cambridge, JJ Thomson Avenue, Cambridge, CB3 0US UK}
\affiliation{DAMTP, University of Cambridge, Wilberforce Road, Cambridge, CB3 0WA, UK}

\begin{abstract}
We study the dissipative dynamics of a class of interacting ``gamma-matrix'' spin models coupled to a Markovian environment. For spins on an arbitrary graph, we construct a Lindbladian that maps to a non-Hermitian model of free Majorana fermions hopping on the graph with a background classical $\mathbb{Z}_2$ gauge field. We show, analytically and numerically, that the steady states and relaxation dynamics are qualitatively independent of the choice of the underlying graph, in stark contrast to the Hamiltonian case. We also show that the exponentially many steady states provide a concrete example of mixed-state topological order, in the sense of strong-to-weak spontaneous symmetry breaking of a one-form symmetry. While encoding only classical information, the steady states still exhibit long-range quantum correlations. Afterward, we examine the relaxation processes toward the steady state by numerically computing decay rates, which we generically find to be finite, even in the dissipationless limit. However, we identify symmetry sectors where fermion-parity conservation is enhanced to fermion-number conservation, where we can analytically bound the decay rates and prove that they vanish in the limits of both infinitely weak and infinitely strong dissipation. Finally, we show that while the choice of coherent dynamics is very flexible, exact solvability strongly constrains the allowed form of dissipation. Our work establishes an analytically tractable framework to explore nonequilibrium quantum phases of matter and the relaxation mechanisms toward them.
\end{abstract}

\maketitle

\section{Introduction}

Interacting spin systems in two dimensions (2D) are a fertile ground for fundamental and applied physics. Among the many phases of matter that 2D spin systems can host, quantum spin liquids (QSLs)~\cite{anderson1973,anderson1987} stand out due to their exotic properties, such as the absence of local magnetic ordering, fractionalized excitations, and often topological order~\cite{kalmeyer1987,rokhsar1988,wen1991,senthil2000,moessner2001,kitaev2006AnnPhys}. These properties also make QSLs and other topologically ordered states potential candidates to realize fault-tolerant quantum computers and topological quantum memories~\cite{dennis2002JMP,kitaev2001Anyons,nayak2008RMP}. These fundamental and potentially technological features motivate intense theoretical and experimental research into QSLs~\cite{lee2008,savary2017RPP,zhou2017RMP,takagi2019NatRev,broholm2020Sci,chamorro2020,clark2021AnnuRev,trebst2022,khatua2023PhysRep,hermanns2018AnuuRev,knolle2019annurev}.

In the rich landscape of QSLs, the Kitaev honeycomb model~\cite{kitaev2006AnnPhys}---a model of spins-$1/2$ on the honeycomb lattice with highly anisotropic Heisenberg interactions---occupies a special place, as it provides a rare example of an exactly solvable model with both gapped and gapless QSL phases. The fractionalization of the spins into free itinerant Majorana fermions and static classical $\mathbb{Z}_2$ gauge fields yields not only a detailed knowledge of the ground states but also key insights into the dynamics and signatures of the anyonic excitations~\cite{knolle2014PRL,halasz2014PRB,knolle2015PRB,nasu2015PRB,nasu2016NatPhys,nasu2017PRL,gohlke2017PRL,hermanns2018AnuuRev,knolle2019annurev}.
Several variants of the Kitaev model that preserve exact solvability while hosting distinct types of QSLs have been proposed~\cite{yao2007PRL,yao2009PRL,chua2011PRB,wu2009PRB,willans2011PRB}.
However, imperfections in devices or samples and external influences are, to a certain extent, unavoidable in realistic systems, including in current-term noisy quantum hardware. It is thus natural to ask what the fate of topological order and the QSL state is once the system is explicitly coupled to the environment: which, if any, exactly solvable QSL models can still be constructed, and what is their phenomenology?

In this paper, we show that Kitaev's approach can be used to construct an exactly solvable interacting model on an \emph{arbitrary} two-dimensional graph, while including coupling to a Markovian bath that leads to Lindbladian dynamics. 
The local degrees of freedom of these models are ``gamma matrices''~\cite{wu2009PRB,yao2009PRL,chua2011PRB}, replacing Pauli matrices of spin-1/2 models by elements of a suitable Clifford algebra. 
In stark contrast to the closed-system limit, where qualitatively different states arise for different choices of 2D spin models~\cite{kitaev2006AnnPhys,yao2007PRL,yao2009PRL,chua2011PRB}, we show that the qualitative properties of these models are largely insensitive to details of the coherent interactions (e.g., graph geometry or relative coupling strengths). In particular, the steady-state manifold is completely determined by symmetry and has the same structure for any graph or parameters of the model. The transient dynamics are also qualitatively similar. 
Two special cases of our model on the square lattice were studied in Refs.~\onlinecite{shackleton2024,gidugu2024,dai2023PRB}. There, the authors identified the steady states, reduced the Lindbladian to its free-fermion representation inside fixed symmetry sectors, and numerically studied the dominant relaxation timescales. Other graphs have not been investigated; most prominently missing is a free-fermion dissipative honeycomb model (the spin-$1/2$ honeycomb model of Ref.~\onlinecite{hwang2024Quantum}, which, while sharing the same steady states, is not free-fermion-reducible; see also Ref.~\onlinecite{pocklington2025}). 

In addition to providing key insights into the dynamics of dissipative QSLs, the gamma-matrix model is also particularly useful to study the intensely debated nature of mixed-state topological order~\cite{lee2025quantum,wang2024SciPost,wang2025PRXQ,chen2024PRL,dai2025PRB,sohal2025PRXQ,sang2024PRL,li2024arXiv,wang2025PRXQanalogEE,lessa2025arXiv,sang2024PRX,li2025PRB,ellison2025PRXQ}. In the ground-state physics of pure states, the characteristic features of topological order, such as robust ground-state degeneracy and long-range correlations, go hand-in-hand with the ability to store and manipulate quantum information coherently. On the other hand, it has been understood that mixed states are richer and allow exotic phases of matter where a quantum memory is degraded to a classical memory but still exhibits nontrivial long-range quantum correlations~\cite{bao2023arXiv,fan2024PRXQ,li2025PRB,ellison2025PRXQ,kim2024arXiv}. It was also recognized~\cite{lee2023PRXQ} that the transition between a coherent ground-state phase and an incoherent mixed-state phase can happen sharply at a critical decoherence strength (or, in a continuous-time description, dynamically at a critical time), when a strong symmetry of the system is spontaneously broken~\cite{lee2023PRXQ,ma2023arXiv,sala2024PRB,lessa2025PRXQ,zhang2025PRB,ma2025PRXQ,moharramipour2024PRXQ,gu2024arXiv,liu2024arXiv,weinstein2025PRL,ogunnaike2023PRL,xu2025PRB,huang2025PRB,zhang2024arXiv,guo2024arXiv,kim2024arXiv,ando2024arXiv,chen2025PRB,sun2024arXiv,orito2025PRB,lu2024arXiv,feng2025arXiv} to a weak symmetry. If the spontaneously broken symmetry is a higher-form symmetry~\cite{zhang2025PRB}, i.e., whose symmetry operators and charged operators are extended objects, then the resulting mixed state exhibits a sense of mixed-state topological order. We show that the gamma-matrix models give examples of such strong-to-weak spontaneous symmetry breaking (SW-SSB) and that its steady-state manifold forms a classical memory with such quantum mixed-state topological order.

Finally, we can ask where our work fits in the broader landscape of exactly solvable open quantum systems, which has steadily grown in recent years to include Lindbladian dynamics of quadratic systems~\cite{prosen2008,prosen2008prl,prosen2010jstat,prosen2010njp,zunkovivc2010,prosen2010jphysa,znidaric2010jphysa,prosen2011ilievski,zunkovic2014}, free Hamiltonians with dephasing~\cite{znidaric2010,znidaric2011,temme2012,eisler2011}, steady states of boundary-driven interacting spin chains~\cite{prosen2011a,prosen2011b,prosen2013prl,karevski2013,prosen2014,popkov2015,prosen2015REVIEW,buca2018,popkov2020prl}, interacting Lindbladians that are mapped to Bethe-ansatz integrable non-Hermitian Hamiltonians~\cite{medvedyeva2016,rowlands2018,shibata2019a,shibata2019b,ziolkowska2020SciPost}, collective spin models~\cite{ribeiro2019}, triangular Lindbladians~\cite{torres2014PRA,nakagawa2020PRL,buca2020NJP}, and quantum-channel circuits built from integrable Trotterizations~\cite{sa2021PRB,deleeuw2021PRL,deleeuw2022ARXIV,su2022PRB,znidaric2024b}. Our model stands apart from these examples because it displays features expected from generic open quantum systems, e.g., parametric separation of timescales~\cite{wang2020PRL,sommer2021PRR,li2022PRB,hartmann2024,popkov2021PRL,nakagawa2020PRL,Zhou2021PRR}, anomalous relaxation~\cite{sa2022PRR,garcia2023PRD2,shackleton2024,mori2024,yoshimura2024,yoshimura2025}, and SW-SSB~\cite{lee2023PRXQ,ogunnaike2023PRL,ma2023arXiv,ma2025PRXQ,sala2024PRB,lessa2025PRXQ,xu2025PRB,moharramipour2024PRXQ,gu2024arXiv,huang2025PRB,zhang2024arXiv,zhang2025PRB,liu2024arXiv,guo2024arXiv,weinstein2025PRL,kim2024arXiv,ando2024arXiv,chen2025PRB,sun2024arXiv,orito2025PRB,lu2024arXiv,feng2025arXiv},
while retaining analytic tractability and being amenable to efficient large-scale numerical treatment.

\begin{figure*}[t]
    \centering
    \includegraphics[width=0.97\textwidth]{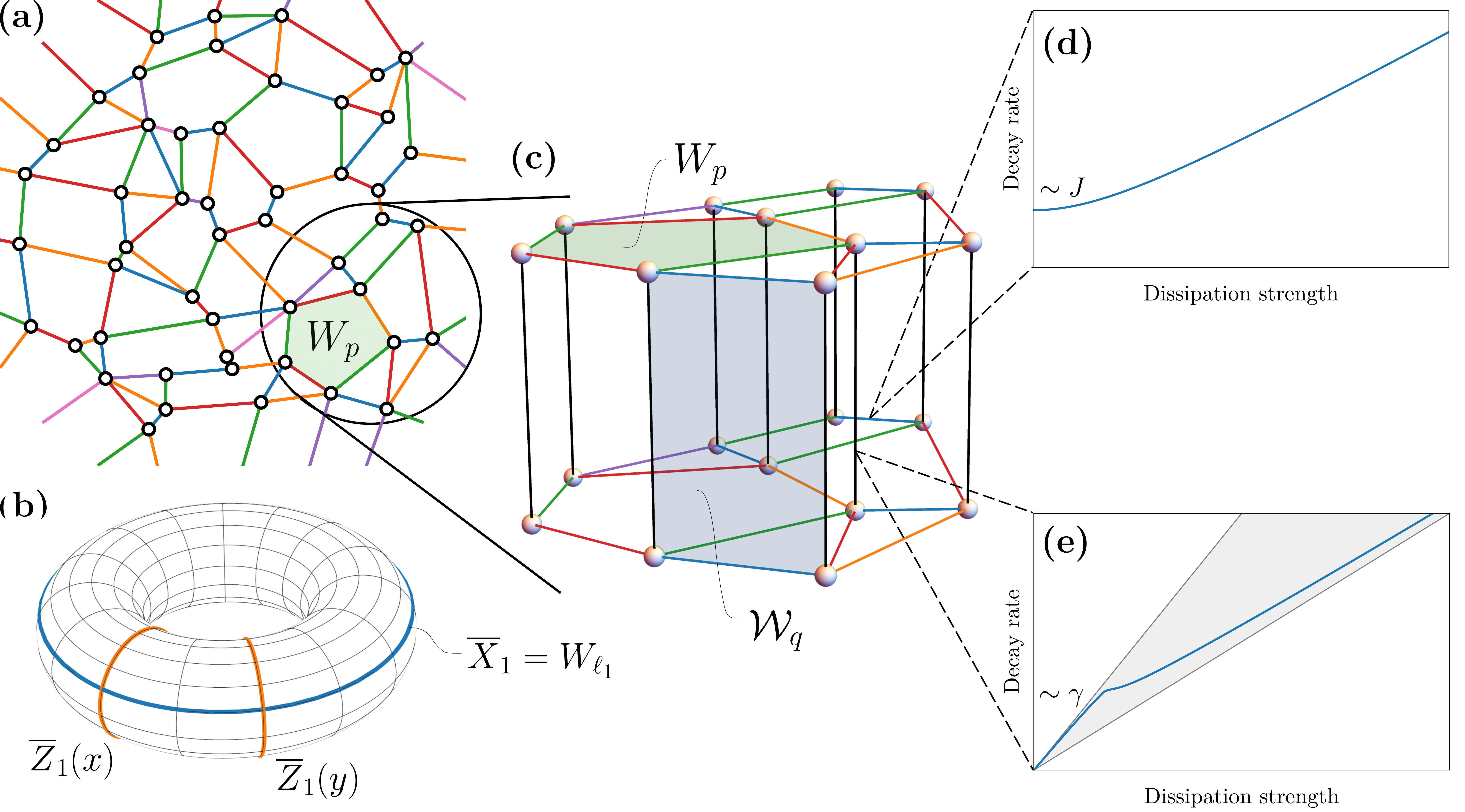}
    \caption{Schematic representation of the main results. 
    (a) A ``proper edge coloring'' of a graph (i.e., such that no two edges of the same color meet at any vertex) defines a gamma-matrix Hamiltonian with Heisenberg interactions set by the edge colors. The gamma-matrix ``spins'' and itinerant Majorana fermions (resulting from the gamma-matrix-to-Majorana mapping of Sec.~\ref{sec:free-fermion}) live on the vertices and gauge fields $u_{ij}$ (resulting from the same mapping) live on the edges. The gauge fields give rise to fluxes $W_p$ through elementary plaquettes $p$. 
    (b) When the system is realized on a topologically nontrivial surface like a torus, fluxes can thread noncontractible loops. One can define logical operators $\bX_a$ and $\bZ_a$ for each noncontractible loop. In the steady state, the product of two $\bZ$ operators acquires a nonzero correlation even when the two operators are arbitrarily far apart. Consequently, the $\bX$ operators are spontaneously broken from a strong to a weak one-form symmetry.
    (c) Vectorization: doubling of degrees of freedom to represent the open system. In the free-fermion solvable case, two copies (layers) of the graph are connected by interlayer gauge fields $v_j$ (black vertical lines) arising from dephasing.
    The symmetry sectors of the Lindbladian are labeled by $W_p$, $\bX_a$, and interlayer fluxes $\scW_q$, and each has its own fermionic vacuum. The steady states are the fermionic vacua in sectors with $\scW_q=1$ for all interlayer plaquettes $q$, i.e., without any gauge flips ($\scW_q=-1)$, and they are labeled by their flux patterns.
    The decaying states of the Lindbladian are the vacua with gauge flips and fermionic excitations.
    (d) When intralayer gauge fields are flipped, the excitations conserve only fermionic parity, which is determined by a Pfaffian method we describe. For weak dissipation, their decay rate is proportional to the Hamiltonian coupling strength $J$ when the thermodynamic limit is taken first (anomalous relaxation).
    (d) If only interlayer gauge fields $v_j$ are flipped, the excitations are labeled by a conserved fermion number, and their decay rates can be analytically bounded with Bendixson inequalities (shaded region). At weak dissipation, the decay rate of all excitations is linear in the dissipation strength $\gamma$.}
    \label{fig:MainResults}
\end{figure*}

\subsection*{Main results}

Our main results are summarized in Fig.~\ref{fig:MainResults}. We consider a Lindbladian gamma matrix model on an arbitrary graph constrained only by a \emph{color rule}, Fig.~\ref{fig:MainResults}(a).
We identify certain closed and open strings of gamma matrices as the strong and weak symmetries~\cite{buca2012} of the model, respectively.
Accordingly, the Lindbladian has an exponentially large number of steady states, labeled by the eigenvalues of closed strings (fluxes).

If the graph is embedded on a topologically nontrivial surface, closed strings can wrap around noncontractible loops, Fig.~\ref{fig:MainResults}(b). Their eigenvalues in different steady states encode \emph{classical} information~\cite{li2025PRB,ellison2025PRXQ}, which can be manipulated with logical operators constructed from closed gamma-matrix strings on the dual graph. Despite being a classical memory, the steady states display nontrivial topological order intrinsic to mixed states, which we detect by the SW-SSB~\cite{lee2023PRXQ,ogunnaike2023PRL,ma2023arXiv,ma2025PRXQ,sala2024PRB,lessa2025PRXQ,xu2025PRB,moharramipour2024PRXQ,gu2024arXiv,huang2025PRB,zhang2024arXiv,zhang2025PRB,liu2024arXiv,guo2024arXiv,weinstein2025PRL,kim2024arXiv,ando2024arXiv,chen2025PRB,sun2024arXiv,orito2025PRB,lu2024arXiv,feng2025arXiv} of the noncontractible fluxes. 

To go beyond the general symmetry considerations, we restrict our attention to the subset of jump operators for which the Lindbladian can be reduced to a bilayer non-Hermitian Hamiltonian of free Majorana fermions hopping on a background of classical $\mathbb{Z}_2$ gauge fields (both inside each layer and between the two layers), Fig.~\ref{fig:MainResults}(c).
The steady states of this bilayer operator are the vacua of interlayer (adjoint) fermions (pairings of one Majorana from each layer), distinguished by different flux patterns. 
Open gamma-matrix strings carry fractionalized operators at their endpoints, which can be of three types, corresponding to intralayer and interlayer gauge flips (leading to additional vacua that are not steady states) and Majorana fermions that populate the vacua. Each type gives parametrically distinct relaxation timescales.

In the presence of intralayer gauge flips, the dynamics conserve only the parity of adjoint fermions, Fig.~\ref{fig:MainResults}(d). We derive an analytical constraint on the parity of the excitation based on the sign of the Pfaffian of the single-particle Lindbladian. The spectral gap becomes finite at zero dissipation if the thermodynamic limit is taken first, a phenomenon dubbed anomalous relaxation~\cite{garcia2023PRD2,mori2024,yoshimura2024,yoshimura2025}.

In the absence of intralayer gauge flips, the dynamics conserve adjoint-fermion number, Fig.~\ref{fig:MainResults}(e). We show that the ensuing relaxation rates are largely independent of the microscopic details of the graph structure and provide analytic bounds. Furthermore, we prove that the decay rates always vanish in the limits of dissipation going to zero and, in some sectors, also to infinity.

\section{Dissipative gamma-matrix models on arbitrary graphs}
\label{sec:model}

\subsection{The Kitaev honeycomb model}

We start by recalling Kitaev's original honeycomb model~\cite{kitaev2006AnnPhys}. Consider the honeycomb lattice with spin-$1/2$ degrees of freedom living on the lattice sites described by Pauli matrices $\sigma^\mu_i$, $i\in\{1,\dots,N\}$ and $\mu\in\{x,y,z\}$. The Hamiltonian of the model describes a highly anisotropic Heisenberg interaction of the spins,
\begin{equation}
\label{eq:H_Kitaev}
    H_\mathrm{K}=
    -J_x\sum_{x\text{-edges}}\sigma^x_i\sigma^x_j
    -J_y\sum_{y\text{-edges}}\sigma^y_i\sigma^y_j
    -J_z\sum_{z\text{-edges}}\sigma^z_i\sigma^z_j,
\end{equation}
where the $x$-, $y$-, and $z$-edges connect sites $i$ and $j$ along the three different directions of the honeycomb lattice.

The Pauli matrices describing spin-$1/2$ degrees of freedom are the simplest (two-dimensional) Clifford algebra: the two generators $\sigma^x_i=(\sigma^x_i)^\dagger$ and $\sigma^y_i=(\sigma^y_i)^\dagger$ satisfy $\{\sigma^x_i,\sigma^y_i\}=0$ and $(\sigma^x_i)^2=(\sigma^y_i)^2=1$, and $\sigma^z_i=-\i \sigma^x_i\sigma^y_i$ is the chiral element of the algebra (i.e., an operator that anticommutes with all its elements). The crucial feature of Eq.~(\ref{eq:H_Kitaev}) is that there is a single $x$-, $y$-, and $z$-type link emanating from each site, which allowed Kitaev to reduce the Hamiltonian to that of free Majorana fermions hopping on a background gauge field~\cite{kitaev2006AnnPhys}, after fractionalization of the spins into Majorana fermions.
(Some variants of the Kitaev model can also be solved without recourse to fractionalization~\cite{gamayun2022,terentenkov2024}.)
For lattices with a higher valence (site coordination number), the same construction is still possible if one chooses the local degrees of freedom as gamma matrices in a higher-dimensional Clifford algebra. 
Hamiltonians of this type were studied in Refs.~\onlinecite{wu2009PRB,yao2009PRL,chua2011PRB}. 

\subsection{The gamma-matrix model}

We consider a model defined on an arbitrary graph $G$ with $N$ vertices $i$ and $E$ edges $\langle ij \rangle$, with $z=z(i)$ the local vertex valence (number of incident edges at vertex $i$). The local degrees of freedom are Euclidean gamma matrices $\Gamma_j^\mu=(\Gamma_j^\mu)^\dagger$, $j\in\{1,\dots,N\}$ and $\mu\in\{1,\dots,2k\}$, of dimension $2^k\times 2^k$ that satisfy $\{\Gamma_j^\mu,\Gamma_j^\nu\}=2\delta_{\mu\nu}$ and $[\Gamma^\mu_i,\Gamma^\nu_j]=0$ for $i\neq j$. Here, $2k=z$ ($2k=z+1$) if $z$ is even (odd).
Moreover, for convenience, we denote $\Gamma_j^0=\id$ (noting that $\Gamma^0_j$ is \emph{not} a member of the Clifford algebra and does \emph{not} satisfy the anticommutation relations).
The set of all possible products of gamma matrices forms a basis of operators of the local Hilbert space; of particular importance are 
the $\mu\neq\nu$ bilinears,
\begin{equation}
    \Gamma^{\mu\nu}_j=\i \Gamma^\mu_j\Gamma^\nu_j=\frac{\i}{2}[\Gamma^\mu_j,\Gamma^\nu_j],\quad \mu,\nu\in\{1,\ldots, 2k\},
\end{equation}
the chiral element,
\begin{equation}
	\label{eq:Gamma_chiral}
	\Gamma^\mathrm{ch}_j=\Gamma_j^{2k+1}:=\i^k \Gamma_j^1\cdots \Gamma_j^{2k},
\end{equation}
and the string operators defined in the next subsection.

\begin{figure}[tp]
    \centering
    \includegraphics[width=0.7\columnwidth]{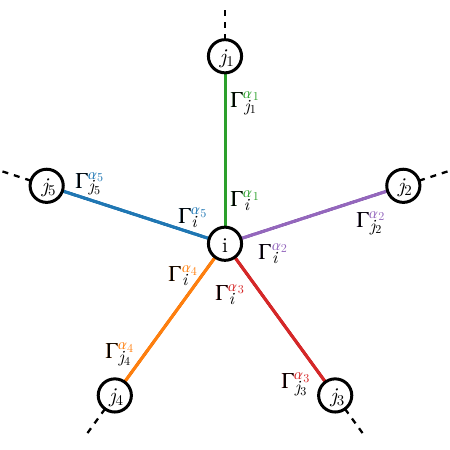}
    \caption{Spin Hamiltonians defined by the color rule.
    Each edge of the graph is assigned a color, and two gamma matrices with that color at the two endpoints of the edge interact in a term of the Hamiltonian. We illustrate the general construction for a vertex $i$ of valence $z(i)=5$.}
    \label{fig:HamiltonianIllustration}
\end{figure}

We assign a \emph{color} $\alpha=\alpha(\langle ij\rangle)\in\{1,\dots,1+\max_i z(i)\}$ to each edge, such that no two edges of the same color meet at any vertex of the graph---the \emph{color rule}, see Fig.~\ref{fig:MainResults}(a) for an example. That is, we assign a \emph{proper edge coloring} to the graph. 

The Hamiltonian of the model is then given by Heisenberg interactions of the gamma matrices determined by the colors of the edges, 
\begin{equation}
\label{eq:H_gamma}
    H=-\sum_{\langle ij\rangle} J_{ij}\Gamma^{\alpha(\langle ij\rangle)}_i\Gamma^{\alpha(\langle ij\rangle)}_j
    =:-\sum_{\langle ij\rangle} J_{ij}K_{ij},
\end{equation}
with real $J_{ij}=J_{ji}$, see Fig.~\ref{fig:HamiltonianIllustration}. Because of the color rule, each gamma matrix appears in exactly one edge operator $K_{ij}$, which, therefore, satisfy
\begin{equation}
\label{eq:def_H}
    K_{ij}K_{i'j'}=(-1)^{\delta_{ii'}}(-1)^{\delta_{jj'}}K_{i'j'}K_{ij}.
\end{equation}
This property will be the key to free-fermion reducibility.
Special cases of this Hamiltonian support different types of QSL ground states, namely, Kitaev's original model on the honeycomb lattice~\cite{kitaev2006AnnPhys} [Eq.~(\ref{eq:H_Kitaev})] and its generalizations to the square~\cite{yao2009PRL}, decorated honeycomb~\cite{yao2007PRL}, and kagome~\cite{chua2011PRB} lattices.

\subsection{Symmetries of the gamma-matrix Hamiltonian}
\label{sec:strings_fluxes}

The symmetries of the Hamiltonian are constructed from the string operators
\begin{equation}\label{eq:def_W}
\begin{split}
W_{\scP}^{\mu\nu}&=(-\i)^{|\scP|}
\Gamma_{j_1}^{\mu \alpha_{1,2}}
\Gamma_{j_2}^{\alpha_{1,2}\alpha_{2,3}} \cdots
\Gamma_{j_{|\scP|}}^{\alpha_{|\scP|-1,|\scP|} \nu}
\\ &=
\Gamma^\mu_{j_1}
K_{j_1j_2} K_{j_2j_3}\cdots K_{j_{|\scP|}-1,j_{|\scP|}}
\Gamma_{j_{|\scP|}}^\nu,
\end{split}
\end{equation}
where $\scP=\{j_1,\dots,j_{|\scP|}\}$ is a self-avoiding path on the graph of length $|\scP|$ passing through the sites $j_1,\dots,j_{|\scP|}$, $\alpha_{m,m+1}:=\alpha(\langle j_m j_{m+1}\rangle)$ matches the edge color, and $\mu,\nu\in\{0,1,\dots,2k+1\}$ such that $\mu\neq\alpha(\langle j_1,j_2\rangle)$ and $\nu\neq\alpha(\langle j_{|\scP|-1}j_{|\scP|}\rangle)$. (Recall that $\Gamma_j^0=\id$ by convention.)
From the second line in Eq.~(\ref{eq:def_W}), we see that $W_{\scP}^{\mu\nu}$ corresponds to a pair of gamma matrices, $\Gamma^\mu_{j_1}$ and $\Gamma_{j_{|\scP|}}^\nu$, at the endpoints of a string of terms of the Hamiltonian, $K_{ij}$. 

Throughout this work, some special cases of Eq.~(\ref{eq:def_W}) will prove important and we collect them here for later reference.
If the string terminates at the identity at both ends ($\mu=\nu=0$), then it corresponds to a pure string of Hamiltonian terms and we denote it by
\begin{equation}\label{eq:WP}
W_\scP=
\prod_{\langle ij\rangle\in\scP} K_{ij}.
\end{equation}
If the string terminates at a chiral element at both ends ($\mu=\nu=2k+1$), we denote it by
\begin{equation}
\label{eq:Wch}
     W^\ch_\scP=\Gamma_{j_1}^\ch K_{j_1j_2} K_{j_2j_3}\cdots K_{j_{|\scP|}-1,j_{|\scP|}} \Gamma_{j_{|\scP|}}^\ch
\end{equation}
and call it a chiral (open) string.
Strings can also be defined on a closed path $\scC$ if $j_1=j_{|\scC|}$;
in that case, the string operator, corresponding to a ``flux'' through the surface delimited by $\scC$, is given by
\begin{equation}
\label{eq:WC}
    W_\scC=\prod_{\langle ij\rangle\in\scC} K_{ij}.
\end{equation}

Not all strings of the form (\ref{eq:def_W}) commute with the Hamiltonian.
Since each $K_{ij}$ in $H$ anticommutes with precisely zero or two $K_{i'j'}$ in $W_\scC$, we find $[H,W_\scC]=0$ for all closed $\scC$. 
On the other hand, the $K_{ij}$ at the endpoint of an open string $W_\scP$ anticommutes with all $K_{i'j'}$ that share that endpoint but have a different color; hence, $[H,W_\scP]\neq0$ for all open $\scP$.
The addition of a gamma matrix $\Gamma_j^\mu$ with $\mu\in\{1,\dots,2k\}$ at the end of the string leads to commutation with all $K_{i'j'}$ except the one with color $\mu$; hence, also $[H,W_\scP^{\mu\nu}]\neq0$ for all open $\scP$ and $\mu,\nu\in\{1,\dots,2k\}$.  
Finally, open string operators commute with all $K_{ij}$, and hence the Hamiltonian, if and only if they terminate at a chiral element at both ends; hence, $[H,W_\scP^\ch]=0$.

A generic Hamiltonian of the form~(\ref{eq:def_H}) does not have any other symmetries besides $W_\scC$ and $W_\scP^\ch$. Moreover, not all $W_\scC$ and $W_\scP^\ch$ commute among themselves: $W_\scC W_{\scC'}=W_{\scC'} W_\scC$ and $W_\scC W^{\ch}_{\scP}=W^{\ch}_{\scP} W_\scC$, for all $\scC$, $\scC'$, and $\scP$, but $W^{\ch}_{\scP}W^{\ch}_{\scP'}=-W^{\ch}_{\scP'}W^{\ch}_{\scP}$ if and only if $\scP$ and $\scP'$ share an odd number of endpoints (they commute otherwise); hence, the eigenstates of $H$ are not simultaneous eigenstates of all $W_\scC$ and $W_\scP$, and only the former define the quantum numbers of $H$.

\subsection{Lindbladian dynamics}

We now couple the system to its environment. The time evolution of the system's density matrix $\rho_t$ is given by a quantum master equation, $\dot{\rho}_t=\scL(\rho_t)$, where the Liouvillian superoperator of Lindblad form~\cite{belavin1969,lindblad1976,gorini1976} (Lindbladian for short) is:
\begin{equation}
\label{eq:Lindblad}
	\mathcal{L}(\rho_t)
		=	-\i [H,\rho_t]+\sum_j \left(L_j \rho_t L_j^\dagger-\frac{1}{2}\{L_j^\dagger L_j, \rho_t\}\right).
\end{equation}

Our main focus will be on free-fermion reducible Lindbladians (see Sec.~\ref{sec:free-fermion}), which, as we shall show, requires single-site generalized ``dephasing'' jump operators (see Appendix~\ref{app:nontrivial_SS} for further details):
\begin{equation}
\label{eq:jumpops}
	L_j=\sqrt{\gamma_j}\,\Gamma^{\ch}_j,
\end{equation}
where $\gamma_j$ is the on-site dissipation strength.\footnote{If $z(j)$ is odd, $\Gamma_j^{2k}$ enters neither the Hamiltonian nor the jump operators. We can either regard it as an inert gamma matrix only necessary to complete the Hilbert space---which replicates the Lindbladian, leading to two degenerate copies that can be treated individually; or add a jump operator $L'_j=\sqrt{\gamma'_j}\Gamma^{2k}_j$---the formulas below must then be modified accordingly but the qualitative picture remains intact.} (We use the terms dissipation and dephasing interchangeably in this work.)
The jump operators~(\ref{eq:jumpops}) generalize the standard dephasing mechanism in spin-$1/2$ systems ($L_j=\sqrt{\gamma_j}\sigma^z_j$) to gamma-matrix ``spins''.
More generally, the symmetry analysis of Secs.~\ref{sec:symm_lindblad} and \ref{sec:TO} holds for any jump operator that is a single string of gamma matrices and conserves all fluxes $W_\scC$, i.e., $[L_m,W_\scC]=0$ for all $m$ and $\scC$. An example of such a jump operator is the edge operator $L_{\langle ij\rangle}\propto K_{ij}$~\cite{hwang2024Quantum}. 
This choice does not, however, preserve the free-fermion description of the model (see Appendix~\ref{app:nontrivial_SS}). Simultaneous flux conservation and free-fermion solvability reduce the choice of jump operators to Eq.~(\ref{eq:jumpops}). We will show below that, once this choice is made, the dynamics of the open system are largely independent of the choice of graph (i.e., of Hamiltonian); in contrast, for closed systems, different choices of graph lead to vastly different physics~\cite{kitaev2006AnnPhys,yao2007PRL,yao2009PRL,chua2011PRB,wu2009PRB,willans2011PRB}.

\subsection{Strong and weak symmetries}

Unitary symmetries of open quantum systems come in two flavors: strong and weak~\cite{buca2012} (sometimes also called exact and average, respectively), the properties of which we now briefly review. Let $S$ be a unitary operator in the physical Hilbert space. $S$ is a \emph{weak} symmetry of $\scL$ if $\scL(S\rho S^\dagger)=S\scL(\rho)S^\dagger$, for any operator $\rho$. Introducing the superoperator $\scS$ that acts by conjugation with $S$ (i.e., the adjoint action of $S$ on the space of operators), $\scS(\rho)=S\rho S^\dagger$, a weak symmetry is equivalent to the superoperator commutation relation $[\scL,\scS]=0$. 
On the other hand, $S$ is a \emph{strong} symmetry if $S$ commutes with all elements of the set $\{H,L_j,L_j^\dagger\}$ ($j\in\{1,\dots,N\}$). It follows that the Lindbladian superoperator commutes with both $\scS_+$ and $\scS_-$, where $\scS_+$ is the left action of $S$, $\scS_+(\rho)=S\rho$, and $\scS_-$ is the right action, $\scS_-(\rho)=\rho S^\dagger$. Clearly, a strong symmetry is also weak, but the converse is not true. A simple example of a weak symmetry that is not strong (and that occurs in our model) is if $S$ commutes with $H$ but anticommutes with (some of) the jump operators.
	
Let the distinct eigenvalues of $S$ be denoted as $s_\alpha=e^{\i \theta_\alpha}$ ($\alpha\in\{1,\dots,n\}$). The distinct eigenvalues of the weak symmetry superoperator $\scS$ are hence $\mathfrak{s}_{\alpha\beta}=s_\alpha s_\beta^*=e^{\i (\theta_\alpha-\theta_\beta)}$. Since the Lindbladian and the symmetry operator commute, they can be simultaneously diagonalized, with sectors labeled by the eigenvalues of the symmetry. As such, weak symmetry sectors are labeled by $\mathfrak{s}_{\alpha\beta}$ and strong symmetry sectors by a pair $(s_\alpha,s_\beta)$, where $s_\alpha$ ($s_\beta^*$) is the eigenvalue of $\scS_+$ ($\scS_-$). Traceful eigenoperators, namely, all steady states, belong to diagonal symmetry sectors~\cite{buca2012}, i.e., the $n$ symmetry sectors with $s_\alpha=s_\beta$ and $\mathfrak{s}_{\alpha\beta}=1$. Generically, each such sector contains one steady state (with fine-tuning, there could be multiple steady states per sector, but this does not occur here).

\subsection{Symmetries of the gamma-matrix Lindbladian}
\label{sec:symm_lindblad}

The strong and weak symmetries of the Lindbladian~(\ref{eq:Lindblad}) are given by the string operators defined in Eqs.~(\ref{eq:Wch}) and (\ref{eq:WC}).
\begin{itemize}
    \item \textit{Strong symmetries.} Since any $W_\scC$ commutes both with the Hamiltonian and all the jump operators, it is a strong symmetry of $\mathcal{L}$. 
    By a slight abuse of notation, we denote the left-action superoperator as $\scW_{\scC+}=W_\scC$ and the right-action superoperator as $\scW_{\scC-}=\Wt_\scC$.
    \item \textit{Weak symmetries.} Open chiral strings $W^\ch_{\scP}$ commute with the Hamiltonian but \emph{anticommute} with the jump operators located at the ends of the string. Consequently, $W^\ch_\scP$ are weak Liouvillian symmetries. Their superoperator action is the adjoint action of the chiral string operator,
    \begin{equation}
    \label{eq:scWch}
        \scW_\scP^\ch(\rho)=W^\ch_\scP \rho W^\ch_\scP.
    \end{equation}
\end{itemize} 

Not all $W_\scC$, $\Wt_\scC$, and $\scW_\scP^\ch$ are independent. As worked out in detail in Appendix~\ref{app:fluxes}, a minimal set---which labels the symmetry sectors---is $\{W_p,W_\ell,\scW_q\}_{p,\ell,q}$. Here, $W_p$ are $W_\scC$ for $\scC=p$, with $p$ the elementary plaquettes (or faces) of the graph. If the graph is embedded on a closed surface, the total flux through that surface is $\prod_p W_p=+1$ identically; this reduces the number of independent plaquettes by one. The flux through any other contractible loop is then the product of the fluxes through the plaquettes contained in the loop. Additionally, the surface on which we embed the graph may have noncontractible loops $\ell$ and there is an additional flux $W_{\ell}$ along each such loop (of which there are $b_1$---the first Betti number of the surface). There is also a weak symmetry $\scW_q:=\scW^{\ch}_{\langle ij\rangle}$ along each of the edges of the graph; in a doubled (i.e., vectorized) representation that we discuss in Sec.~\ref{sec:vectorization}, the weak symmetry corresponds to the flux through an \emph{interlayer plaquette} $q$ [see also Fig.~\ref{fig:MainResults}(c)]. Any string $\scW_\scP^\ch$ is then given by an appropriate product of elementary fluxes $\scW_q$.
As detailed in Appendix~\ref{app:fluxes}, the fluxes $\Wt_p$ and $\Wt_\ell$ are not independent as they are constrained by $W_p$, $W_\ell$ and $\scW_q$. In total, we have $2E-N+1$ independent conserved fluxes: $E-N+1$ strong fluxes $W_p$ and $W_\ell$ and $E$ weak fluxes $\scW_q$.
The fluxes have eigenvalues $\pm1$ ($\pm\i$) if the length of their boundary is even (odd), as follows from the Hermiticity of the gamma matrices and the prefactor $\i^{|\scP|}$ in Eq.~(\ref{eq:def_W})---this choice makes the antiunitary symmetries of the model manifest~\cite{kitaev2006AnnPhys} (see also Sec.~\ref{sec:excitations_number}). 
$\scW_q$ always has an even number of edges and thus always squares to $+1$.
We denote the eigenvalues of $W_p$, $W_\ell$, $\Wt_p$, $\Wt_\ell$, and $\scW_q$ by $w_p$, $w_\ell$, $\wt_p$, $\wt_\ell$, and $\scw_q$, respectively. We further denote a strong-flux configuration (i.e., an assignment $w_p=\pm1$ or $\pm\i$ and $w_\ell=\pm1$ or $\pm\i$ for all $p$ and $\ell$) by $\bm{w}$ and a weak-flux configuration (i.e., an assignment $\scw_q=\pm1$ for all $q$) by $\bm{\scw}$; there are $2^{E-N+1}$ strong-flux configurations and $2^E$ weak-flux configurations.

\section{Mixed-state topological order}
\label{sec:TO}

\begin{figure*}
    \centering
    \includegraphics[width=0.7\textwidth]{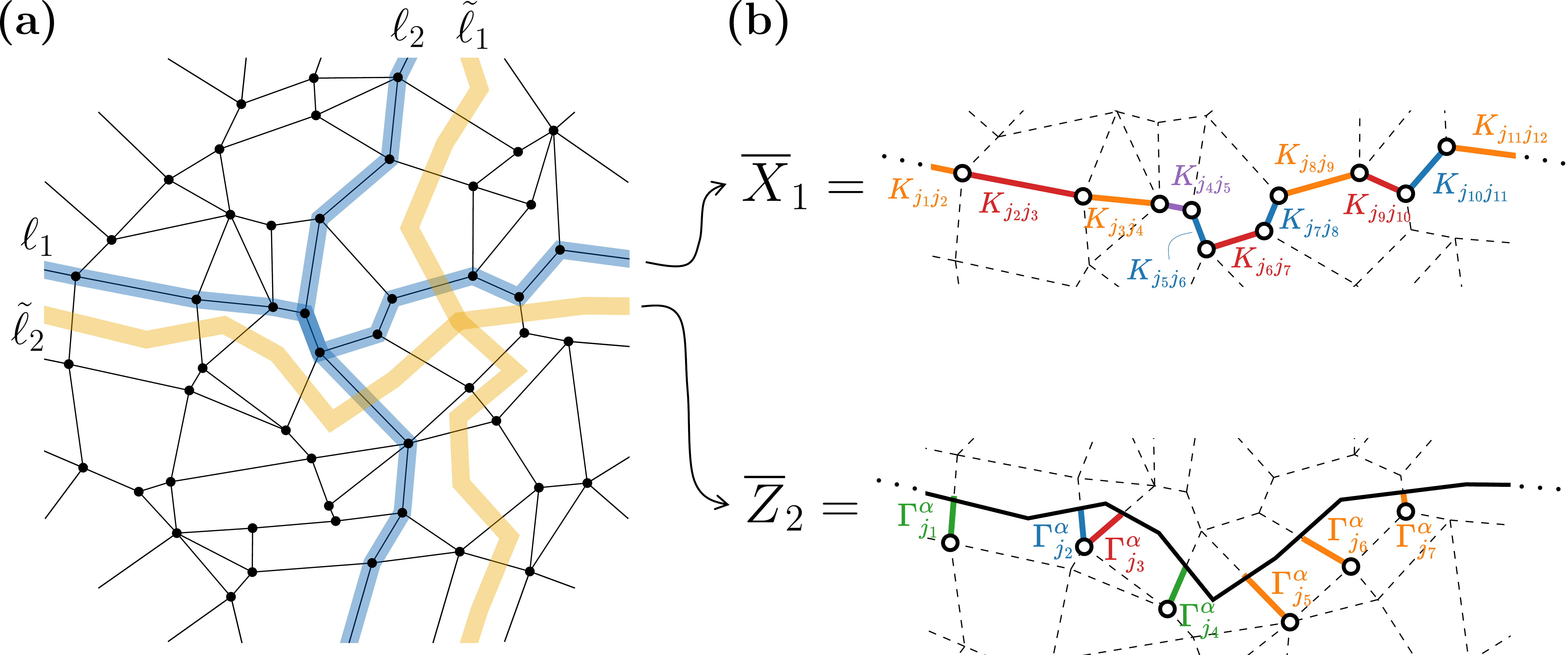}
    \caption{Logical operators of the Lindbladian gamma-matrix model. (a) Two noncontractible loops $\ell_1$ and $\ell_2$ for a random graph on the torus and two noncontractible loops $\lt_1$ and $\lt_2$ on the dual graph. $\ell_a$ and $\lt_b$ intersect and odd (even) number of times if $a=b$ ($a\neq b$). The same construction holds for higher-genus surfaces and any graph. (b) Construction of the logical operations in terms of gamma matrices.}
    \label{fig:logicals}
\end{figure*}

\subsection{Steady states}
\label{sec:steady}

The steady states $\rho_\mathrm{ss}$ of $\mathcal{L}$ satisfy $\dot\rho_\mathrm{ss}=\mathcal{L}(\rho_\mathrm{ss})=0$, i.e., they are eigenoperators of $\mathcal{L}$ with zero eigenvalue. As discussed above, there is one such steady state for each eigenvalue of a strong symmetry, i.e., for each strong-flux configuration. Hence, we have $2^{E-N+1}$ distinct steady states, which we can build explicitly. 

Because jump operators are Hermitian, the identity matrix (the fully-mixed state) is annihilated by $\scL$ and, because of the strong symmetry property, so is any product of strong fluxes. However, the identity and such products are not eigenstates under left multiplication by other flux combinations; therefore, they do not correspond to physical steady states. Instead, the latter are given by the infinite-temperature state projected into the sectors with fixed eigenvalue of the strong fluxes:
\begin{equation}
\label{eq:ss}
    \rho_{\bm{w}}
    =\frac{1}{\mathcal{Z}}P_{\bm{w}}.
\end{equation}
Here,
\begin{equation}
\label{eq:def_projector}
    P_{\bm{w}}=\prod_{\ell}\frac{\id+w_\ell W_\ell}{2}
    \prod_p \frac{\id+w_{p}W_{p}}{2}
\end{equation}
is the projector into the symmetry sector labeled by the strong-flux configuration $\bm{w}$, and $\mathcal{Z}=2^{N-1}$ is the dimension of the symmetry sectors.\footnote{This can be shown as follows. The total Hilbert space dimension is $\prod_i 2^{z(i)/2}=2^{\sum_i z(i)/2}=2^E$, where in the last step we used the handshaking lemma. Assuming all the sectors have the same dimension, we divide the total dimension by the number $2^{E-N+1}$ of distinct steady states to find $\scZ=2^{N-1}$. We justify our assumption once we introduce Majorana fermions in Sec.~\ref{sec:free-fermion}.
}
Using Kitaev-model terminology~\cite{kitaev2006AnnPhys}, we say there is a vortex at plaquette $p$ if $w_p=-1$ and, therefore, a steady state is labeled by its vortex pattern.
Because of the commutativity of all fluxes, we have $\scW_q(\rho_{\bm{w}})=\rho_{\bm{w}}$, as required from a steady state under a weak symmetry.

In closed quantum systems, a hallmark of topological order is topological ground-state degeneracy, with distinct ground states labeled by eigenvalues of noncontractible string operators~\cite{wen1990,wen1995,wen2003PRL,kitaev2001Anyons,levin2005PRB,nayak2008RMP,wen2017RMP,simon2023book}. This degeneracy signals a generalized form of symmetry breaking, namely that of one-form symmetries furnished by these noncontractible strings~\cite{hastings2005PRB,nussinov2009,gaiotto2015JHEP,mcgreevy2023Annu}.
In our gamma-matrix model, $W_\ell$ are such operators and their eigenvalues $w_\ell$ label steady states; these are thus topologically degenerate in this sense. 
Hence, a natural question is to what extent $\rho_{\bm{w}}$ exemplifies mixed-state topological order~\cite{wang2024SciPost,wang2025PRXQ,chen2024PRL,sang2024PRX,dai2025PRB,li2025PRB,ma2025PRXQ,sohal2025PRXQ,sang2024PRL,ellison2025PRXQ,li2024arXiv,sala2024PRB,lessa2025PRXQ,wang2025PRXQanalogEE,moharramipour2024PRXQ,gu2024arXiv,zhang2025PRB,lessa2025arXiv}. We show that $\rho_{\bm{w}}$ displays topological order in the sense of exhibiting SW-SSB~\cite{lee2023PRXQ,ogunnaike2023PRL,ma2023arXiv,ma2025PRXQ,sala2024PRB,lessa2025PRXQ,xu2025PRB,moharramipour2024PRXQ,gu2024arXiv,huang2025PRB,zhang2024arXiv,zhang2025PRB,liu2024arXiv,guo2024arXiv,weinstein2025PRL,kim2024arXiv,ando2024arXiv,chen2025PRB,sun2024arXiv,orito2025PRB,lu2024arXiv,feng2025arXiv} of the one-form symmetries~\cite{zhang2025PRB} $W_\ell$.
These states can, however, encode only classical information (akin to the decohered toric code~\cite{lee2023PRXQ,sang2024PRX,bao2023arXiv,fan2024PRXQ,ellison2025PRXQ,zhang2025PRB,lessa2025arXiv,li2025PRB}), as one would naturally expect from their almost fully-mixed nature.

\subsection{Symmetries and order parameters}

We view $\rho_{\bm{w}}$ as logical states, and to perform logical operations, we define (see also Fig.~\ref{fig:logicals}):
\begin{align}
    &\bX_a=W_{\ell_a}=\prod_{\langle ij\rangle\in\ell_a}K_{ij},
    \\
    &\bZ_a=\prod_{\langle ij\rangle\in\lt_a}\Gamma_i^{\alpha},
\end{align}
where $a\in\{1,\dots,b_1\}$ and $\lt_a$ is a path on the dual graph
such that $\ell_a$ and $\lt_b$ intersect an odd (even) number of times if $a=b$ ($a\neq b$). The index $i$ of $\Gamma_i^\alpha$ in $\bZ_a$ can be chosen as either endpoint of each edge of the original graph crossed by $\lt_a$, and $\alpha=\alpha(\langle ij\rangle)$ is the color set by the Hamiltonian. These logical operators satisfy
\begin{equation}
\label{eq:logical_algebra}
    [\bX_a,\bX_b]=[\bZ_a,\bZ_b]=0
    \quad \text{and} \quad
    \bX_a \bZ_b = (-1)^{\delta_{ab}}\bZ_b \bX_a.
\end{equation}
The logical operators can be deformed without changing the commutation relations (\ref{eq:logical_algebra}) by multiplying them with products of $K_{ij}$ and $\Gamma_j^\ch$. In particular, the two choices of index $i$ for $\Gamma_i^\alpha$ above differ by $K_{ij}$ on the edge $\langle ij \rangle$ crossed by $\lt_a$. The path of $\bX_a$ is deformable by the multiplication by a $W_p$ (a direct-graph plaquette) adjacent to $\ell_a$; similarly, the path of $\bZ_a$ on the dual graph is deformable by the multiplication by $\Gamma_j^\ch$ (a dual-graph plaquette) on a site $j$ adjacent to $\lt_a$.
The logical operators act on $\rho_{\bm{w}}$ as
\begin{align}
    \label{eq:logical_Z}
    &\bX_a \rho_{\bm{w}} = \rho_{\bm{w}}\bX_a = w_{\ell_a}\rho_{\bm{w}},
    \\
    \label{eq:logical_X}
    &\bZ_a \rho_{\bm{w}} = \rho_{\bm{w'}}\bZ_a,
\end{align}
with $\bm{w'}$ such that $w'_{p}=w_p$ for all $p$ and $w'_{\ell_b}=(-1)^{\delta_{ab}}w_{\ell_b}$.

\subsection{One-form strong-to-weak spontaneous symmetry breaking}

Above, we discussed the strong and weak symmetries of the Lindbladian. They can also be defined directly for mixed states. A state $\rho$ has a strong symmetry if $\scS_+(\rho)=S\rho=e^{\i\theta}\rho$ for some $\theta$, while it has a weak symmetry if instead only $\scS(\rho)=S \rho S^\dagger =\rho$, where, as before, $S$ is a unitary symmetry operator. Let us consider a local operator $O$ that is charged under $S$ such that $S^\dagger O S = e^{\i q_O}O$, $q_O\notin 2\pi\mathbb{Z}$. Then, it follows that if $\rho$ has a weak symmetry, $\Tr[\rho O]=0$; if $\rho$ has a strong symmetry, besides the previous condition, also $\Tr[\rho O \rho O^\dagger]=0$. 

If upon substituting $O\to O(x) \bO(y)$, where $O(x)$ and $\bO(y)$ are operators with charge $q_O$ and $-q_O$ and localized near $x$ and $y$, respectively, the correlators acquire a finite value when $|x-y|\to\infty$, then the respective symmetry is spontaneously broken by $\rho$~\cite{lee2023PRXQ,ma2023arXiv,sala2024PRB,lessa2025PRXQ,zhang2025PRB,liu2024arXiv,weinstein2025PRL,gu2024arXiv}. More concretely, if the standard correlator (which can be viewed as a R\'enyi-1 correlator)
\begin{equation}
    \Tr[\rho O(x)\bO(y)]\sim O(1),
\end{equation}
then the symmetry (either strong or weak) is completely broken. If, instead, the R\'enyi-1 correlator vanishes but the R\'enyi-2 correlator 
\begin{equation}
    \frac{\Tr[\rho O(x)\bO(y) \rho \bO(y)^\dagger O(x)^\dagger]}{\Tr[\rho^2]}\sim O(1),
\end{equation}
then the strong symmetry is spontaneously broken to a weak one; there is SW-SSB. 
(This is because $\scO(\rho)=\bO\rho \bO^\dagger$ is charged under the strong symmetry $\scS_+$, $\scO\scS_+ =e^{\i q_O}\scS_+\scO$, but neutral under the weak symmetry $\scS$, $ \scO\scS = \scS\scO$.)
If both correlators vanish, the symmetry is unbroken.

We now turn to our gamma-matrix model. The logical operator $\bZ_a$ in Eq.~(\ref{eq:logical_Z}) is charged under the noncontractible flux $W_{\ell_a}$: $W_{\ell_a} \bZ_a W_{\ell_a} = -\bZ_a$. In any steady state $\rho_{\bm{w}}$, the R\'enyi correlators are given by
\begin{align}
    \label{eq:renyi-1_gamma}
    &\Tr[\rho_{\bm{w}} \bZ_a(x)\bZ_a(y)]=0,
    \\
    \label{eq:renyi-2_gamma}
    &\frac{\Tr[\rho_{\bm{w}} \bZ_a(x)\bZ_a(y) \rho_{\bm{w}} \bZ_a(y)\bZ_a(x)]}{\Tr[\rho_{\bm{w}}^2]}=1,
\end{align}
for any $x\neq y$, where $\bZ(x)$ is defined on a path that can be continuously deformed to $\lt_a$ and passes through site $x$. 
To arrive at Eq.~(\ref{eq:renyi-1_gamma}) we used that $\bZ_a(x)$ and $\bZ_a(y)$ differ by a string of $K_{ij}$ and $\Gamma_j^\ch$, and the product of a steady state and any string $R$ of gamma matrices that is algebraically independent of $W_p$ and $W_\ell$ (hence effectively acts as a tensor product $\rho_{\bm{w}}\otimes R$) is traceless.
In Eq.~(\ref{eq:renyi-2_gamma}), the string of gamma matrices arising from the left action commutes with $\rho_{\bm{w}}$ and cancels the string resulting from the right action; the numerator thus equals $\Tr[\rho_{\bm{w}}^2]$ and the result follows.
Because the charged operators are extended one-dimensional objects, the noncontractible fluxes are one-form symmetries~\cite{hastings2005PRB,nussinov2009,gaiotto2015JHEP,mcgreevy2023Annu} and the steady states exemplify one-form SW-SSB~\cite{zhang2025PRB}. 

Finally, we note that an alternative, nonequivalent definition of SW-SSB exists~\cite{lessa2025PRXQ}, in which the R\'enyi-2 correlator is replaced by the fidelity correlator $\Tr\sqrt{\sqrt{\rho}O(x)\bO(y) \rho \bO(y)^\dagger O(x)^\dagger\sqrt{\rho}}$. In our case, at least in the steady state, the two correlators coincide.

\subsection{Classical memory}

While $\rho_{\bm{w}}$ is a logical state, it is only a classical memory, encoding $b_1$ classical bits. Indeed, an initial state with finite overlaps $\pi_{\bm{w}}$ with the steady states $\rho_{\bm{w}}$ (where $\sum_{\bm{w}}\pi_{\bm{w}}=1$), evolves, at infinitely-long times, into the statistical mixture (i.e., incoherent superposition) $\sum_{\bm{w}}\pi_{\bm{w}}\rho_{\bm{w}}$. This idea can be formalized~\cite{li2025PRB,ellison2025PRXQ} by identifying $\rho_{\bm{w}}$ as \emph{isolated} extremal points in the convex set of steady states.

This classical memory is inherently robust against local perturbations, by an argument similar to the ``cleaning lemma''~\cite{bravyi2009NJP}. Any perturbation to the Hamiltonian or any jump operator that does not commute with the strong fluxes leads to an infinite-temperature steady state in the subspace labeled by those fluxes. However, if the perturbation does not commute with a noncontractible flux but is local, we can deform the noncontractible loop around the perturbation and use this deformed logical operator instead. The robustness holds as long as the perturbation's support does not wind around any noncontractible loop (if it does, then it removes the corresponding encoded bits from the memory). 

Here we addressed topological order and the robustness of the classical memory only at times sufficiently long such that the state can be smoothly connected to the steady state. Studying these properties at finite times remains an interesting open question for future work.

\section{Fractionalization and free-fermion representation}
\label{sec:free-fermion}

\subsection{Majorana representation}
\label{sec:Majoranamap}

The Lindbladian~(\ref{eq:Lindblad}) becomes exactly solvable (i.e., reducible to exponentially many single-particle Lindbladians) if we fractionalize the degrees of freedom. To this end, we introduce $2k+2$ flavours of Majorana fermions per site, $b^\mu_j$ ($\mu\in\{1,\dots,2k\}$), $b_j^\ch=b_j^{2k+1}$, and $c_j$, and write the gamma matrices as
\begin{equation}
\label{eq:frac_gamma}
    \Gamma^\mu_j=\i b^\mu_j c_j.
\end{equation} 
The Majoranas are Hermitian and satisfy the Clifford algebra $\{b^\mu_i,b^\nu_j\}=2\delta_{ij}\delta_{\mu\nu}$, $\{c_i,c_j\}=2\delta_{ij}$, and $\{b_i^\mu,c_j\}=0$. The inclusion of two additional Majoranas (i.e., an ordinary fermion) per site doubles the Hilbert space dimension, but because the chiral element is defined from the other gamma matrices through Eq.~(\ref{eq:Gamma_chiral}), we have the constraint 
\begin{equation}
\label{eq:D_projection}
	D_j:=\i^{k+1} b_j^1\cdots b_j^{2k+1}c_j=1,
\end{equation}
i.e., only states in the subspace of even (local) Majorana parity $D_j=1$ on each site $j$ are mapped to valid states in the gamma-matrix Hilbert space.
For any physical state $\rho$, $D_j$ thus acts as a strong symmetry:
\begin{equation}\label{eq:physical_states}
    D_j\rho =\rho D_j=\rho.
\end{equation}

\subsection{Vectorization}
\label{sec:vectorization}

To proceed, we represent the Lindbladian superoperator as an operator on the space of density matrices, a procedure known as vectorization~\cite{prosen2008}. The vectorized density matrix is given by $|\rho\rangle=\sum_{mn}\langle m|\rho|n\rangle |m\rangle|n\rangle$, where $\{|n\rangle\}$ forms a basis of the original Hilbert space and $\{|m\rangle|n\rangle\}$ is a basis of the doubled (vectorized) Hilbert space. The identity matrix is mapped to the (unnormalized) infinite-temperature thermofield double state $|0\rangle:=\sum_m |m\rangle |
m\rangle$. For each Majorana fermion $c_j$ ($b^\mu_j$) in the original Hilbert space, we introduce two Majoranas in the doubled Hilbert space, $c_j$ and $\ct_j$ ($b^\mu_j$ and $\bt^\mu_j$); the former represent left action on the density matrix, while the latter represent (a phase times) right action. When acting on the identity matrix, left and right multiplication should coincide, leading to $c_j|0\rangle=e^{\i\theta}\ct_j|0\rangle$. The phase $e^{\i\theta}$ is fixed by demanding anticommutation of Majoranas from different copies: acting on both sides of the equality with $c_j$ we have $|0\rangle=-e^{\i\theta}\ct_j c_j|0\rangle=-e^{2\i\theta}|0\rangle$, whence $e^{\i\theta}=\pm \i$. The choice of sign is arbitrary and we fix
\begin{equation}
\label{eq:c_fiducial}
    c_j|0\rangle =-\i \ct_j|0\rangle
    \qquad \text{and} \qquad
    b^\mu_j|0\rangle =-\i \bt^\mu_j|0\rangle.
\end{equation}

Using Eq.~(\ref{eq:c_fiducial}), one can then vectorize the left and right action of $b^\mu_j$ ($c_j$) on an arbitrary operator as follows~\cite{prosen2008}. For $M$ a fermion monomial with $DM=(-1)^{\eta_M}MD$, $D=\prod_jD_j$, we map $c_j M\to c_j|M\rangle$ and $Mc_j\to(-1)^{\eta_M}(-\i\ct_j)|M\rangle$ ($b^\mu_j M\to b^\mu_j|M\rangle$ and $Mb^\mu_j\to(-1)^{\eta_M}(-\i\bt^\mu_j)|M\rangle$).\footnote{The inclusion of the fermionic parity $(-1)^{\eta_M}$ ensures that the product $c_i M c_j$ maps to the same expression regardless of whether the left or right action is vectorized first: $(c_iM)c_j=(-1)^{\eta_{c_iM}}(-\i\ct_j)|c_iM\rangle=(-1)^{\eta_{c_iM}}(-\i\ct_j)c_i|M\rangle=(-1)^{\eta_{M}}c_i(-\i\ct_j)|M\rangle=c_i (Mc_j)$, where in the second equality we used $|c_jM\rangle=c_j|M\rangle$ and in the third the anticommutation of $c_i$ and $\ct_j$ and $(-1)^{\eta_{c_iM}}=-(-1)^{\eta_M}$.}
A gamma matrix acting from the right gets mapped to $\Gammat^\mu_j=-\i \bt^\mu_j \ct_j$ [note the relative minus sign compared with Eq.~(\ref{eq:frac_gamma})].
With this mapping, the Lindbladian reads
\begin{equation}
\label{eq:Lind_quad_Majorana}
	\mathcal{L}=
	\sum_{\langle ij\rangle}J_{ij}\left(
	u_{ij}c_ic_j+\ut_{ij}\ct_i\ct_j
	\right)
	+\sum_j \gamma_j \left(v_j\i c_j\ct_j-\id\right),
\end{equation}
where the edge operators (referred to as \emph{gauge fields}) $u_{ij}$, $\ut_{ij}$, and $v_j$ are defined as
\begin{equation}
\label{eq:edge_operators_fermionic}
    u_{ij} = \i b^\alpha_i b_j^\alpha, \qquad
    \ut_{ij} = -\i \bt^\alpha_i \bt^\alpha_j, \qquad
    v_j = \i b^{\ch}_j \bt^{\ch}_j,
\end{equation}
with $\alpha=\alpha(\langle ij\rangle)$.
It proves convenient to regard the Lindbladian as a bilayer system with two identical layers with Hamiltonian couplings of opposite sign and imaginary interlayer couplings, see Fig.~\ref{fig:MainResults}(c). The Majorana fermions $c_j$ and $\ct_j$ live on the vertices of the upper and lower layers, respectively, the gauge fields $u_{ij}$ and $\ut_{ij}$ live on the edges of the upper and lower layers, and the gauge fields $v_j$ live on the interlayer edges.

All other operators can also be rewritten in terms of Majorana fermions and gauge fields. In this representation, the strong symmetry (\ref{eq:WC}) reads
\begin{equation}
\label{eq:W_u}
    W_\scC=(-\i)^{|\scC|}\prod^\circlearrowleft_{\langle ij \rangle \in \scC} u_{ij},
    \quad 
    \Wt_\scC=\i^{|\scC|}\prod^\circlearrowleft_{\langle ij \rangle \in \scC} \ut_{ij},
\end{equation}
for its left and right actions, respectively. The weak symmetry, i.e., conjugation with $W_q^\ch$ [Eq.~(\ref{eq:scWch})], becomes
\begin{equation}
\label{eq:interlayerflux_gaugefields}
    \scW_q =-v_i u_{ij} v_j \ut_{ji}=u_{ij}\ut_{ij}v_iv_j.
\end{equation}
This is a \emph{closed} string around an interlayer plaquette. 
Finally, the right action of the constraint Eq.~(\ref{eq:D_projection}) maps to
\begin{equation}
\label{eq:Dtilde_projection}
	\Dt_j:=(-\i)^{k+1} \bt_j^1\cdots \bt_j^{2k+1}\ct_j=1.
\end{equation}

It proves convenient to define the adjoint (interlayer) complex fermions~\cite{prosen2008}
\begin{equation}
	\label{eq:def_fFermions}
	f_j=\frac12(c_j+\i\ct_j),
	\qquad
	f_j^\dagger=\frac12(c_j-\i\ct_j),
\end{equation}
which satisfy the canonical anticommutation relation $\{f_i,f_j^\dagger\}=\delta_{ij}$. Because of Eq.~(\ref{eq:c_fiducial}), the fiducial state $|0\rangle$ is the adjoint fermionic vacuum, satisfying
\begin{equation}
\label{eq:fuv0}
   f_j|0\rangle=0,
   \quad
   u_{ij}|0\rangle=\ut_{ij}|0\rangle,
   \quad \text{and} \quad  
   v_i|0\rangle=-|0\rangle.
\end{equation}
Inverting Eq.~(\ref{eq:def_fFermions}) and inserting into Eq.~(\ref{eq:Lind_quad_Majorana}), we obtain
\begin{equation}
\label{eq:Lind_quad_Fermion}
\begin{split}
    &\scL=\sum_{\langle ij\rangle}J_{ij}\left[
    (u_{ij}+\ut_{ij})(f_i^\dagger f_j+f_if_j^\dagger)
    +(u_{ij}-\ut_{ij})\right.
    \\
    &\left.\times(f_i^\dagger f_j^\dagger+f_if_j)
    \right]+2
    \sum_j \gamma_j v_j f_j^\dagger f_j
    -\sum_j\gamma_j(1+v_j).
\end{split}
\end{equation}

\subsection{Gauge invariance}

Because of the color rule, the edge operators (\ref{eq:edge_operators_fermionic}) all commute with each other and with the Lindbladian. Moreover, they satisfy $u^2_{ij}=\ut^2_{ij}=v_j^2=1$. Hence, they can be viewed as classical $\mathbb{Z}_2$ gauge fields, $u_{ij},\ut_{ij},v_i\in\{-1,+1\}$, and we can replace the operators by these values in the Lindbladian.
Since for each $\pm1$ value of $u_{ij}$, $\ut_{ij}$, and $v_i$, Eq.~(\ref{eq:Lind_quad_Majorana}) [or Eq.~(\ref{eq:Lind_quad_Fermion})] is quadratic in fermionic operators, the model is now exactly solvable. The Lindbladian block-diagonalizes into sectors where the gauge fields act as multiples of the identity:
\begin{equation}
	\label{eq:Liouv_blocks_uv}
	\scL = \bigoplus_{\{u_{ij},\ut_{ij},v_j\}}\scL_{u_{ij},\ut_{ij},v_j}.
\end{equation}

Although all gauge fields commute with the Lindbladian, they do not commute with $D_j$ and $\Dt_j$ [Eqs.~\eqref{eq:D_projection} and \eqref{eq:Dtilde_projection}] and cannot label physical states; the physical subspace on which $\mathcal{L}$ acts [defined by Eq.~(\ref{eq:physical_states})] is a superposition of these sectors.
Conjugating by $D_j$, any gauge field terminating at site $j$ changes its sign:
\begin{equation}
\label{eq:gauge_transf_fields}
    D_ju_{ij} =-u_{ij}D_j,
    \quad
    D_jv_{j} =-v_{j}D_j
\end{equation}
(and similarly for the second copy); $D_j$ thus effects a $\mathbb{Z}_2$ gauge transformation. Furthermore,
\begin{equation}
\label{eq:gauge_transf_fermions}
    D_j c_j = - c_j D_j
\end{equation}
(and similarly for the second copy), thus $D_j$ also effects the $\mathbb{Z}_2$ gauge transformation of ``matter fields''.

The total parity superoperator $\scD=\prod_j D_j\Dt_j$ is a gauge transformation on every site, hence leaving all gauge fields invariant, and thus providing no gauge-field constraint. However, it constrains the matter fields. Writing 
\begin{equation}
    \scD=(-1)^{\scN_f}\prod_{\langle ij\rangle}u_{ij}\ut_{ij}\prod_j(-v_j),
\end{equation}
where the adjoint fermionic parity is defined as
\begin{equation}
\label{eq:parity_f}
    (-1)^{\scN_f}=\prod_j(-\i c_j \ct_j), 
    \quad
    \scN_f=\sum_j f^\dagger_j f_j,
\end{equation}
shows that the physical constraint $\scD=+1$ can be seen as fixing the adjoint fermion parity to be the same as the ``gauge parity''~\cite{shackleton2024}:\footnote{If $z(j)$ is odd, there are two additional Majoranas $b_j^{2k}$ and $\bt_j^{2k}$ that appear neither in the Lindbladian nor any gauge fields $u_{ij}$, $\ut_{ij}$, or $v_j$, but still enter the definition of $D_j$ and $\Dt_j$ and thus affect the constraint $\scD=1$. Consequently, there is an additional factor in Eq.~\eqref{eq:parity_constraint}: $\prod'_{j} (\i v^{2k}_j)$, where $v^{2k}_j=\i b^{2k}_j\bt^{2k}_j$ (with $(v_j^{2k})^2=1$) and the product $\prod_j'$ is over sites with odd $z(j)$. By the handshaking lemma, the number $2M$ of such sites is always even and hence $\prod'_{j} (\i v^{2k}_j)=(-1)^M \prod'_{j} v^{2k}_j$. Since $v^{2k}_j$ does not enter the Lindbladian but commutes with it and with all gauge fields, we can simply choose an arbitrary gauge $v_j^{2k}=\pm1$ such that $\prod'_{j} v^{2k}_j=(-1)^M$ without changing the physics; Eq.~\eqref{eq:parity_constraint} then holds for odd $z(j)$ as well.
}
\begin{equation}
\label{eq:parity_constraint}
    (-1)^{\scN_f}=\prod_{\langle ij\rangle}u_{ij}\ut_{ij}\prod_j(-v_j).
\end{equation}

Since $D_j$ and $\Dt_j$ are gauge transformations, the physical subspace [defined by Eq.~(\ref{eq:physical_states})] has gauge-invariant states. 
By Eqs.~(\ref{eq:gauge_transf_fields}) and (\ref{eq:gauge_transf_fermions}), the Lindbladian is also gauge invariant, as is its spectrum (because it is preserved by similarity transformations, such as gauge transformations). However, the eigenstates $|\varphi\rangle$ obtained from Eq.~\eqref{eq:Lind_quad_Majorana} [or (\ref{eq:Lind_quad_Fermion})] are not, since we must choose \emph{some} gauge to compute them (we fix this gauge in Sec.~\ref{sec:gauge_fixing}). To obtain physical eigenstates, we project $|\varphi\rangle$ into the physical subspace using the projector~\cite{yao2009PRL,shackleton2024}
\begin{equation}
    \mathcal{Q}=\prod_j \frac{\id+D_j}{2}\frac{\id+\Dt_j}{2}.
\end{equation}
The physical state $\mathcal{Q}|\varphi\rangle$ corresponds to an equal-weight superposition of eigenstates for all gauge-equivalent configurations\footnote{As there are only $N-1$ independent gauge transformations, despite its appearance $\mathcal{Q}$ does not perform $N$ independent projections, since Eq.~(\ref{eq:parity_constraint}) is already implemented (i.e., $\mathcal{Q}$ and $\mathcal{D}$ are not algebraically independent).} $\{u_{ij},\ut_{ij},v_j\}$; these all have the same eigenvalue (which follows from the gauge invariance of the Lindbladian, $\mathcal{L}\mathcal{Q}|\varphi\rangle=\mathcal{Q}\mathcal{L}|\varphi\rangle$).
$\mathcal{Q}|\varphi\rangle$ is the vectorization of $\mathcal{Q}(\varphi)=Q\varphi Q$, where $Q=\prod_j(\id+D_j)/2$; since $D_jQ=QD_j=Q$, $\mathcal{Q}(\varphi)$ satisfies Eq.~(\ref{eq:physical_states}) by construction.

A set of gauge-invariant observables that label physical states is furnished by string operators, but not all string operators: only the closed loops (flux operators) $W_\scC$, $\Wt_\scC$, and $\scW_{\scP}^\ch$ are gauge invariant. 
As mentioned in Sec.~\ref{sec:symm_lindblad} and shown in Appendix~\ref{app:fluxes}, the elementary fluxes are the fluxes $W_p$ through the plaquettes in one layer, the fluxes $W_\ell$ through noncontractible loops in the same layer, and the fluxes $\scW_q$ through interlayer plaquettes; their number totals $2E-N+1$. The fluxes $\Wt_\scC$ in the second layer can be obtained from the preceding ones, $W_\scC$ and $\scW_\scP^\ch$, since the flux through a closed cell is unity.
This reasoning provides an alternative block-decomposition of $\mathcal{L}$ in terms of gauge-invariant fluxes,
\begin{equation}
\label{eq:Ldecomp_fluxes}
    \mathcal{L}=\bigoplus_{\bm{w},\bm{\scw}}\mathcal{L}_{\bm{w},\bm{\scw}},
\end{equation}
with each block labeled by a strong- and a weak-flux configuration $\bm{w}$ and $\bm{\scw}$, respectively.
Since for each vertex $i$ there are $z(i)$ gamma matrices, by the handshaking lemma there are $\sum_i z(i)=2E$ gamma matrices in total. Subtracting the number $2E-N+1$ of (strong and weak) fluxes leaves $N-1$ binary degrees of freedom for the $c$ (and $\ct$) fermions [cf.\ the normalization $\mathcal{Z}$ of a sector with fixed fluxes in Eq.~(\ref{eq:ss})], in agreement with there being one adjoint fermion mode per site and fixed total adjoint fermionic parity [Eq.~(\ref{eq:parity_constraint})].

The fiducial state $|0\rangle$ is not an eigenstate of the strong fluxes $W_p$ and $W_\ell$ and thus does not belong to a single symmetry sector. It is, however, an eigenstate of the weak fluxes $\scW_q$ with eigenvalue one, see Eqs.~(\ref{eq:interlayerflux_gaugefields}) and (\ref{eq:fuv0}). We thus construct $2^{E-N+1}$ distinct vacua by projecting $|0\rangle$ into the symmetry sector labeled by $\bm{w}$, $|0\rangle_{\bm{w},\bm{\scw}=\bm{1}}=P_{\bm{w}}|0\rangle$ [with the fermionized form of Eq.~\eqref{eq:def_projector}]. The projection does not change the quantum numbers $\scw_q=1$ for any $q$; we denoted this weak-flux configuration by $\bm{\scw}=\bm{1}$.
The projected state $|0\rangle_{\bm{w},\bm{\scw}=\bm{1}}$ is an adjoint-fermion vacuum, $f_j|0\rangle_{\bm{w},\bm{\scw}=\bm{1}}=0$, because each $c_j$ and $\ct_j$ commutes with the projector built from $W_p$ and $W_\ell$. Upon projecting into the physical subspace with $\mathcal{Q}$, this vacuum is the vectorization of the (unnormalized) steady states $\rho_{\bm{w}}$.
We can also construct additional (non-steady-state) adjoint-fermion vacua by populating $|0\rangle_{\bm{w},\bm{\scw}=\bm{1}}$ with $b_j^\alpha$ Majoranas. As we will discuss in detail in Sec.~\ref{sec:map_gamma_fermion}, this flips the weak fluxes $\scw_q$ for some $q$ (that depend on the choice of $b_j^\alpha$); for a fixed $\bm{w}$ we have one vacuum for each choice of $\bm{\scw}$. In total, we can thus build $2^{2E-N+1}$ distinct vacua $|0\rangle_{\bm{w},\bm{\scw}}$ defined by
\begin{equation}
\label{eq:vaccum_w}
\begin{alignedat}{2}&f_j|0\rangle_{\bm{w},\bm{\scw}}=0, 
\quad &&W_p|0\rangle_{\bm{w},\bm{\scw}}=w_p |0\rangle_{\bm{w},\bm{\scw}},
\\
&W_\ell|0\rangle_{\bm{w},\bm{\scw}}=w_\ell |0\rangle_{\bm{w},\bm{\scw}},
\quad &&\scW_q|0\rangle_{\bm{w},\bm{\scw}}=\scw_q |0\rangle_{\bm{w},\bm{\scw}}.
\end{alignedat}
\end{equation}

\subsection{Gauge fixing}
\label{sec:gauge_fixing}

To perform actual calculations, we must fix a gauge. As discussed in Sec.~\ref{sec:steady}, there are exponentially many steady states, one for each strong-flux configuration $\bm{w}$ (i.e., vortex pattern) and $\bm{\scw}=\bm{1}$ (weak-symmetry constraint). The strong-symmetry constraint $W_p=\Wt_p$ is then automatically enforced as discussed in Sec.~\ref{sec:symm_lindblad} and Appendix~\ref{app:fluxes}. To satisfy $\scw_q=1$ for all $q$ [cf.\ Eq.~(\ref{eq:interlayerflux_gaugefields})] we make the gauge choice
\begin{equation}
\label{eq:gauge_choice}
    u_{ij}=\ut_{ij} 
    \quad \text{and} \quad
    v_i=-1,
\end{equation}
which is consistent with Eq.~(\ref{eq:fuv0}).\footnote{This is a different gauge choice from that of Ref.~\onlinecite{shackleton2024}, namely, $u_{ij}=-\ut_{ij}$ and staggered $v_i$ (note there is a minus sign difference between our definition of $\ut_{ij}$ and that of Ref.~\onlinecite{shackleton2024}). While the two choices are equivalent for the square lattice, the gauge in Ref.~\onlinecite{shackleton2024} can only be generalized to bipartite graphs, whereas ours works for any graph. Moreover, this gauge choice allows us to completely gauge away the edge variables $v_i$, see below.}
For any $\bm{w}$, $\scL|0\rangle_{\bm{w},\bm{\scw}=\bm{1}}=0$ and, hence, as expected from the discussion at the end of the previous subsection, we identify the adjoint-fermionic vacuum as the vectorization of the steady state $\rho_{\bm{w}}$. 
The excitations inside a steady-state sector are then obtained from the other $2^{N-1}-1$ eigenstates of $\scL$ with $(-1)^{\scN_f}=1$ [to maintain Eq.~(\ref{eq:parity_constraint})].

To go beyond steady-state sectors, i.e., to $|0\rangle_{\bm{w},\bm{\scw}\neq\bm{1}}$, we need $\scw_q=-1$ for some $q$, which by Eq.~(\ref{eq:interlayerflux_gaugefields}) can be achieved by flipping either the interlayer gauge field to $v_i=+1$ on sites $i$ belonging to some set $\setV$
or the intralayer gauge fields to $\ut_{ij}=-u_{ij}$ on the edges $\langle ij \rangle$ belonging to some set $\setU$.
We emphasize that the ``intralayer flips'' refer to flips of the fluxes in the second layer with respect to the first in the bilayer picture, and not to flips of the physical fluxes with respect to some reference configuration (this would correspond to changing to a different steady-state sector).
In non-steady-state sectors, the fermionic vacuum $|0\rangle_{\bm{w},\bm{\scw}\neq\bm{1}}$ is not annihilated by $\scL$ [and is not even a physical state if the number of gauge flips is odd, cf.\ Eq.~(\ref{eq:parity_constraint})]; all states are excited.

For a particular gauge choice $\{u_{ij}\}$ compatible with the background flux configuration $\bm{w}$ and the set of gauge flips $\setU$ and $\setV$, the Lindbladian reads
\begin{align}
    \scL
    &=4\sum_{\langle ij\rangle\in\overline{\setU}} J_{ij} u_{ij} f_i^\dagger f_j
    +2\sum_{\langle ij\rangle\in\setU} J_{ij} u_{ij} (f_i^\dagger f_j^\dagger+f_i f_j)
    \nonumber\\ \label{eq:Lindblad_non_steady}
    &-2\sum_{j\in\overline{\setV}} \gamma_j f_j^\dagger f_j
    -2\sum_{j\in\setV} \gamma_j (1-f_j^\dagger f_j),
\end{align}
with the adjoint fermionic parity fixed to $(-1)^{\scN_f}=(-1)^{|\setU|+|\setV|}$. Here, $\overline{\setU}$ ($\overline{\setV}$) denotes the complement of the set $\setU$ ($\setV$) and $|\setU|$ ($|\setV|$) denotes the cardinality of $\setU$ ($\setV$). 

From Eq.~(\ref{eq:Lindblad_non_steady}), we immediately see that excitations come in two types: if, for a given symmetry sector, there exists no gauge choice without intralayer flips ($\setU\neq\emptyset$), then there are superconducting terms in the Lindbladian and it conserves only adjoint-fermion-parity, $[\scL,(-1)^{\scN_f}]=0$. On the other hand, if a gauge choice exists such that there are only interlayer flips or no gauge flips at all ($\setU=\emptyset$), then the superconducting terms vanish, and the dynamics conserves the adjoint-fermion number, $[\scL,\scN_f]=0$. The projection into these symmetry sectors has more symmetry than the full Lindbladian.

Importantly, we can \emph{always} make the gauge choice $\setV=\emptyset$
and perform all necessary gauge flips as intralayer gauge flips: for each $\scw_q=\scw_{\langle ij \rangle}=-1$, we can simply flip the $\ut_{ij}$ on the corresponding edge while keeping all $v_j=-1$ fixed (i.e., there is a one-to-one correspondence between $\ut_{ij}$ and $\scw_{\langle ij \rangle}$). 
The gauge choice $\setV=\emptyset$ would simplify the Lindbladian by dropping the last term in Eq.~(\ref{eq:Lindblad_non_steady}) and allow for a unified treatment of all sectors. However, as a consequence of their enhanced symmetry, sectors that are gauge equivalent to $\setU=\emptyset$ (i.e., where $W_\scC=\Wt_\scC$ for all $\scC$) have a different phenomenology from those that are not.
For this reason, we shall treat the two cases separately in Secs.~\ref{sec:excitations_parity} and \ref{sec:excitations_number}, respectively.
Before that, in Sec.~\ref{sec:map_gamma_fermion}, we identify the gamma-matrix strings in the original representation that correspond to the fractionalized excitations discussed above.

\subsection{Steady states and excitations from strings}
\label{sec:map_gamma_fermion}

A generic state does not have definite quantum numbers under all conserved fluxes. To determine the relaxation time of such a state, we would need to find all the symmetry sectors with which it overlaps, and then compute the smallest spectral gap therein. This is an intractable task given the exponential number of sectors (for small system sizes, one can devise optimization schemes for this task~\cite{gidugu2024}, but even then, the accessible system sizes are very limited).

An alternative is to compute the relaxation timescales of basis elements of the space of operators: any operator can be written as a linear combination of terms of the form 
\begin{equation}
\label{eq:basis_states}
    \sigma=\Sigma_L\rho_0\Sigma_R,
\end{equation}
where $\Sigma_L$ and $\Sigma_R$ are arbitrary strings of gamma matrices\footnote{By construction, gamma-matrix strings are unitary and, with a judicious choice of prefactor, they can always be made Hermitian.} and $\rho_0$ is a fixed density matrix.\footnote{While $\sigma$ in Eq.~(\ref{eq:basis_states}) is not a proper density matrix as it is traceless unless $\Sigma_L=\Sigma_R=\id$, the set of all possible $\sigma$ with $\rho_0$ fixed forms a basis of the space of operators.} Without loss of generality, we choose $\rho_0$ to be the vortex-free steady-state, i.e., $\rho_0=P_{\bm{w}=\bm{1}}/\mathcal{Z}$. 
After vectorization, the reference state $\rho_0$ is mapped to the vortex-free fermionic vacuum $|0\rangle_{\bm{w}=\bm{1},\bm{\scw}=\bm{1}}$.
If $\Sigma_R=\id$ ($\Sigma_L=\id$), then we say $\sigma$ is constructed by \emph{left (right) action} of $\Sigma_L$ ($\Sigma_R$). If $\Sigma_L=\Sigma_R=\Sigma\neq\id$, we say it is built by the \emph{adjoint action} of $\Sigma$. 

Crucially, any state of the form~(\ref{eq:basis_states}) is fully supported inside a single symmetry sector: $\rho_0$ has fixed and known quantum numbers, namely, $w_p=w_\ell=\scw_q=1$ for all $p$, $\ell$, and $q$; and all strings $\Sigma_L$ and $\Sigma_R$ either commute or anticommute with all conserved fluxes, hence they either preserve or flip these quantum numbers.
More precisely, if we denote
\begin{alignat}{9}
    &\Sigma_L W_p= s_p W_p \Sigma_L
    \quad &&\text{and}\quad
    \Sigma_L W_\ell= s_\ell W_\ell \Sigma_L,
    \\
    &\Sigma_L W^\ch_q = \scs_q^L W^\ch_q \Sigma_L
    \quad &&\text{and}\quad
    \Sigma_R W^\ch_q = \scs_q^R W^\ch_q \Sigma_R,
\end{alignat}
with $s_p,s_\ell,\scs_q^L,\scs_q^R\in\{-1,1\}$, then $\sigma$ belongs to the sector labeled by the quantum numbers $w_p=s_p$, $w_\ell=s_\ell$, and $\scw_q=\scs_q^L\scs_q^R$. 
We do not need to specify the commutation relation of $\Sigma_R$ with the strong fluxes, since $\wt_p$ and $\wt_\ell$ are completely fixed by $w_p$, $w_\ell$ and $\scw_q$.
We then choose a gauge configuration $u_{ij}$ compatible with the strong fluxes and fix the inter- and intralayer gauge flips $\setV$ and $\setU$ (recall that there is gauge freedom in this choice) in accordance with the weak fluxes.

{\setlength{\tabcolsep}{4pt}
\begin{table*}[t]
\caption{Summary of the mapping between states in the spin and fractionalized representations. The numbering of the first column follows that of the main text. In the second and third columns, we list the strings, defined in Eqs.~(\ref{eq:def_W})--(\ref{eq:WC}), that act on the reference state $\rho_0$ and whether the action is left, right, or adjoint.
The fourth column indicates in which symmetry sector the resulting state $\sigma=\Sigma_L\rho_0\Sigma_R$ lives, and whether it is a steady or a transient state. The final four columns specify the gauge and fermionic content of the state to which $\sigma$ is mapped in the fractionalized picture, namely, vortices, intralayer and interlayer gauge excitations, and Majorana fermions $c_j$. States with multiple types of excitations are obtained by composing these elementary operations.}
\label{tab:excitations}
\begin{tabular}{@{}llllllll@{}}
\toprule \toprule
 & Type of string  & Action & Mapped state & Vortices & Intralayer & Interlayer & Fermions
\\ \toprule
(i) & \begin{tabular}[c]{@{}l@{}}$W_\scC$\\$W_\scP$\\$W_\scP^\ch$\end{tabular}                                                        & \begin{tabular}[c]{@{}l@{}}Any\\Adjoint\\Adjoint\end{tabular}                           & Itself                                                                                        & No       & No         & No         & No                                                                          \\ \midrule
(ii) & $W_\scP^{\mu\nu}$                                                    & Adjoint                         & \begin{tabular}[c]{@{}l@{}}Steady state in a different\\symmetry sector\end{tabular}     & Yes      & No         & No         & No                                                                          \\ \midrule
(iii) & $W_\scP^{\mu\nu}$                                                    & Right                      & \begin{tabular}[c]{@{}l@{}}Transient operator in a\\sector with $w_p\neq\wt_p$\end{tabular}  & No      & Yes        & No         & No                                                                          \\ \midrule
(iv) & $W_\scP^\ch$                                                   & Left/right                            & \begin{tabular}[c]{@{}l@{}}Transient operator in a\\sector with $\scw_q=-1$\end{tabular}    & No       & No         & Yes        & No                                                                          \\ \midrule
(v) & $W_\scP$      & Left/right                            & \begin{tabular}[c]{@{}l@{}}Transient operator in\\the same symmetry sector\end{tabular}              & No       & No         & No         & Yes  
\\ \bottomrule \bottomrule
\end{tabular}
\end{table*}
}

Our strategy is to choose strings $\Sigma_L$ and $\Sigma_R$ such that we can construct representative states with each type of excitation described in Sec.~\ref{sec:gauge_fixing}: vortices, intra- and interlayer gauge excitations, and Majorana fermions. To this end, it suffices to consider the strings $W^{\mu\nu}_\scP$, $W_\scP$, $W_\scP^\ch$, and $W_\scC$ defined in Eqs.~(\ref{eq:def_W})--(\ref{eq:WC}), as described in the following and summarized in Table~\ref{tab:excitations} and Fig.~\ref{fig:strings}.
\begin{enumerate}[(i)]
    \item \textit{Closed strings $W_\scC$}. As explained in Sec.~\ref{sec:strings_fluxes}, $W_\scC$ commutes with $W_{\scC'}$ and $W^\ch_\scP$ for all $\scC'$ and $\scP$. Hence, for any action (left, right, or adjoint) of $W_\scC$, we have $s_p=s_\ell=\scs_q^L=\scs_q^R=1$ for all $p$, $\ell$, and $q$, and $\rho_0$ is mapped to itself ($\sigma=\rho_0$). Accordingly, the action of $W_\scC$ does not produce any excitations.
    \item \textit{Open strings $W_\scP^{\mu\nu}$ with $\mu,\nu\in\{1,\dots,2k\}$}. $W_\scP^{\mu\nu}$ is formed by a body $\prod_{\langle ij \rangle\in\scP}K_{ij}$ and two gamma matrices at the endpoints. In the fractionalized picture, this maps to a string of gauge fields, $\prod_{\langle ij \rangle\in\scP}(-\i u_{ij})$, with a Majorana fermion $b_j^\mu$ at the endpoint $j$ of $\scP$ (and similarly for the other endpoint); see Fig.~\ref{fig:strings}(a). We have $\scs_q^{L,R}=-1$ on the edge that emanates from $j$ with color $\mu$, and $s_p=-1$ on the two plaquettes adjacent to it, which in the fractionalized representation reads as the Majorana $b_j^\mu$ flipping the gauge field $u_{ij}$ on the edge with color $\mu$. If the action of $W_\scP^{\mu\nu}$ is adjoint, the $\scs_q^{L,R}=-1$ cancel each other out ($\scw_q=\scs_q^L\scs_q^R=+1$) and the string excites only vortices in the system. Hence, $\sigma=W^{\mu\nu}_\scP\rho_0 W^{\mu\nu}_\scP=P_{\bm{w}}/\mathcal{Z}$ is a different steady state, where $w_p=-1$ on the plaquettes mentioned above and $w_p=+1$ otherwise; in the vectorized representation, $\sigma$ maps to the vacuum $|0\rangle_{\bm{w}\neq\bm{1},\bm{\scw}=\bm{1}}$.
    \item When $W^{\mu\nu}_\scP$ acts on the right, we do not induce vortices ($s_p=+1$ for all $p$), but the $\scs_q^{L,R}$ no longer cancel, leading instead to an intralayer gauge excitation on the same edge ($\ut_{ij}=-u_{ij}$). The mapped state $\sigma=\rho_0W^{\mu\nu}_\scP$ is not a steady state but instead a decaying operator in a symmetry sector with $w_p\neq \wt_p$; its vectorization is the vacuum $|0\rangle_{\bm{w}=\bm{1},\bm{\scw}\neq\bm{1}}$.
    On the other hand, the left action can be interpreted as the composition of the right and adjoint actions. In that case, $\sigma$ has both vortices and intralayer gauge excitations.
    \item \textit{Open chiral strings $W_\scP^\ch$ ($\mu,\nu=2k+1$)}. Similarly to before, these strings have a ``$\mathbb{Z}_2$ gauge body'', but they carry chiral Majoranas $b_j^{\ch}$ at their endpoints; see Fig.~\ref{fig:strings}(b). $W_\scP^\ch$ commutes with all $W_\scC$, thus $s_p=s_\ell=1$ for all $p$ and $\ell$, and chiral strings do not excite vortices. On the other hand, $\scs_q^{L,R}=-1$ for all edges $q$ that emanate from the endpoints of the string and $\scs_q^{L,R}=1$ for all other edges. In the fractionalized picture, it corresponds to the chiral Majorana $b_j^\ch$ flipping the gauge field $v_j$; this is gauge-equivalent to flipping all the gauge fields $u_{ij}$ emanating from site $j$ instead (as can also be seen by direct inspection of the commutation relations of $u_{ij}$, $\ut_{ij}$, $v_j$, and $b^\ch_j$), and we conclude that open chiral strings effect interlayer gauge flips at their endpoints.
    For the adjoint action, the interlayer gauge flips induced in each layer cancel, and no excitations are induced. We thus see that $\sigma=W_\scP^\ch\rho_0=\rho_0W_\scP^\ch$ is a transient operator in a symmetry sector with $\scw_q=-1$ only and maps to a fermionic vacuum $|0\rangle_{\bm{w}=\bm{1},\bm{\scw}\neq\bm{1}}$ with $\bm{\scw}$ as described above. Hence, starting from the vacuum in the reference sector, we can generate all other possible vacua by acting with $W^{\mu\nu}_\scP$ and $W_\scP^\ch$.
    \item \textit{Open strings $W_\scP$ ($\mu,\nu=0$)}. This final type of string carries instead itinerant Majorana fermions $c_j$ at its endpoints; see Fig.~\ref{fig:strings}(c). By the arguments laid out after Eq.~(\ref{eq:WC}), $W_\scP$ does not commute with every $K_{ij}$, but it does commute with all $W_\scC$ and $W_\scP^\ch$. Hence, $s_p=s_\ell=\scs_q^L=\scs_q^R=1$ and there are no gauge excitations. $\sigma=W_\scP\rho_0=\rho_0W_\scP$\footnote{At the vectorized level, the equality of left and right actions translates to Eq.~(\ref{eq:vaccum_w}).} is a decaying operator in the same symmetry sector as $\rho_0$ and is vectorized to $c_ic_j|0\rangle_{\bm{w}=\bm{1},\bm{\scw}=\bm{1}}=f^\dagger_if^\dagger_j|0\rangle_{\bm{w}=\bm{1},\bm{\scw}=\bm{1}}$ [where we used Eqs.~(\ref{eq:def_fFermions}) and (\ref{eq:vaccum_w})], creating a superposition of fermionic excitations. On the other hand, the adjoint action of $W_p$ leaves $\rho_0$ invariant, which in the vectorized picture can be understood as inserting both a $c_j$ and a $\ct_j$ fermion into the system, using the vacuum to reflect the $\ct_j$ to another $c_j$ [Eq.~(\ref{eq:c_fiducial})], and then annihilating it with the first $c_j$.
\end{enumerate}

More general strings, not of the form (\ref{eq:def_W}) can also be considered but lead to no new excitations. For example, closed strings that are not the products of $K_{ij}$ but some other gamma-matrix bilinears, do not commute with the fluxes and thus induce vortices along the body of the string. 
Composing the previous operations, every operator in any symmetry sector can be reached, thus completing the enumeration of possible excitations in both the gamma-matrix and fermion representations.

\section{Fermion-parity-conserving dynamics}
\label{sec:excitations_parity}

We now determine the eigenvalues of the excitations of the free-fermion Lindbladian in a given symmetry sector, which determine the respective relaxation rates. We start with the more general case of fermion-parity conservation, which offers a unified solution for any sector. While for sectors with intralayer gauge flips ($\setU\neq\emptyset$) this is the only procedure available, in sectors with $W_\scC=\Wt_\scC$ for all $\scC$ (i.e., that can be gauged to $\setU=\emptyset$) it is convenient to exploit the enhanced symmetry (fermion-number conservation) to obtain information beyond the general solution. This will be addressed in Sec.~\ref{sec:excitations_number}.

\subsection{Diagonalization of the Lindbladian}

We can write the Majorana representation~(\ref{eq:Lind_quad_Majorana}) of the Lindbladian in block-matrix form as
\begin{equation}
\label{eq:Lindblad_intralayer}
    \scL=\frac{1}{4}\sum_{ij}
    \begin{pmatrix} c_i & \ct_i \end{pmatrix}
    \begin{pmatrix}
    \dsA_{ij} & \i \dsD_{ij} \\ 
    -\i \dsD_{ij} & \tA_{ij}
    \end{pmatrix}
    \begin{pmatrix} c_j \\ \ct_j \end{pmatrix}
    -A_0,
\end{equation}
where $A_0=\sum_j \gamma_j$ and the real $N\times N$ single-particle matrices $\dsA$, $\tA$ and $\dsD$ are given by
\begin{equation}
\label{eq:def_dsA_dsD}
    \dsA_{ij}=4J_{ij}u_{ij}, 
    \quad
    \tA_{ij}=4J_{ij}\ut_{ij}, 
    \quad \text{and} \quad
    \dsD_{ij}=2\delta_{ij}\gamma_i v_i.
\end{equation}
Denoting $\vec{c}=(c_1\ \ct_1\cdots c_N \ \ct_N)^T$, with $(\cdot)^T$ denoting the transpose, and
\begin{equation}
\label{eq:def_A}
    A=P^T\begin{pmatrix}
    \dsA & \i \dsD \\ 
    -\i \dsD & \tA
    \end{pmatrix}P,
\end{equation}
where $P$ is the permutation matrix defined by $P\vec{c}=(c_1 \cdots c_N\,\ct_1\cdots \ct_N)^T$, we write Eq.~(\ref{eq:Lindblad_intralayer}) as 
\begin{equation}
    \scL=\frac{1}{4}\vec{c}^T A \vec{c}-A_0.
\end{equation}

\begin{figure}[t]
    \centering
    \includegraphics[width=\columnwidth]{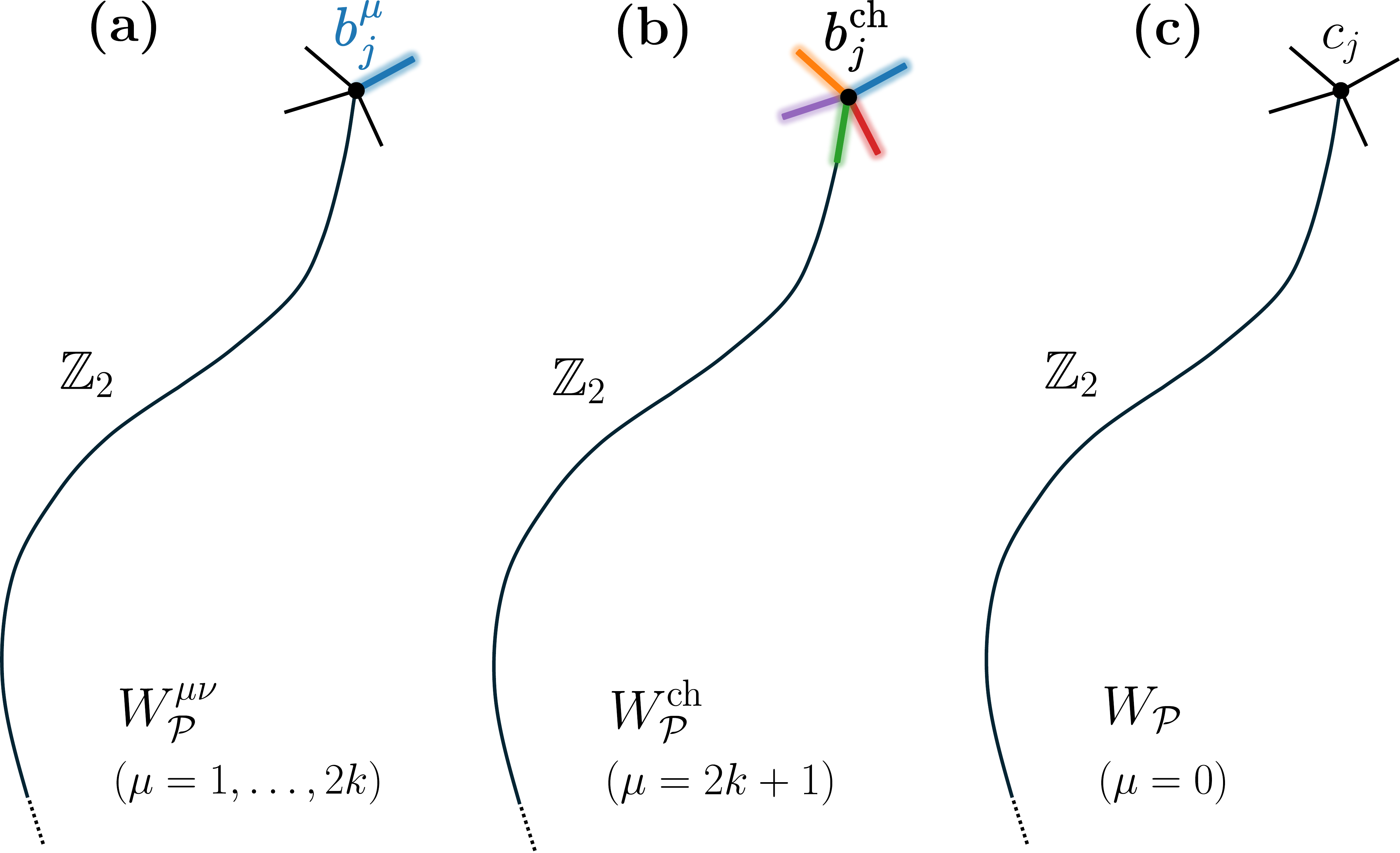}
    \caption{Fractionalized representation of string excitations of the gamma-matrix model. Schematic representation of strings of gamma matrices on an open path $\scP$, defined in Eqs.~(\ref{eq:def_W})--(\ref{eq:Wch}) (only one endpoint $j$ depicted). The body of the string, $\prod_{\langle ij \rangle\in\scP}K_{ij}$, is mapped to a product of $\mathbb{Z}_2$ gauge fields, $\prod_{\langle ij \rangle\in\scP}(-\i u_{ij})$, with Majorana fermions at each end.
    (a) The string $W_\scP^{\mu\nu}$ with $\mu,\nu\in\{1,\dots,2k\}$ carries a Majorana fermion $b_j^\mu$ at its endpoint, which excites intralayer gauge flips on the edge with color $\mu$ that emanates from $j$ (represented by a bold colored edge) and, for the adjoint action, two vortices adjacent to that edge. (b) The string $W_\scP^\ch$ carries a chiral Majorana fermion $b_j^{\ch}$ at its endpoint that effects an interlayer gauge flip at $j$ (which is gauge-equivalent to flipping all intralayer gauge fields emanating from $j$, as depicted by the colored edges). (c) The string $W_\scP$ carries an itinerant Majorana fermion $c_j$ at its endpoint, introducing a superposition of fermionic excitations at $j$.}
    \label{fig:strings}
\end{figure}

Because of the antisymmetry of $A$, its $2N$ eigenvalues come in pairs $(\beta_k,-\beta_k)$, where, without loss of generality, we order $0\leq \Re \beta_1 \leq \cdots \leq \Re \beta_N$; if $\Re\beta_k=0$, then we take $\Im\beta_k>0$.\footnote{Here, we assumed that the Lindbladian is sufficiently generic such that it is diagonalizable and that there are no eigenvalues $\beta_k=0$. The latter would imply there are multiple steady states in a single symmetry sector and lead to the breakdown of the Pfaffian method described below.}
Denoting the right eigenvectors of $A$ by $r$, $Ar_{2k-1}=\beta_kr_{2k-1}$ and $Ar_{2k}=-\beta_k r_{2k}$, we define the $2N\times 2N$ matrix $O$ whose rows are $O_{2k-1}=(r^T_{2k-1}+r^T_{2k})/\sqrt{2}$ and $O_{2k-1}=-\i(r^T_{2k-1}-r^T_{2k})/\sqrt{2}$. Since the eigenvectors of $A$ are complex and biorthogonal, $O$ is a complex orthogonal matrix, $OO^T=O^TO=\id$, and we can use it to reduce $A$ to the canonical form $A=O^T B O$~\cite{kitaev2001Wire,prosen2008}, with
\begin{equation}
    B=\begin{pmatrix}
        0 & -\i \beta_1 & & \\
        \i \beta_1 & 0 & & \\
        & & \ddots & & \\
        & & & 0 & -\i\beta_N \\
        & & & \i\beta_N & 0
    \end{pmatrix}.
\end{equation}
With this change of basis, we can diagonalize $\scL$ as
\begin{equation}
    \scL
    =-\frac{\i}{2}\sum_{k=1}^N \beta_k a_{2k-1}a_{2k}-A_0
    =-\sum_{k=1}^N\beta_k d'_kd_k-B_0,
\end{equation}
with $B_0=A_0-\sum_{k=1}^N \beta_k/2$.
Here, $\vec{a}=(a_1\cdots a_{2N})^T=O \vec{c}$ are rotated Majoranas, which by virtue of $O$ being complex orthogonal, satisfy the anticommutation relation $\{a_i,a_j\}=2\delta_{ij}$ but are non-Hermitian, $a_i\neq a_i^\dagger$. The associated fermionic creation and annihilation operators,
\begin{equation}
\label{eq:def_dFermions}
    d'_k=\frac12(a_{2k-1}-\i a_{2k}),
    \qquad
    d_k=\frac12(a_{2k-1}+\i a_{2k}),
\end{equation}
referred to as normal master modes (NMMs)~\cite{prosen2008}, satisfy
\begin{equation}
    \{d_i,d_j\}=\{d'_i,d'_j\}=0, \quad
    \{d_i,d'_j\}=\delta_{ij}, \quad
    d_i'\neq d_i^\dagger.
\end{equation}
The NMM vacuum $|\Omega\rangle$ is annihilated by all annihilation operators $d_k$, and excitations are built on top of it by the action of the creation operators $d'_k$, $|\bm{\nu}\rangle=\prod_k (d'_k)^{\nu_k}|\Omega\rangle$, with $\nu_k\in\{0,1\}$. The many-body eigenvalues are then 
\begin{equation}
\label{eq:NMM_ev}
\begin{split}
    \Lambda_{\bm{\nu}}
    &=-\sum_{k=1}^N\beta_k\nu_k-B_0
    \\
    &=-\sum_{k=1}^N\beta_k\left(\nu_k-\frac{1}{2}\right)-\sum_{j=1}^N\gamma_j.
\end{split}
\end{equation}

\subsection{Parity of physical states from the Pfaffian}

To proceed, we must determine the states with the correct fermionic parity [cf.\ Eq.~(\ref{eq:parity_constraint})]. To this end, we define the fermion number and parity in the new basis:
\begin{equation}
\label{eq:parity_d}
    (-1)^{\scN_d}=\prod_{k=1}^N(-\i a_{2k-1}a_{2k}),
    \qquad
    \scN_d=\sum_{k=1}^N d_k'd_k.
\end{equation}
The many-body eigenstates have parity $(-1)^{\scN_d}|\bm{\nu}\rangle=\prod_{k}(-1)^{\nu_k}|\bm{\nu}\rangle$, which we can fix using the following:
\begin{lemma*}
A state with a total of $\nu=\sum_k \nu_k$ populated NMMs is physical if and only if
\begin{equation}
\label{eq:parity_det}
    (-1)^\nu=(\det O) (-1)^{|\setV|+|\setU|}.
\end{equation}
Moreover, if there are $|\setI|$ purely imaginary $\beta_k$, then 
\begin{equation}
\label{eq:det_Pf}
    \det O=\sgn\big((-\i)^{|\setI|}\Pf \i A\big),
\end{equation}
where $\Pf$ is the Pfaffian.

\end{lemma*}
\begin{proof}
Performing the change of basis between the original and rotated Majoranas, we have that $(-1)^{\scN_f}=\prod_{j=1}^N(-\i c_j \ct_j)=(-\i)^N\sum_{j_1\cdots j_{2N}}O_{j_11}\cdots O_{j_{2N}2N}a_{j_1}\cdots a_{j_{2N}}$. Terms with any repeated indices do not contribute to the sum: if $j_a=j_b$ for some $a\neq b$, then $\sum_{j_a} O_{j_aa}O_{j_ab}a_{j_a}^2=(OO^T)_{ab}=\delta_{ab}=0$. Thus, only the terms with all indices $j_1,\dots,j_N$ distinct contribute. Since the product $a_{j_1}\cdots a_{j_{2N}}$ is fully antisymmetric in its indices, we can replace it with $\varepsilon_{j_1\cdots j_{2N}}a_1\cdots a_{2N}$, where $\varepsilon_{j_1\cdots j_{2N}}$ is the fully antisymmetric symbol, and it follows that $(-1)^{\scN_f}=(-1)^{\scN_d}\sum_{j_1\cdots j_{2N}}\varepsilon_{j_1\cdots j_{2N}}O_{j_11}\cdots O_{j_{2N}2N}=(-1)^{\scN_d} \det O$, where we used Eq.~(\ref{eq:parity_d}), i.e., the change of basis is parity preserving if and only if $O$ is smoothly connected to the identity, which is parity preserving.
Using Eq.~(\ref{eq:parity_constraint}), Eq.~(\ref{eq:parity_det}) follows. 
Furthermore, $\Pf \i A=\Pf(O^T \i B O)=\det O\, \Pf \i B=\det O \prod_{j=1}^N\beta_j$. Since $A$ is similar to a real matrix ($A=T A^* T^{-1}$ with $T=P(\sigma^z\otimes \id)P^T$), and $\pm\beta_j$ are its eigenvalues, we have three possibilities: (i) $\beta_k\in\mathbb{R}$ and it is, by assumption, positive; (ii) $\beta_k$ come in complex-conjugate pairs, i.e., there exists some $k\neq j$ such that $\beta_k=\beta_j^*\implies\beta_j\beta_k>0$; or (iii) $\beta_k\in\mathrm \i\mathbb{R}$ with, by assumption, $\mathrm{Im}\beta_k>0$. Hence, $\Pf\i B=\i^{|\setI|}R$ for some $R>0$. Since $O$ is orthogonal, $\pm1=\det O=\Pf\i A/\Pf\i B=(-\i)^{|\setI|}\Pf\i A /R$, and taking the sign of both sides, we arrive at Eq.~(\ref{eq:det_Pf}).
\end{proof}

\subsection{Spectral gap and relaxation rates}
\label{sec:parity_gap}

With the prescription of Eq.~(\ref{eq:parity_det}), all physical states and their relaxation rates can in principle be computed. In particular, the spectral gap $\Delta=\min_{\bm{\nu}} |\Re \Lambda_{\bm{\nu}}|$ gives the asymptotic long-time decay rate of states or observables with overlap in that symmetry sector. 

\begin{figure*}[t]
    \centering
    \includegraphics[width=0.8\textwidth]{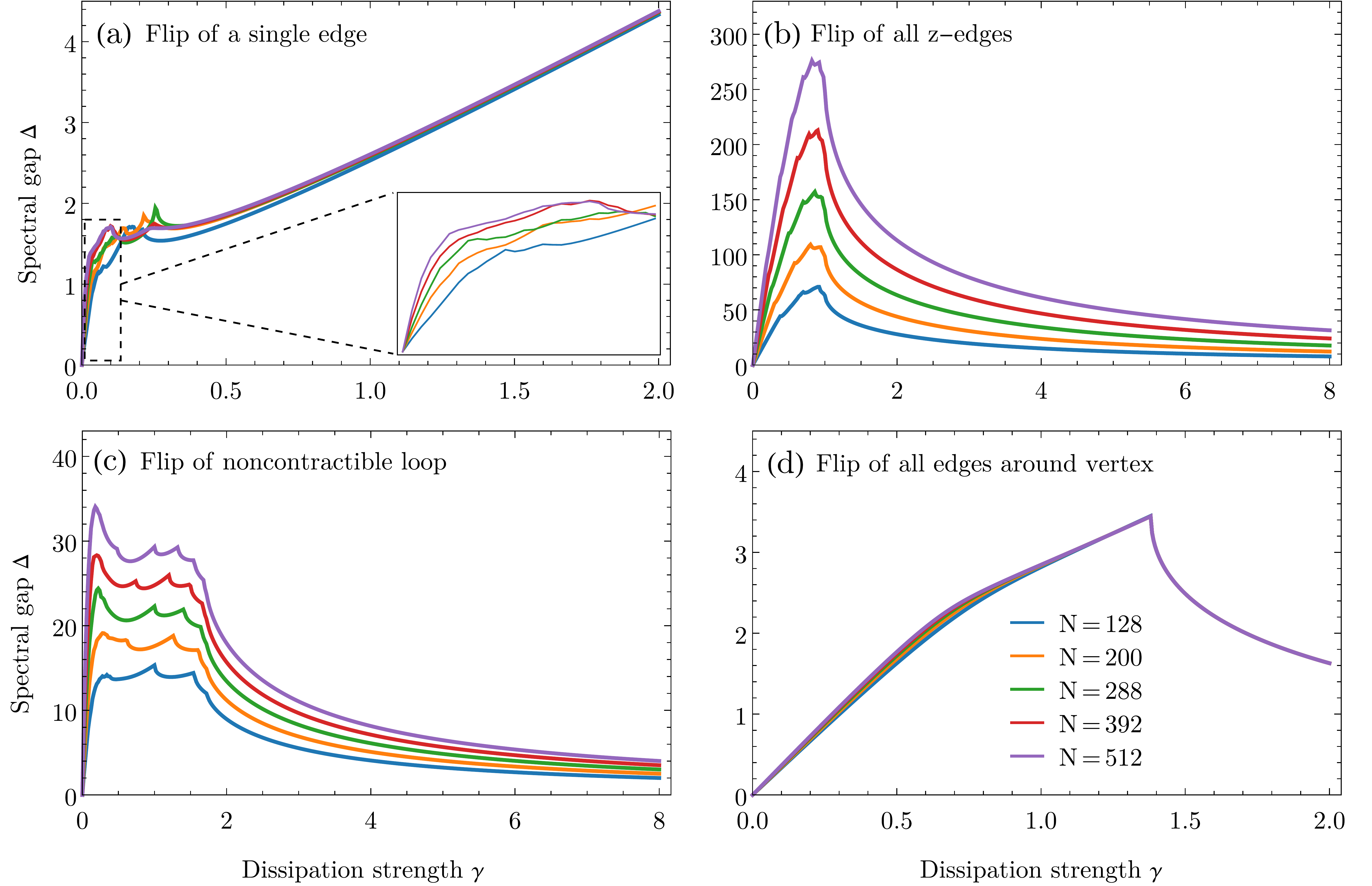}
    \caption{Spectral gap as a function of dissipation strength in fermion-parity-conserving sectors. 
    The gap was computed for the honeycomb lattice with $N=128$--$512$ and periodic boundary conditions on the torus (with the same number of sites along each direction) using the prescription of Sec.~\ref{sec:parity_gap}, with the single-particle eigenvalues $\beta_k$ obtained from exact diagonalization of Eq.~(\ref{eq:def_A}). We set uniform Hamiltonian couplings and gauge fields $J_{ij}=u_{ij}=1$ (which gives uniform 0-flux), uniform dissipation $\gamma_i=\gamma$ for all $i$, and fixed the gauge to $v_i=-1$ for all $i$ (i.e., $\setV=\emptyset$). Each panel shows a different type of intralayer gauge flips $\ut_{ij}=-1$ (i.e., choice of $\setU$): (a) a single arbitrary edge; (b) all edges of a given color; (c) all edges on a noncontractible loop; and (d) all edges emanating from an arbitrary vertex (which is gauge-equivalent to a single interlayer flip on that vertex and $\setU=\emptyset$, with conserved fermion number, cf.\ upper left panel of Fig.~\ref{fig:intergapUniform}).
    }
    \label{fig:intragapHoney}
\end{figure*}

Let us first consider steady-state sectors (i.e., $\dsA=\tA$, $\setU=\emptyset$, $\setV=\emptyset$).
For $\gamma_j\to0$, $\Pf \i A \to\Pf\i\dsA\Pf\i\tA=(\Pf\i\dsA)^2=\det \i\dsA$ and $|\setI|\to N$ (because $A$ becomes real antisymmetric). Thus, $(-\i)^{|\setI|}\Pf\i A\to\det \dsA\geq 0$ (since $\dsA$ is real antisymmetric), hence $\det O\to1$. Then, by continuity of $\det$, the NMM vacuum $|\Omega\rangle$ is physical for small $\gamma_j$ and, by Eq.~(\ref{eq:NMM_ev}), it is the slowest decaying state. Hence, it must coincide with the steady state (i.e., the original adjoint-fermion vacuum $|0\rangle_{\bm{w},\bm{\scw}=\bm{1}}$) and we must have $B_0=0$, i.e., $\sum_k \beta_k/2=\sum_j \gamma_j$~\cite{prosen2008} (this sum rule does not hold outside steady-state sectors). 
The existence of the vacuum steady state for any $\gamma_j$ indicates that the Pfaffian cannot change sign upon increasing $\gamma_j$, and all physical states in these sectors have an even number of fermion excitations, $\prod_j(-1)^{\nu_j}=1$. As discussed above, there is an additional U(1) symmetry in these sectors, and excitations are labeled by the adjoint-fermion number $n$. The gap in this sector is given by $\Delta_n=\sum_{k=1}^n \Re\beta_k$ ($n\in\{0,2,4,\dots, N\}$). This case will be treated in detail in Sec.~\ref{sec:excitations_number}.

Now we turn to non-steady state sectors ($\setU\neq\emptyset$, $\setV\neq\emptyset$, or both). All states, including the NMM vacuum, are decaying. If $\sgn\big((-\i)^{|\setI|}\Pf \i A\big)=(-1)^{|\setV|+|\setU|}$, the NMM vacuum is physical and $\Delta=B_0$. Otherwise, the longest-lived physical state has a single occupied NMM and the gap reads $\Delta=B_0+\Re\beta_1$. Fixing a gauge configuration is not enough to determine whether the vacuum is physical or not, as $\sgn\Pf\i A$ can flip as a function of some parameter of the model (say, one of the $\gamma_j$). Moreover, if $\Re\beta_1=0$, both the first physical and first nonphysical states share the same decay rate, $\Delta=B_0=B_0+\Re\beta_1$.

\subsection{Numerical results}
\label{sec:numerics_parity}

In Fig.~\ref{fig:intragapHoney}, we plot the gap $\Delta$ as a function of dissipation strength $\gamma$ for the honeycomb lattice on the torus for fixed $\setV=\emptyset$ and different choices of $\setU$.
The qualitative features of the gap for the square~\cite{shackleton2024}, honeycomb (Fig.~\ref{fig:intragapHoney}), and triangular (Appendix~\ref{app:additional_numerics}) lattices are similar, indicating that dissipative dynamics are more strongly determined by the structure of the Lindbladian (\ref{eq:Lindblad_intralayer}) than by the microscopic details of the couplings $J_{ij}$ (i.e., of the underlying graph).

Some general features, which can be understood perturbatively in $\gamma$ or $1/\gamma$, are immediately visible at finite $N$: the gap grows linearly as $\gamma\to\infty$ and closes as $\gamma\to0$. At intermediate values of $\gamma$, a rich nonmonotonic behavior occurs due to a succession of crossings of the single-particle eigenvalues $\Re\beta_k$ or exceptional points [i.e., values of $\gamma$ at which a complex-conjugate pair $(\beta_k,\beta_k^*)$ merges into a doubly-degenerate real eigenvalue]. The kinks in the gap that result from these crossings and exceptional points are particularly visible, e.g., in Fig.~\ref{fig:intragapHoney}(c).

At small $\gamma$, the slope of the gap increases with increasing $N$, see Figs.~\ref{fig:intragapHoney}(a)--(c), which suggests that the weak dissipation ($\gamma\to0$) and thermodynamic ($N\to\infty$) limits do not commute. Hence, if the thermodynamic limit is taken first, the gap converges to a finite constant of order $J$ (where $J$ is the typical coupling strength of $H$) in the limit of $\gamma\to0$, signaling a much faster relaxation than expected from perturbation theory. This phenomenon, dubbed anomalous relaxation~\cite{garcia2023PRD2}, signals that internal interactions, not external dissipation, govern the relaxation. It has been recently observed in a variety of systems~\cite{sa2022PRR,garcia2023PRD2,mori2024,prosen2002,prosen2004,prosen2007,shackleton2024,yoshimura2024,yoshimura2025,znidaric2024,znidaric2024b}, and it is proposed to be a generic feature of ergodic many-body quantum dynamics~\cite{mori2024,yoshimura2024,yoshimura2025,jacoby2024,zhang2024} (via a relation to quantum Ruelle-Pollicott resonances~\cite{mori2024,yoshimura2025}).\footnote{The noncommutativity of limits can be understood heuristically as follows~\cite{mori2024}. For an $N$-site system with bulk dissipation $\gamma$, there is a natural effective dissipation strength $\gamma_{\mathrm{eff}}=N\gamma$. It is thus manifest that the two limits do not commute: $\lim_{N\to\infty}\lim_{\gamma\to0}\gamma_{\mathrm{eff}}=0$, while $\lim_{\gamma\to0}\lim_{N\to\infty}\gamma_{\mathrm{eff}}=\infty$.}
Nevertheless, the weak-dissipation dynamics of our gamma-matrix model are not fully mixing because of the emergent integrability arising from the fractionalizing into an exponentially large number of gauge sectors with free-fermion excitations. In every sector gauge-equivalent to $\setU=\emptyset$, the Linbdladian has an additional conserved quantity in the limit of $\gamma\to0$, namely the open gauge strings $W^{\ch}_\scP$ (together with $H$ itself, these exhaust the conservation laws of the Hamiltonian and hence the $\gamma\to0$ limit of the Lindbladian). By virtue of these extra conservation laws, there is an infinite number (in the thermodynamic limit) of gaps that close, and the system as a whole evades anomalous relaxation; see, e.g., Fig.~\ref{fig:intragapHoney}(d). By choosing the gauge where the particle-number conservation is manifest, we can analytically bound the gap in different particle-number sectors and prove the absence of anomalous relaxation; see Sec.~\ref{sec:bounds}. Notwithstanding, the (also extensive) number of sectors with finite gaps further illustrates the ubiquity of anomalous relaxation as a generic phenomenon in strongly interacting many-body dynamics. 

On the other hand, certain symmetry sectors display the quantum Zeno effect---the closing of the spectral gap with increasing dissipation strength---which was reported for homogeneous dissipation on the square lattice in Ref.~\onlinecite{shackleton2024}. For the gap to close in this limit, there must exist an operator $\Sigma_R$ that commutes with all jump operators $L_j\sim\Gamma_j^\ch$. Such operators exist in sectors that either are gauge-equivalent to $\setU=\emptyset$ or have $\setV=\emptyset$ and $(-1)^{|\setU|}=-1$. The former will be discussed in detail in Sec.~\ref{sec:numerics_number}, and here we focus on the latter. For any choice of $\setU$ with even $|\setU|$, we consider a product $\Sigma_R$ of open gauge strings $W_\scP^{\mu\nu}$ ($\mu,\nu\in\{1,\dots,2k\}$) whose endpoints connect a pair of sites from which a gauge flip emanates, such that all flips are connected by one such string. From Sec.~\ref{sec:map_gamma_fermion}, $\sigma=\rho_{\bm{w}}\Sigma_R$ belongs to the sector labeled by $\setU$ and, because $W_\scP^{\mu\nu}$ is built from two gamma matrices on each site, it commutes with all chiral elements $\Gamma_j^\ch$. Hence, $\sigma$ is annihilated by $\scL$ in the limit $\gamma_j\to\infty$. On the other hand, if $|\setU|$ is odd, one of the gauge flips is not paired up. Hence, there exists a site $j$ for which $\Gamma_j^\mu$ appears only once in $\Sigma_R$, which thus anticommutes with $\Gamma_j^\ch$, and the gap diverges. These predictions are confirmed in Fig.~\ref{fig:intragapHoney}.\footnote{At first sight, this reasoning should imply that the gap diverges as $\gamma\to\infty$ in Fig.~\ref{fig:intragapHoney}(d) because $|\setU|$ is odd. However, as discussed in Sec.~\ref{sec:map_gamma_fermion}, the action of $\Gamma_j^\ch$ is not to flip all the $\ut_{ij}$ around a site but to flip the corresponding $v_j$. A representative string in the gauge-equivalent representation with $\setV=\emptyset$ (which we are employing in the numerical calculations) is $\Sigma_R=D_j\Gamma_j^\ch$. This operator, besides flipping the intralayer gauge fields around site $j$, also anticommutes with the fermionic parity $(-1)^{\scN_f}$, leading to the result of Fig.~\ref{fig:intragapHoney}(d), which is also in agreement with following section, cf.\ Fig.~\ref{fig:intergapUniform}.}

\section{Fermion-number-conserving dynamics}
\label{sec:excitations_number}

\subsection{Diagonalization of the Lindbladian}

The Lindbladian of Eq.~(\ref{eq:Lindblad_non_steady}) with $\setU=\emptyset$ reads
\begin{equation}
\label{eq:Lindblad_interlayer}
    \scL=\sum_{ij} f_i^\dagger \dsL_{ij} f_j - L_0, 
\end{equation}
where the real $N\times N$ single-particle matrix $\dsL=\dsA+\dsD$ is the sum of the antisymmetric and diagonal matrices $\dsA$ and $\dsD$ defined in Eq.~(\ref{eq:def_dsA_dsD}) and the constant $L_0=\sum_j (1+v_j)\gamma_j=2\sum_{j\in\setV}\gamma_j$ ensures trace preservation.

Let $\lambda_k$ be the complex eigenvalues of $\dsL$ (single-particle eigenvalues) ordered by real part, $\Re\lambda_1\geq\Re\lambda_2\geq\cdots\geq \Re\lambda_N$. 
We diagonalize $\dsL$ as $\dsL=\dsVR \tL \dsVL^\dagger$, where $\tL=\mathrm{diag}(\lambda_1,\dots,\lambda_N)$ and $\dsVR$ and $\dsVL$ are matrices whose columns are, respectively, the right and left eigenvectors of $\dsL$ (which are distinct because $\dsL$ is non-Hermitian). It follows that 
\begin{equation}
\mathcal{L}=\sum_{i}\lambda_i g_i'g_i-L_0,
\end{equation}
where $\vec{g}=\dsVL^\dagger \vec{f}$ and $\vec{g'}=\vec{f}^\dagger \dsVR$. Similarly to the previous section, $g_j$ and $g_j'$ are NMMs: $g_j$ is an annihilation operator and $g'_j$ is a creation operator that satisfy the canonical commutation relations $\{g_i,g_j'\}=\delta_{ij}$ because of the biorthogonality of the left and right eigenvectors ($\dsVR\dsVL^\dagger=\id$), but $g'_j\neq g_j^\dagger$. Using again the biorthogonality of the eigenvectors, 
\begin{equation}
\scN_f=\sum_j f^\dagger_j f_j=\sum_j g'_j g_j=:\scN_g,
\end{equation}
where $\scN_g$ is the number operator of the NMMs. The change of basis is thus number conserving, as one would expect from the fact that $g_j$ and $g'_j$ do not mix creation and annihilation operators. In particular, the adjoint-fermion vacuum $|0\rangle_{\bm{w},\bm{\scw}}$ is also the vacuum of the NMMs.
The many-body eigenvalues and eigenstates of $\scL$ are then given by $\scL|\bm{n}\rangle=\Lambda_{\bm{n}} |\bm{n}\rangle$. Here,
$|\bm{n}\rangle=\prod_k (g'_k)^{n_k}|0\rangle_{\bm{w},\bm{\scw}}$, where $n_k\in\{0,1\}$ are the eigenvalues of $\scN_g$. The many-body eigenvalues are explicitly given in terms of single-body eigenvalues through 
$\Lambda_{\bm{n}}=\sum_{k} \lambda_k n_k-L_0$. Contrary to the standard situation~\cite{prosen2008}, while all physical many-body states satisfy $\Re\Lambda_{\bm{n}}\leq 0$, the single-particle $\lambda_k$ can have a positive real part if we are not in a steady-state sector. The adjoint-fermion number $n=\sum_k n_k$ is conserved throughout the evolution and is subject to the parity constraint $(-1)^{n}=(-1)^{|\setV|}$. Because of this, we can consider every subsector of fixed $n$ independently and compute a spectral gap for each, which gives 
\begin{equation}
\label{eq:gap_n}
    \Delta_n=L_0-\sum_{k=1}^n \Re\lambda_k.
\end{equation}

\subsection{Analytical bounds on the spectral gap and relaxation rates}
\label{sec:bounds}

The eigenvalues of $\dsL$ and, correspondingly, the spectral gap can, in general, only be determined numerically. 
Nevertheless, we can bound the gap using Bendixson inequalities~\cite{bendixson1902,mirskylinear}, which assert that, for an arbitrary real matrix $M$, the real part of any eigenvalue $z$ is bounded by $x_m\leq \Re z \leq x_M$, where $x_m$ and $x_M$ are the smallest and largest (real) eigenvalues of the symmetric part $(M+M^T)/2$, respectively. For completeness, a simple proof is given in Appendix~\ref{app:Bendixson}. In our case, the bounds are straightforward to establish since the symmetric part of $\dsL$ is diagonal and, hence, its eigenvalues coincide with its entries. 
Accordingly, the spectral gap for $\setV\neq \emptyset$ is bounded by
\begin{equation}
\label{eq:bounds_flips}
\begin{split}
    \max\Big\{0,2\Big(\sum_{j\in\setV}\gamma_j-n&\max_{j\in\setV}\gamma_j\Big)\Big\}
    \leq\Delta_n
    \\
    &\leq2\Big(\sum_{j\in\setV} \gamma_j+n\max_{j\in\setV}\gamma_j\Big).
\end{split}
\end{equation}
For $\setV=\emptyset$, because $(-1)^n=1$ and $L_0=0$, we have, with $n$ even,
\begin{equation}
\label{eq:bounds_noflips}
    2n\min_{j\in\setV}\gamma_j\leq \Delta_n\leq 2n\max_{j\in\setV}\gamma_j.
\end{equation}
The bounds (\ref{eq:bounds_flips}) and (\ref{eq:bounds_noflips}) are independent of any details of the graph on which the system lives. If the dissipation is spatially homogeneous, $\gamma_j=\gamma$, then the bound (\ref{eq:bounds_noflips}) is tight, $\Delta=2n\gamma$ ($2\gamma$ per excited fermion). For $|\setV|>0$, the bound (\ref{eq:bounds_flips}) is also the tightest for homogeneous dissipation, $2(|\setV|-n)\gamma\leq \Delta\leq 2(|\setV|+n)\gamma$ (the lower bound should be replaced by $0$ when $n>|\setV|$). 

In Sec.~\ref{sec:numerics_parity}, we discussed the presence of anomalous relaxation (i.e., a finite gap when the thermodynamic limit is taken before the dissipationless limit) in sectors with $\setU\neq\emptyset$ and its absence in sectors with $\setU=\emptyset$. We now revisit the latter in light of our analytical bounds, as anomalous relaxation is justified by the noncommutativity of the thermodynamic ($N\to\infty$) and dissipationless ($\gamma\to0$) limits, which is ruled out by our bounds. Indeed, by Eq.~(\ref{eq:bounds_flips}), $\Delta$ is bounded from above by an expression that is independent of $N$ (crucially, the sum is over the sites with interlayer gauge flips, not all sites). Hence, taking the limit $\gamma_j\to0$ for all $j$, takes the upper bound to zero (even in the thermodynamic limit), and the spectral gap closes at $\gamma_j=0$.

On the other hand, the lower bound on $\Delta$ as $\gamma_j\to\infty$ is consistent with the quantum Zeno effect. For each sector with $\setU=\emptyset$, there is always a conserved operator in the limit $\gamma\to\infty$, hence the gap goes to zero in this limit.\footnote{Note that the heuristic argument outlined in Sec.~\ref{sec:numerics_parity} is not operative here, as the two limits $N\to\infty$ and $\gamma\to\infty$ have the same effect on the effective dissipation strength $\gamma_{\mathrm{eff}}$.}
Indeed, for any $\setV\neq\emptyset$, consider the operator $\sigma = \rho_{\bm{w}}\prod_{j\in\setV} \Gamma_j^\ch$ (for any steady state $\rho_{\bm{w}}$), which belongs to the sector with flux configuration $\bm{w}$, interlayer gauge flips $\setV$, and intralayer gauge flips $\setU=\emptyset$. $\sigma$ has zero eigenvalue under the dissipative part of $\scL$ (which is the only contribution for $\gamma\to\infty$), thus proving the quantum Zeno effect for $\sigma$. 
We now turn to the lower bound on $\Delta$ given by Eq.~(\ref{eq:bounds_flips}). When we apply the operator $\Gamma_j^\ch$ to $\rho_{\bm{w}}$, we flip $v_j$ but also add a Majorana $c_j$ (we can see $\Gamma_j^\ch$ as a special case of an open string of length zero that has one chiral end and one $\mu=0$ end; cf.\ Fig.~\ref{fig:strings}). Hence, $n=|\setV|$ and, since $\sum_{j\in\setV}\gamma_j-|\setV|\max_{j\in\setV}\gamma_j<0$, the lower bound is zero, consistently with the quantum Zeno effect.

\begin{figure*}[t]
    \centering
    \includegraphics[width=0.8\textwidth]{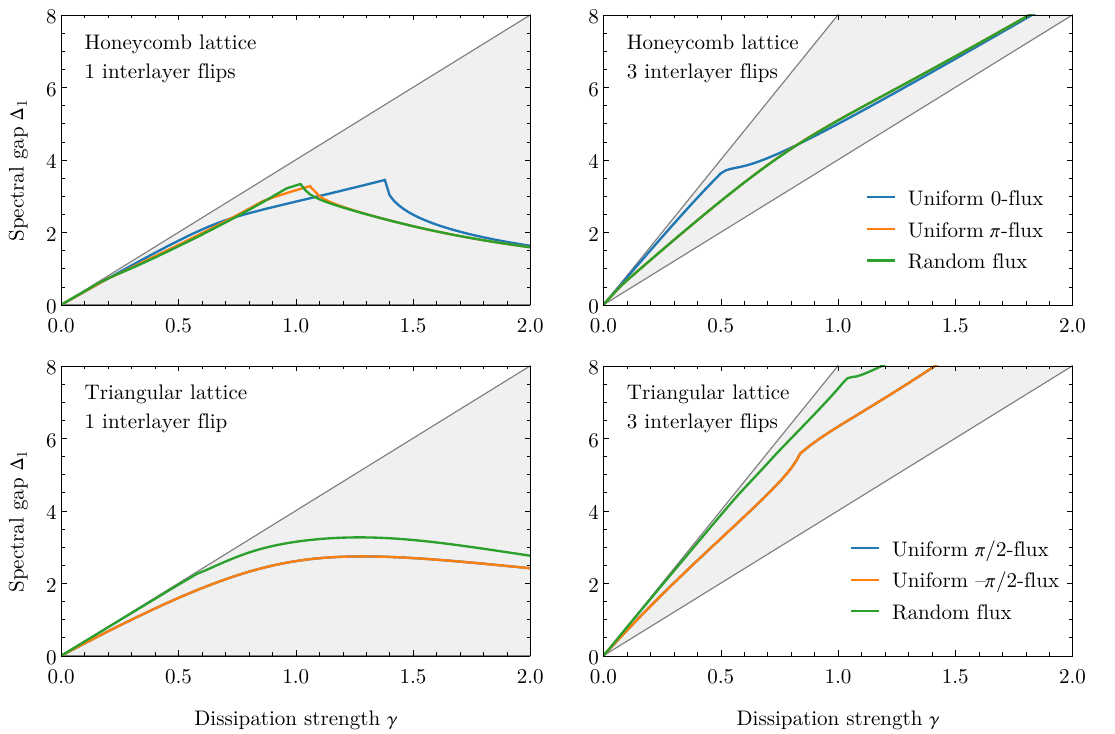}
    \caption{Spectral gap as a function of dissipation strength in fermion-number-conserving sectors.
    We computed the gap for $n=1$ on the honeycomb lattice with $N=512$ [for which the gap has converged, cf.\ Fig.~\ref{fig:intragapHoney}(d)] (top row) and the triangular lattice with $N=400$ (bottom row) for periodic boundary conditions on the torus (with the same number of sites along each direction) using the prescription of Eq.~(\ref{eq:gap_n}), with the single-particle eigenvalues $\lambda_k$ obtained from exact diagonalization of $\dsL=\dsA+\dsD$. We set uniform Hamiltonian couplings $J_{ij}=1$ and dissipation $\gamma_i=\gamma$, and fixed the gauge fields to $u_{ij}=\ut_{ij}$ (i.e., $\setU=\emptyset$). In each panel we considered different flux configurations $u_{ij}$ (colored lines), namely, uniform $0$-flux ($w_p=w_\ell=1$ for all $p$ and $\ell$), uniform $\pi$-flux ($w_p=w_\ell=-1$ for all $p$ and $\ell$), and random flux configurations for the honeycomb lattice (upper row); and uniform $\pi/2$-flux ($w_p=w_\ell=\i$ for all $p$ and $\ell$), uniform $-\pi/2$-flux ($w_p=w_\ell=-\i$ for all $p$ and $\ell$), and random flux configurations for the triangular lattice (lower row). The left and right panels correspond to one and three interlayer gauge flips, respectively (the fermion parity $(-1)^{n}=-1$ restricts us to an odd number of gauge flips). The shaded grey area gives the analytical bounds (\ref{eq:bounds_flips}). In all cases, the spectral gap closes in the dissipationless limit (absence of anomalous dissipation) and for $n=|\setV|=1$ it also closes in the infinite-dissipation limit (quantum Zeno effect).
    }
    \label{fig:intergapUniform}
\end{figure*}

\subsection{Numerical results}
\label{sec:numerics_number}

In Fig.~\ref{fig:intergapUniform}, we plot the spectral gap $\Delta_1$ as a function of dissipation strength $\gamma$ for the honeycomb and triangular lattices. For each lattice, we considered uniform Hamiltonian couplings, $J_{ij}=1$, and dissipation, $\gamma_j=\gamma$, and several choices of flux configurations, namely, uniform $0$-flux, uniform $\pi$-flux, and random flux configurations for the honeycomb lattice (which is bipartite), and uniform $\pm\pi/2$-flux and random flux configurations for the triangular lattice (which is not). For each case, we show both the numerically computed gap (full colored curves) and the analytical bounds established above (shaded grey region). As expected, the numerical results are always inside the analytical bounds, but there can be a large variability for different flux configurations. Qualitatively similar results are found for disordered systems (see Appendix~\ref{app:additional_numerics}).

While the spectral gap for a given flux configuration can be easily computed, the question of which flux configuration minimizes $\Delta$ for fixed $\setV$ remains. 
Given the exponential number of flux configurations, a numerical search among all of them is intractable. In the Hamiltonian system, a theorem by Lieb~\cite{lieb1994} states that the ground state (i.e., the state with the lowest real energy) belongs to the homogeneous $0$-flux sector on the honeycomb lattice and to the homogeneous $\pi$-flux sector on the square lattice. Unfortunately, no similar result holds in the dissipative case. This can be seen, for instance, in the upper left panel of Fig.~\ref{fig:intergapUniform}, where the gap belongs to different flux sectors for different values of $\gamma$.

Finally, we comment on the degeneracy of the gap for the triangular lattice, see the bottom row of Fig.~\ref{fig:intergapUniform}, exemplifying a general feature of nonbipartite graphs~\cite{kitaev2006AnnPhys}. These graphs 
have plaquettes with an odd-length boundary and imaginary flux $w_p=\pm\i$, cf.\ Eq.~(\ref{eq:W_u}). Let $C$ be the antiunitary operator defined by $C\Gamma_j^\mu C^{-1}=-\Gamma_j^\mu$ and $C \i C^{-1}=-\i$ (this operator can always be constructed from an appropriate product of gamma matrices and the complex conjugation operator). Because the Lindbladian and the flux operators are always real sums of products of an even number of gamma matrices, $[C,\scL]=[W_p,C]=0$. Consequently, $W_p(C|\bm{n}\rangle)=CW_p|\bm{n}\rangle=w_p^* (C|\bm{n}\rangle)$ and, similarly, $\scL(C|\bm{n}\rangle)=C\scL|\bm{n}\rangle=\Lambda_{\bm{n}}^* (C|\bm{n}\rangle)$. The states $|\bm{n}\rangle$ and $C|\bm{n}\rangle$ have the same relaxation rate, as their eigenvalues have the same real part. If the graph is bipartite, all plaquettes have an even-length boundary, $w_p^*=w_p$, and the two eigenstates belong to the same flux sector; different flux sectors are unrelated. On the other hand, if the graph is nonbipartite, there are some plaquettes with odd-length boundaries for which $w_p^*=-w_p$, and the two sectors related by flux flips have degenerate relaxation rates. Analogous considerations hold for the fluxes through noncontractible loops.

\section{Conclusions}
\label{sec:conclusions}

We have investigated an exactly solvable dissipative model of interacting spins, focusing on the symmetry and topological properties of the steady states and the relaxation dynamics toward them. With a judicious choice of jump operators, the model is mapped to free Majorana fermions hopping on a graph in the presence of a background $\mathbb{Z}_2$ gauge field. 
The Lindbladian of the system has an exponentially large number of steady states indexed by conserved-flux configurations. 
On topologically nontrivial surfaces, noncontractible fluxes encode classical information but exhibit mixed state topological order, in the sense of SW-SSB of a one-form symmetry. 

We have analyzed the possible excitations of the model, which correspond to Majorana fermions or gauge flips carried by the endpoints of open strings. Fermionic vacua without gauge excitations correspond to the steady states, while vacua with gauge excitations are mapped to decaying states in different symmetry sectors. Having identified the vacua, we then populate them with Majorana fermions to obtain the remaining transient states and the relaxation times of such excitations. When intralayer gauge flips are present, only the parity of adjoint fermions is conserved, with anomalous relaxation emerging in the thermodynamic limit. Without intralayer gauge flips, the adjoint-fermion number is conserved, and we can analytically bound the relevant relaxation rates.

As in the closed Kitaev honeycomb model, our model combines exact solvability with nontrivial topological features. This allowed us to efficiently explore both the symmetry properties of steady states and the relaxation dynamics toward them. However, an important open question remains: how do the topological properties manifest dynamically? The free-fermion representation of the model should enable an efficient characterization of this process. Specifically, starting from a pure state quantum memory, we expect that the SW-SSB occurs sharply at a critical time---before which the decoherence is reversible, and after which the ability to encode quantum information is irreversibly lost. Further investigation is also needed to clarify the relation between this dynamical transition and anomalous relaxation, where the Lindbladian remains gapped at all dissipation strengths. The dissipative gamma-matrix model offers a promising avenue for addressing this question.

Moreover, while we found the relaxation dynamics to be qualitatively similar for different graphs, the degeneracy of the spectral gap in flux sectors related by antiunitary symmetry shows that there can be subtle differences in systems belonging to different symmetry classes (in this case, bipartite versus nonbipartite graphs); similar fine differences exist between bipartite and nonbipartite monitored Kitaev circuits~\cite{klocke2024arXiv}. The non-Hermitian symmetry classification~\cite{bernard2002,kawabata2019PRX,zhou2019PRB} of Lindbladians was recently completed for many-body~\cite{sa2023PRX,kawabata2023PRXQ,garcia2024PRD} and single-particle~\cite{lieu2020PRL,altland2021PRX,kawasaki2022PRB} systems. In this framework, an exhaustive symmetry classification of the gamma-matrix Lindbladian and the dynamical consequences of different symmetries is an interesting direction for future studies.

Finally, throughout the paper, we assumed that dissipation acts on all sites of the graph. However, this is not required by our construction, and we can set $\gamma_j\neq0$ only on a subset of graph edges. A more nontrivial construction exploits the existence of ``dangling edges'' without interactions in the graph. If we add dissipation only to sites with such dangling edges, we do not need to enlarge the local Hilbert space (i.e., consider higher-dimensional gamma matrices) to include dissipation. In particular, in this way we can study the Lindbladian dissipation of spin-$1/2$ models instead of only the more exotic gamma-matrix models. In a system with boundaries, if we only introduce dissipation at sites at the boundary that are connected to fewer than three neighbors, we can use the ``missing'' Pauli matrix going out of these sites as a jump operator. More generically, we can study open versions of any closed Kitaev-like model on an arbitrary graph by cutting certain edges $\langle ij\rangle$ and replacing the coherent interactions on them by incoherent interactions on the two endpoints of the opened edge. 
These additional scenarios preserve the color rule and free-fermion reducibility and, therefore, can be handled with the formalism we developed in this paper. As the ensuing phenomenology can be quite rich and different from that of bulk dissipated models, we defer a detailed analysis of this interesting problem to future work.

\begin{acknowledgments}
This work was supported by a Research Fellowship from the Royal Commission for the Exhibition of 1851 (L.~S.) and EPSRC Grant No.\ EP/V062654/1.
\end{acknowledgments}

\appendix
\renewcommand{\appendixname}{APPENDIX}

\section{\uppercase{Nontrivial steady states: no-go result}}
\label{app:nontrivial_SS}

In this appendix, we show that, although there is considerable freedom in choosing the Hamiltonian contribution to the dynamics (i.e., the construction holds for gamma-matrix models respecting the color rule on arbitrary graphs), free-fermion reducibility constrains the dissipative processes to ``generalized dephasing'' with fully-mixed, infinite-temperature steady states.

We start by noticing that nontrivial steady states follow from at least one jump operator $L$ that satisfies $L^\dagger\neq \eta L$ for any complex $\eta$ (this is a necessary but not sufficient condition and the discussion can be straightforwardly generalized to more than one jump operator). Let us assume the existence of such a jump operator. Since the gamma matrices are Hermitian, $L$ must be a linear superposition of gamma matrices with complex coefficients, $L=\sum_{j,\mu}\ell_{j,\mu}\Gamma_j^\mu$. Following the same steps as in Sec.~\ref{sec:free-fermion}, we can write the term $L\rho L^\dagger$ as
\begin{equation}
    L\rho L^\dagger \to \sum_{ij,\mu\nu}\ell_{i,\mu}\ell^*_{j,\nu}v_{ij}^{\mu\nu}\i c_i \ct_j|\rho\rangle,
\end{equation}
where, analogously to before, $v_{ij}^{\mu\nu}=\i b^\mu_i \bt_j^\nu$. 
Two operators $v_{ij}^{\mu\nu}$ \emph{anticommute} whenever they share an odd number of indices (they also anticommute with certain $u_{ij}$ and $\ut_{ij}$ that arise in the Hamiltonian contribution). Since the full set of edge operators $\{u_{ij},\ut_{ij},v_{ij}^{\mu\nu}\}$ does not commute with itself, the operators cannot be replaced by their eigenvalues $\pm1$, and the fermionized Lindbladian does not split into individual free-fermion sectors. This only occurs when each $L$ is a single gamma matrix, whence a trivial steady state follows.

\begin{figure*}[t]
    \centering
    \includegraphics[width=0.8\textwidth]{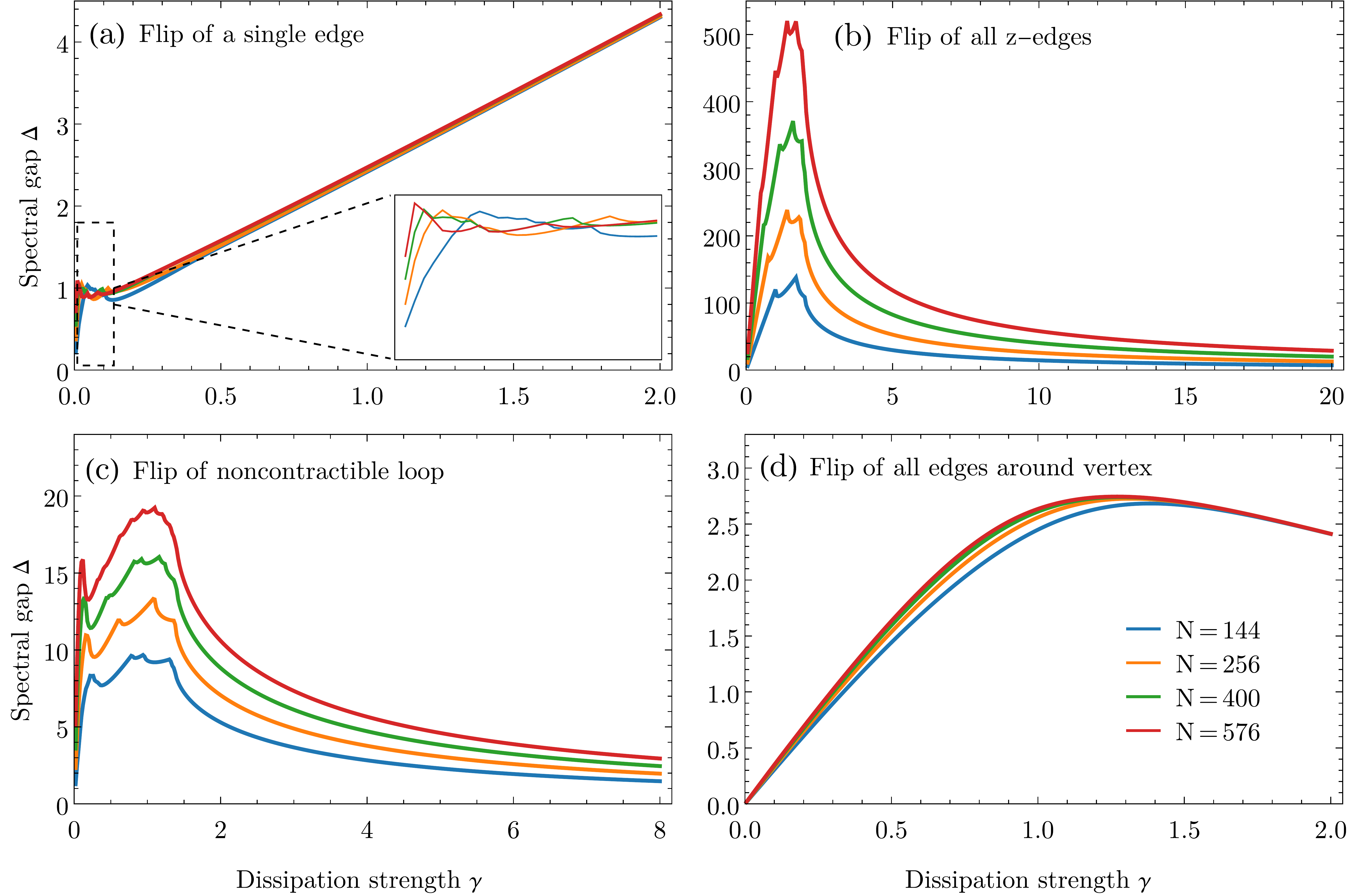}
    \caption{Additional numerical results for the spectral gap as a function of dissipation strength in sectors with fermion-parity conservation.
    The gap was computed for the triangular lattice with $N=144$--$576$ and periodic boundary conditions on the torus, using the prescription of Sec.~\ref{sec:parity_gap}. Different panels correspond to the different choices of intralayer gauge flips (i.e., of $\setU$) described in the caption of Fig.~\ref{fig:intragapHoney}.}
    \label{fig:intragapTri}
\end{figure*}

\begin{figure*}[t]
    \centering
    \includegraphics[width=0.8\textwidth]{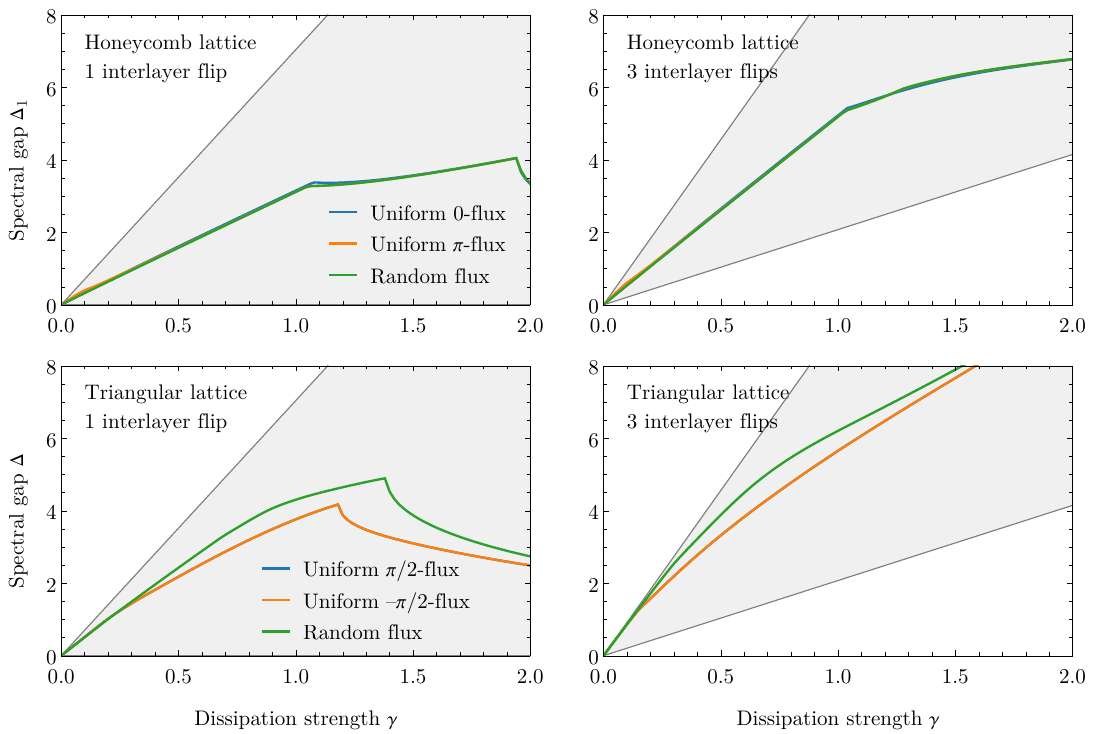}
    \caption{Additional numerical results for the spectral gap as a function of dissipation strength in sectors with fermion-number conservation.`
    We computed the gap for $n=1$ on the honeycomb lattice with $N=512$ (top row) and triangular lattice with $N=400$ (bottom row) for periodic boundary conditions on the torus, using the prescription of Eq.~(\ref{eq:gap_n}). We set $\gamma_j=\gamma X_j$ with $X_j$ independently sampled from a uniform distribution in $[0, 2]$ for each site $j$ and disordered couplings $J_{ij}$ independently sampled from $[-2,2]$. The left panels correspond to a single interlayer gauge flip ($|\setV|=1$) and the right panels to three flips ($|\setV|=3$). In each panel, we show the results for the different background flux configurations described in the caption of Fig.~\ref{fig:intergapUniform}. The shaded gray area gives the analytical bounds (\ref{eq:bounds_flips}), which, as expected, are less tight than in the uniform case.}
    \label{fig:intergapRandom}
\end{figure*}

\section{\uppercase{Counting of independent fluxes}}
\label{app:fluxes}

As mentioned in the main text, not all fluxes $W_\scC$ and $\scW^\ch_\scP$ are independent. In this appendix, we enumerate the independent ones for a graph $G$ with $N$ vertices, $E$ edges, and $F$ faces, which we embed on an arbitrary orientable surface $S$ with genus $g$ and $B$ boundaries in such a way that no edges of $G$ cross (such a surface always exists for sufficiently high genus). For the counting, we need to relate the graph data ($N$, $E$, $F$) to the surface data ($g$, $B$), which can be done through the Euler characteristic $\chi=N-E+F=1-b_1+b_2=2-2g-B$. Here, the first Betti number $b_1=2g+B+b_2-1$ counts the number of noncontractible loops on $S$, while the second Betti number $b_2=1$ if $B=0$ and vanishes otherwise. 

First, we consider the fluxes through the elementary plaquettes $p$ of the graph ($W_p$ and $\Wt_p$), of which there are $F$ in each copy. Additionally, for each copy, we can have $b_1$ noncontractible loops $\ell$ and the associated fluxes $W_{\ell}$ and $\Wt_{\ell}$. There is also a weak flux $\scW_q$ through each interlayer plaquette $q$, of which there is one for each of the $E$ edges. In total, we thus have $2(F+b_1)+E$ elementary fluxes.

The flux through any closed surface is $+1$ because each gauge field, which squares to $+1$, appears twice in the definition of the flux through the surface. Then, if $S$ is closed (i.e., $b_2=1$), the total flux through the surface (i.e., $\prod_p W_p=+1$) is fixed. Moreover, the flux through each of the $F$ interlayer cells (i.e., $W_p\Wt_p \prod_{q\in\partial p}\scW_q=+1$) fixes all the fluxes $\Wt_p$ since there is one $\Wt_p$ per interlayer cell. (The flux $\prod_p \Wt_p$ is already fixed by the previous conditions and, therefore, does not amount to an additional constraint.) Finally, $W_{\ell}\Wt_{\ell}=\prod_q \scW_q$, where the product is taken over interlayer plaquettes $q$ that connect the noncontractible loops $\ell$ in the two copies; this fixes all fluxes $\Wt_\ell$ and gives $b_1$ further constraints.

In total, we thus have $F+b_1+b_2$ constraints, and the number of independent fluxes (gauge-invariant degrees of freedom) is $F+E+b_1-b_2$: $F-b_2$ plaquette fluxes $W_p$, $b_1$ noncontractible-loop fluxes $W_{\ell}$, and $E$ interlayer weak fluxes $\scW_q$, in agreement with Eq.~(\ref{eq:Ldecomp_fluxes}). 
Using the Euler characteristic, the number of independent fluxes can be written as $2E-N+1$.

\section{\uppercase{Additional numerical results}}
\label{app:additional_numerics}

In this appendix, we present additional numerical results to substantiate the claims that the spectral gap as a function of dissipation strength is qualitatively the same for (i) the honeycomb (Fig.~\ref{fig:intragapHoney} in the main text) and triangular (Fig.~\ref{fig:intragapTri} in this appendix) lattices in symmetry sectors with fermion-parity conservation; and (ii) the honeycomb and triangular lattices with uniform (Fig.~\ref{fig:intergapUniform} in the main text) and disordered (Fig.~\ref{fig:intergapRandom} in this appendix) Hamiltonian and dissipative couplings in symmetry sectors with fermion-number conservation.

\section{\uppercase{Bendixson inequalities}}
\label{app:Bendixson}

\begin{theorem*}[Bendixson inequalities~\cite{bendixson1902}]
Let $M$ be an arbitrary real matrix, $z$ any one of its eigenvalues, and $x_m$ and $x_M$, respectively, the smallest and largest eigenvalues of $(M+M^T)/2$, where ${M}^T$ denotes the transpose of $M$. Then,
\begin{equation}
    x_m\leq \Re z \leq x_M.
\end{equation}
\end{theorem*}

\begin{proof}
Let $B=M-x_m \id$ and $v$ be one of its eigenvectors with eigenvalue $y=z-x_m$, $Bv=yv$. We thus have to prove that $\Re y\geq 0$. Since $(B+B^T)/2$ is positive semidefinite by construction, we have
\begin{equation}
    0\leq \langle v, \frac{B+B^T}{2}v\rangle 
    =\frac{1}{2}\langle v, Bv\rangle + \frac{1}{2}\langle Bv, v \rangle
    =\Re y \ \langle v,v \rangle.
\end{equation}
Because $\langle v,v \rangle \geq 0$, we arrive at the desired result. The upper bound on $\Re y$ can be established in the same way by considering the matrix $C=x_M \id - M$.
\end{proof}

\bibliography{bib}

\begin{thebibliography}{160}%
\makeatletter
\providecommand \@ifxundefined [1]{%
 \@ifx{#1\undefined}
}%
\providecommand \@ifnum [1]{%
 \ifnum #1\expandafter \@firstoftwo
 \else \expandafter \@secondoftwo
 \fi
}%
\providecommand \@ifx [1]{%
 \ifx #1\expandafter \@firstoftwo
 \else \expandafter \@secondoftwo
 \fi
}%
\providecommand \natexlab [1]{#1}%
\providecommand \enquote  [1]{``#1''}%
\providecommand \bibnamefont  [1]{#1}%
\providecommand \bibfnamefont [1]{#1}%
\providecommand \citenamefont [1]{#1}%
\providecommand \href@noop [0]{\@secondoftwo}%
\providecommand \href [0]{\begingroup \@sanitize@url \@href}%
\providecommand \@href[1]{\@@startlink{#1}\@@href}%
\providecommand \@@href[1]{\endgroup#1\@@endlink}%
\providecommand \@sanitize@url [0]{\catcode `\\12\catcode `\$12\catcode `\&12\catcode `\#12\catcode `\^12\catcode `\_12\catcode `\%12\relax}%
\providecommand \@@startlink[1]{}%
\providecommand \@@endlink[0]{}%
\providecommand \url  [0]{\begingroup\@sanitize@url \@url }%
\providecommand \@url [1]{\endgroup\@href {#1}{\urlprefix }}%
\providecommand \urlprefix  [0]{URL }%
\providecommand \Eprint [0]{\href }%
\providecommand \doibase [0]{https://doi.org/}%
\providecommand \selectlanguage [0]{\@gobble}%
\providecommand \bibinfo  [0]{\@secondoftwo}%
\providecommand \bibfield  [0]{\@secondoftwo}%
\providecommand \translation [1]{[#1]}%
\providecommand \BibitemOpen [0]{}%
\providecommand \bibitemStop [0]{}%
\providecommand \bibitemNoStop [0]{.\EOS\space}%
\providecommand \EOS [0]{\spacefactor3000\relax}%
\providecommand \BibitemShut  [1]{\csname bibitem#1\endcsname}%
\let\auto@bib@innerbib\@empty
\bibitem [{\citenamefont {Anderson}(1973)}]{anderson1973}%
  \BibitemOpen
  \bibfield  {author} {\bibinfo {author} {\bibfnamefont {P.~W.}\ \bibnamefont {Anderson}},\ }\bibfield  {title} {\bibinfo {title} {{Resonating valence bonds: A new kind of insulator?}},\ }\href {https://doi.org/https://doi.org/10.1016/0025-5408(73)90167-0} {\bibfield  {journal} {\bibinfo  {journal} {Mat. Res. Bull.}\ }\textbf {\bibinfo {volume} {8}},\ \bibinfo {pages} {153} (\bibinfo {year} {1973})}\BibitemShut {NoStop}%
\bibitem [{\citenamefont {Anderson}(1987)}]{anderson1987}%
  \BibitemOpen
  \bibfield  {author} {\bibinfo {author} {\bibfnamefont {P.~W.}\ \bibnamefont {Anderson}},\ }\bibfield  {title} {\bibinfo {title} {{The Resonating Valence Bond State in La$_2$CuO$_4$ and Superconductivity}},\ }\href {https://doi.org/10.1126/science.235.4793.1196} {\bibfield  {journal} {\bibinfo  {journal} {Science}\ }\textbf {\bibinfo {volume} {235}},\ \bibinfo {pages} {1196} (\bibinfo {year} {1987})}\BibitemShut {NoStop}%
\bibitem [{\citenamefont {Kalmeyer}\ and\ \citenamefont {Laughlin}(1987)}]{kalmeyer1987}%
  \BibitemOpen
  \bibfield  {author} {\bibinfo {author} {\bibfnamefont {V.}~\bibnamefont {Kalmeyer}}\ and\ \bibinfo {author} {\bibfnamefont {R.~B.}\ \bibnamefont {Laughlin}},\ }\bibfield  {title} {\bibinfo {title} {{Equivalence of the resonating-valence-bond and fractional quantum Hall states}},\ }\href {https://doi.org/10.1103/PhysRevLett.59.2095} {\bibfield  {journal} {\bibinfo  {journal} {Phys. Rev. Lett.}\ }\textbf {\bibinfo {volume} {59}},\ \bibinfo {pages} {2095} (\bibinfo {year} {1987})}\BibitemShut {NoStop}%
\bibitem [{\citenamefont {Rokhsar}\ and\ \citenamefont {Kivelson}(1988)}]{rokhsar1988}%
  \BibitemOpen
  \bibfield  {author} {\bibinfo {author} {\bibfnamefont {D.~S.}\ \bibnamefont {Rokhsar}}\ and\ \bibinfo {author} {\bibfnamefont {S.~A.}\ \bibnamefont {Kivelson}},\ }\bibfield  {title} {\bibinfo {title} {{Superconductivity and the Quantum Hard-Core Dimer Gas}},\ }\href {https://doi.org/10.1103/PhysRevLett.61.2376} {\bibfield  {journal} {\bibinfo  {journal} {Phys. Rev. Lett.}\ }\textbf {\bibinfo {volume} {61}},\ \bibinfo {pages} {2376} (\bibinfo {year} {1988})}\BibitemShut {NoStop}%
\bibitem [{\citenamefont {Wen}(1991)}]{wen1991}%
  \BibitemOpen
  \bibfield  {author} {\bibinfo {author} {\bibfnamefont {X.~G.}\ \bibnamefont {Wen}},\ }\bibfield  {title} {\bibinfo {title} {{Mean-field theory of spin-liquid states with finite energy gap and topological orders}},\ }\href {https://doi.org/10.1103/PhysRevB.44.2664} {\bibfield  {journal} {\bibinfo  {journal} {Phys. Rev. B}\ }\textbf {\bibinfo {volume} {44}},\ \bibinfo {pages} {2664} (\bibinfo {year} {1991})}\BibitemShut {NoStop}%
\bibitem [{\citenamefont {Senthil}\ and\ \citenamefont {Fisher}(2000)}]{senthil2000}%
  \BibitemOpen
  \bibfield  {author} {\bibinfo {author} {\bibfnamefont {T.}~\bibnamefont {Senthil}}\ and\ \bibinfo {author} {\bibfnamefont {M.~P.~A.}\ \bibnamefont {Fisher}},\ }\bibfield  {title} {\bibinfo {title} {{${Z}_{2}$ gauge theory of electron fractionalization in strongly correlated systems}},\ }\href {https://doi.org/10.1103/PhysRevB.62.7850} {\bibfield  {journal} {\bibinfo  {journal} {Phys. Rev. B}\ }\textbf {\bibinfo {volume} {62}},\ \bibinfo {pages} {7850} (\bibinfo {year} {2000})}\BibitemShut {NoStop}%
\bibitem [{\citenamefont {Moessner}\ and\ \citenamefont {Sondhi}(2001)}]{moessner2001}%
  \BibitemOpen
  \bibfield  {author} {\bibinfo {author} {\bibfnamefont {R.}~\bibnamefont {Moessner}}\ and\ \bibinfo {author} {\bibfnamefont {S.~L.}\ \bibnamefont {Sondhi}},\ }\bibfield  {title} {\bibinfo {title} {{Resonating Valence Bond Phase in the Triangular Lattice Quantum Dimer Model}},\ }\href {https://doi.org/10.1103/PhysRevLett.86.1881} {\bibfield  {journal} {\bibinfo  {journal} {Phys. Rev. Lett.}\ }\textbf {\bibinfo {volume} {86}},\ \bibinfo {pages} {1881} (\bibinfo {year} {2001})}\BibitemShut {NoStop}%
\bibitem [{\citenamefont {Kitaev}(2006)}]{kitaev2006AnnPhys}%
  \BibitemOpen
  \bibfield  {author} {\bibinfo {author} {\bibfnamefont {A.}~\bibnamefont {Kitaev}},\ }\bibfield  {title} {\bibinfo {title} {{Anyons in an exactly solved model and beyond}},\ }\href {https://doi.org/https://doi.org/10.1016/j.aop.2005.10.005} {\bibfield  {journal} {\bibinfo  {journal} {Ann. Phys.}\ }\textbf {\bibinfo {volume} {321}},\ \bibinfo {pages} {2} (\bibinfo {year} {2006})}\BibitemShut {NoStop}%
\bibitem [{\citenamefont {Dennis}\ \emph {et~al.}(2002)\citenamefont {Dennis}, \citenamefont {Kitaev}, \citenamefont {Landahl},\ and\ \citenamefont {Preskill}}]{dennis2002JMP}%
  \BibitemOpen
  \bibfield  {author} {\bibinfo {author} {\bibfnamefont {E.}~\bibnamefont {Dennis}}, \bibinfo {author} {\bibfnamefont {A.}~\bibnamefont {Kitaev}}, \bibinfo {author} {\bibfnamefont {A.}~\bibnamefont {Landahl}},\ and\ \bibinfo {author} {\bibfnamefont {J.}~\bibnamefont {Preskill}},\ }\bibfield  {title} {\bibinfo {title} {{Topological quantum memory}},\ }\href {https://doi.org/10.1063/1.1499754} {\bibfield  {journal} {\bibinfo  {journal} {J. Math. Phys.}\ }\textbf {\bibinfo {volume} {43}},\ \bibinfo {pages} {4452} (\bibinfo {year} {2002})}\BibitemShut {NoStop}%
\bibitem [{\citenamefont {Kitaev}(2003)}]{kitaev2001Anyons}%
  \BibitemOpen
  \bibfield  {author} {\bibinfo {author} {\bibfnamefont {A.~Y.}\ \bibnamefont {Kitaev}},\ }\bibfield  {title} {\bibinfo {title} {{Fault-tolerant quantum computation by anyons}},\ }\href {https://doi.org/https://doi.org/10.1016/S0003-4916(02)00018-0} {\bibfield  {journal} {\bibinfo  {journal} {Ann. Phys.}\ }\textbf {\bibinfo {volume} {303}},\ \bibinfo {pages} {2} (\bibinfo {year} {2003})}\BibitemShut {NoStop}%
\bibitem [{\citenamefont {Nayak}\ \emph {et~al.}(2008)\citenamefont {Nayak}, \citenamefont {Simon}, \citenamefont {Stern}, \citenamefont {Freedman},\ and\ \citenamefont {Das~Sarma}}]{nayak2008RMP}%
  \BibitemOpen
  \bibfield  {author} {\bibinfo {author} {\bibfnamefont {C.}~\bibnamefont {Nayak}}, \bibinfo {author} {\bibfnamefont {S.~H.}\ \bibnamefont {Simon}}, \bibinfo {author} {\bibfnamefont {A.}~\bibnamefont {Stern}}, \bibinfo {author} {\bibfnamefont {M.}~\bibnamefont {Freedman}},\ and\ \bibinfo {author} {\bibfnamefont {S.}~\bibnamefont {Das~Sarma}},\ }\bibfield  {title} {\bibinfo {title} {{Non-Abelian anyons and topological quantum computation}},\ }\href {https://doi.org/10.1103/RevModPhys.80.1083} {\bibfield  {journal} {\bibinfo  {journal} {Rev. Mod. Phys.}\ }\textbf {\bibinfo {volume} {80}},\ \bibinfo {pages} {1083} (\bibinfo {year} {2008})}\BibitemShut {NoStop}%
\bibitem [{\citenamefont {Lee}(2008)}]{lee2008}%
  \BibitemOpen
  \bibfield  {author} {\bibinfo {author} {\bibfnamefont {P.~A.}\ \bibnamefont {Lee}},\ }\bibfield  {title} {\bibinfo {title} {{An End to the Drought of Quantum Spin Liquids}},\ }\href {https://doi.org/10.1126/science.1163196} {\bibfield  {journal} {\bibinfo  {journal} {Science}\ }\textbf {\bibinfo {volume} {321}},\ \bibinfo {pages} {1306} (\bibinfo {year} {2008})}\BibitemShut {NoStop}%
\bibitem [{\citenamefont {Savary}\ and\ \citenamefont {Balents}(2016)}]{savary2017RPP}%
  \BibitemOpen
  \bibfield  {author} {\bibinfo {author} {\bibfnamefont {L.}~\bibnamefont {Savary}}\ and\ \bibinfo {author} {\bibfnamefont {L.}~\bibnamefont {Balents}},\ }\bibfield  {title} {\bibinfo {title} {{Quantum spin liquids: a review}},\ }\href {https://doi.org/10.1088/0034-4885/80/1/016502} {\bibfield  {journal} {\bibinfo  {journal} {Rep. Prog. Phys.}\ }\textbf {\bibinfo {volume} {80}},\ \bibinfo {pages} {016502} (\bibinfo {year} {2016})}\BibitemShut {NoStop}%
\bibitem [{\citenamefont {Zhou}\ \emph {et~al.}(2017)\citenamefont {Zhou}, \citenamefont {Kanoda},\ and\ \citenamefont {Ng}}]{zhou2017RMP}%
  \BibitemOpen
  \bibfield  {author} {\bibinfo {author} {\bibfnamefont {Y.}~\bibnamefont {Zhou}}, \bibinfo {author} {\bibfnamefont {K.}~\bibnamefont {Kanoda}},\ and\ \bibinfo {author} {\bibfnamefont {T.-K.}\ \bibnamefont {Ng}},\ }\bibfield  {title} {\bibinfo {title} {Quantum spin liquid states},\ }\href {https://doi.org/10.1103/RevModPhys.89.025003} {\bibfield  {journal} {\bibinfo  {journal} {Rev. Mod. Phys.}\ }\textbf {\bibinfo {volume} {89}},\ \bibinfo {pages} {025003} (\bibinfo {year} {2017})}\BibitemShut {NoStop}%
\bibitem [{\citenamefont {Takagi}\ \emph {et~al.}(2019)\citenamefont {Takagi}, \citenamefont {Takayama}, \citenamefont {Jackeli}, \citenamefont {Khaliullin},\ and\ \citenamefont {Nagler}}]{takagi2019NatRev}%
  \BibitemOpen
  \bibfield  {author} {\bibinfo {author} {\bibfnamefont {H.}~\bibnamefont {Takagi}}, \bibinfo {author} {\bibfnamefont {T.}~\bibnamefont {Takayama}}, \bibinfo {author} {\bibfnamefont {G.}~\bibnamefont {Jackeli}}, \bibinfo {author} {\bibfnamefont {G.}~\bibnamefont {Khaliullin}},\ and\ \bibinfo {author} {\bibfnamefont {S.~E.}\ \bibnamefont {Nagler}},\ }\bibfield  {title} {\bibinfo {title} {{Concept and realization of Kitaev quantum spin liquids}},\ }\href {https://www.nature.com/articles/s42254-019-0038-2} {\bibfield  {journal} {\bibinfo  {journal} {Nat. Rev. Phys.}\ }\textbf {\bibinfo {volume} {1}},\ \bibinfo {pages} {264} (\bibinfo {year} {2019})}\BibitemShut {NoStop}%
\bibitem [{\citenamefont {Broholm}\ \emph {et~al.}(2020)\citenamefont {Broholm}, \citenamefont {Cava}, \citenamefont {Kivelson}, \citenamefont {Nocera}, \citenamefont {Norman},\ and\ \citenamefont {Senthil}}]{broholm2020Sci}%
  \BibitemOpen
  \bibfield  {author} {\bibinfo {author} {\bibfnamefont {C.}~\bibnamefont {Broholm}}, \bibinfo {author} {\bibfnamefont {R.~J.}\ \bibnamefont {Cava}}, \bibinfo {author} {\bibfnamefont {S.~A.}\ \bibnamefont {Kivelson}}, \bibinfo {author} {\bibfnamefont {D.~G.}\ \bibnamefont {Nocera}}, \bibinfo {author} {\bibfnamefont {M.~R.}\ \bibnamefont {Norman}},\ and\ \bibinfo {author} {\bibfnamefont {T.}~\bibnamefont {Senthil}},\ }\bibfield  {title} {\bibinfo {title} {Quantum spin liquids},\ }\href {https://doi.org/10.1126/science.aay0668} {\bibfield  {journal} {\bibinfo  {journal} {Science}\ }\textbf {\bibinfo {volume} {367}},\ \bibinfo {pages} {eaay0668} (\bibinfo {year} {2020})}\BibitemShut {NoStop}%
\bibitem [{\citenamefont {Chamorro}\ \emph {et~al.}(2020)\citenamefont {Chamorro}, \citenamefont {McQueen},\ and\ \citenamefont {Tran}}]{chamorro2020}%
  \BibitemOpen
  \bibfield  {author} {\bibinfo {author} {\bibfnamefont {J.~R.}\ \bibnamefont {Chamorro}}, \bibinfo {author} {\bibfnamefont {T.~M.}\ \bibnamefont {McQueen}},\ and\ \bibinfo {author} {\bibfnamefont {T.~T.}\ \bibnamefont {Tran}},\ }\bibfield  {title} {\bibinfo {title} {{Chemistry of Quantum Spin Liquids}},\ }\href {https://doi.org/10.1021/acs.chemrev.0c00641} {\bibfield  {journal} {\bibinfo  {journal} {Chem. Rev.}\ }\textbf {\bibinfo {volume} {121}},\ \bibinfo {pages} {2898} (\bibinfo {year} {2020})}\BibitemShut {NoStop}%
\bibitem [{\citenamefont {Clark}\ and\ \citenamefont {Abdeldaim}(2021)}]{clark2021AnnuRev}%
  \BibitemOpen
  \bibfield  {author} {\bibinfo {author} {\bibfnamefont {L.}~\bibnamefont {Clark}}\ and\ \bibinfo {author} {\bibfnamefont {A.~H.}\ \bibnamefont {Abdeldaim}},\ }\bibfield  {title} {\bibinfo {title} {{Quantum Spin Liquids from a Materials Perspective}},\ }\href {https://doi.org/https://doi.org/10.1146/annurev-matsci-080819-011453} {\bibfield  {journal} {\bibinfo  {journal} {Annu. Rev. Condens. Matter Phys.}\ }\textbf {\bibinfo {volume} {51}},\ \bibinfo {pages} {495} (\bibinfo {year} {2021})}\BibitemShut {NoStop}%
\bibitem [{\citenamefont {Trebst}\ and\ \citenamefont {Hickey}(2022)}]{trebst2022}%
  \BibitemOpen
  \bibfield  {author} {\bibinfo {author} {\bibfnamefont {S.}~\bibnamefont {Trebst}}\ and\ \bibinfo {author} {\bibfnamefont {C.}~\bibnamefont {Hickey}},\ }\bibfield  {title} {\bibinfo {title} {{Kitaev materials}},\ }\href {https://doi.org/https://doi.org/10.1016/j.physrep.2021.11.003} {\bibfield  {journal} {\bibinfo  {journal} {Phys. Rep.}\ }\textbf {\bibinfo {volume} {950}},\ \bibinfo {pages} {1} (\bibinfo {year} {2022})}\BibitemShut {NoStop}%
\bibitem [{\citenamefont {Khatua}\ \emph {et~al.}(2023)\citenamefont {Khatua}, \citenamefont {Sana}, \citenamefont {Zorko}, \citenamefont {Gomilšek}, \citenamefont {Sethupathi}, \citenamefont {Rao}, \citenamefont {Baenitz}, \citenamefont {Schmidt},\ and\ \citenamefont {Khuntia}}]{khatua2023PhysRep}%
  \BibitemOpen
  \bibfield  {author} {\bibinfo {author} {\bibfnamefont {J.}~\bibnamefont {Khatua}}, \bibinfo {author} {\bibfnamefont {B.}~\bibnamefont {Sana}}, \bibinfo {author} {\bibfnamefont {A.}~\bibnamefont {Zorko}}, \bibinfo {author} {\bibfnamefont {M.}~\bibnamefont {Gomilšek}}, \bibinfo {author} {\bibfnamefont {K.}~\bibnamefont {Sethupathi}}, \bibinfo {author} {\bibfnamefont {M.~R.}\ \bibnamefont {Rao}}, \bibinfo {author} {\bibfnamefont {M.}~\bibnamefont {Baenitz}}, \bibinfo {author} {\bibfnamefont {B.}~\bibnamefont {Schmidt}},\ and\ \bibinfo {author} {\bibfnamefont {P.}~\bibnamefont {Khuntia}},\ }\bibfield  {title} {\bibinfo {title} {{Experimental signatures of quantum and topological states in frustrated magnetism}},\ }\href {https://doi.org/https://doi.org/10.1016/j.physrep.2023.09.008} {\bibfield  {journal} {\bibinfo  {journal} {Phys. Rep.}\ }\textbf {\bibinfo {volume} {1041}},\ \bibinfo {pages} {1} (\bibinfo {year} {2023})}\BibitemShut {NoStop}%
\bibitem [{\citenamefont {Hermanns}\ \emph {et~al.}(2018)\citenamefont {Hermanns}, \citenamefont {Kimchi},\ and\ \citenamefont {Knolle}}]{hermanns2018AnuuRev}%
  \BibitemOpen
  \bibfield  {author} {\bibinfo {author} {\bibfnamefont {M.}~\bibnamefont {Hermanns}}, \bibinfo {author} {\bibfnamefont {I.}~\bibnamefont {Kimchi}},\ and\ \bibinfo {author} {\bibfnamefont {J.}~\bibnamefont {Knolle}},\ }\bibfield  {title} {\bibinfo {title} {{Physics of the Kitaev Model: Fractionalization, Dynamic Correlations, and Material Connections}},\ }\href {https://doi.org/https://doi.org/10.1146/annurev-conmatphys-033117-053934} {\bibfield  {journal} {\bibinfo  {journal} {Annu. Rev. Condens. Matter Phys.}\ }\textbf {\bibinfo {volume} {9}},\ \bibinfo {pages} {17} (\bibinfo {year} {2018})}\BibitemShut {NoStop}%
\bibitem [{\citenamefont {Knolle}\ and\ \citenamefont {Moessner}(2019)}]{knolle2019annurev}%
  \BibitemOpen
  \bibfield  {author} {\bibinfo {author} {\bibfnamefont {J.}~\bibnamefont {Knolle}}\ and\ \bibinfo {author} {\bibfnamefont {R.}~\bibnamefont {Moessner}},\ }\bibfield  {title} {\bibinfo {title} {{A Field Guide to Spin Liquids}},\ }\href {https://doi.org/https://doi.org/10.1146/annurev-conmatphys-031218-013401} {\bibfield  {journal} {\bibinfo  {journal} {Annu. Rev. Condens. Matter Phys.}\ }\textbf {\bibinfo {volume} {10}},\ \bibinfo {pages} {451} (\bibinfo {year} {2019})}\BibitemShut {NoStop}%
\bibitem [{\citenamefont {Knolle}\ \emph {et~al.}(2014)\citenamefont {Knolle}, \citenamefont {Kovrizhin}, \citenamefont {Chalker},\ and\ \citenamefont {Moessner}}]{knolle2014PRL}%
  \BibitemOpen
  \bibfield  {author} {\bibinfo {author} {\bibfnamefont {J.}~\bibnamefont {Knolle}}, \bibinfo {author} {\bibfnamefont {D.~L.}\ \bibnamefont {Kovrizhin}}, \bibinfo {author} {\bibfnamefont {J.~T.}\ \bibnamefont {Chalker}},\ and\ \bibinfo {author} {\bibfnamefont {R.}~\bibnamefont {Moessner}},\ }\bibfield  {title} {\bibinfo {title} {{Dynamics of a Two-Dimensional Quantum Spin Liquid: Signatures of Emergent Majorana Fermions and Fluxes}},\ }\href {https://doi.org/10.1103/PhysRevLett.112.207203} {\bibfield  {journal} {\bibinfo  {journal} {Phys. Rev. Lett.}\ }\textbf {\bibinfo {volume} {112}},\ \bibinfo {pages} {207203} (\bibinfo {year} {2014})}\BibitemShut {NoStop}%
\bibitem [{\citenamefont {Hal\'asz}\ \emph {et~al.}(2014)\citenamefont {Hal\'asz}, \citenamefont {Chalker},\ and\ \citenamefont {Moessner}}]{halasz2014PRB}%
  \BibitemOpen
  \bibfield  {author} {\bibinfo {author} {\bibfnamefont {G.~B.}\ \bibnamefont {Hal\'asz}}, \bibinfo {author} {\bibfnamefont {J.~T.}\ \bibnamefont {Chalker}},\ and\ \bibinfo {author} {\bibfnamefont {R.}~\bibnamefont {Moessner}},\ }\bibfield  {title} {\bibinfo {title} {{Doping a topological quantum spin liquid: Slow holes in the Kitaev honeycomb model}},\ }\href {https://doi.org/10.1103/PhysRevB.90.035145} {\bibfield  {journal} {\bibinfo  {journal} {Phys. Rev. B}\ }\textbf {\bibinfo {volume} {90}},\ \bibinfo {pages} {035145} (\bibinfo {year} {2014})}\BibitemShut {NoStop}%
\bibitem [{\citenamefont {Knolle}\ \emph {et~al.}(2015)\citenamefont {Knolle}, \citenamefont {Kovrizhin}, \citenamefont {Chalker},\ and\ \citenamefont {Moessner}}]{knolle2015PRB}%
  \BibitemOpen
  \bibfield  {author} {\bibinfo {author} {\bibfnamefont {J.}~\bibnamefont {Knolle}}, \bibinfo {author} {\bibfnamefont {D.~L.}\ \bibnamefont {Kovrizhin}}, \bibinfo {author} {\bibfnamefont {J.~T.}\ \bibnamefont {Chalker}},\ and\ \bibinfo {author} {\bibfnamefont {R.}~\bibnamefont {Moessner}},\ }\bibfield  {title} {\bibinfo {title} {{Dynamics of fractionalization in quantum spin liquids}},\ }\href {https://doi.org/10.1103/PhysRevB.92.115127} {\bibfield  {journal} {\bibinfo  {journal} {Phys. Rev. B}\ }\textbf {\bibinfo {volume} {92}},\ \bibinfo {pages} {115127} (\bibinfo {year} {2015})}\BibitemShut {NoStop}%
\bibitem [{\citenamefont {Nasu}\ \emph {et~al.}(2015)\citenamefont {Nasu}, \citenamefont {Udagawa},\ and\ \citenamefont {Motome}}]{nasu2015PRB}%
  \BibitemOpen
  \bibfield  {author} {\bibinfo {author} {\bibfnamefont {J.}~\bibnamefont {Nasu}}, \bibinfo {author} {\bibfnamefont {M.}~\bibnamefont {Udagawa}},\ and\ \bibinfo {author} {\bibfnamefont {Y.}~\bibnamefont {Motome}},\ }\bibfield  {title} {\bibinfo {title} {{Thermal fractionalization of quantum spins in a Kitaev model: Temperature-linear specific heat and coherent transport of Majorana fermions}},\ }\href {https://doi.org/10.1103/PhysRevB.92.115122} {\bibfield  {journal} {\bibinfo  {journal} {Phys. Rev. B}\ }\textbf {\bibinfo {volume} {92}},\ \bibinfo {pages} {115122} (\bibinfo {year} {2015})}\BibitemShut {NoStop}%
\bibitem [{\citenamefont {Nasu}\ \emph {et~al.}(2016)\citenamefont {Nasu}, \citenamefont {Knolle}, \citenamefont {Kovrizhin}, \citenamefont {Motome},\ and\ \citenamefont {Moessner}}]{nasu2016NatPhys}%
  \BibitemOpen
  \bibfield  {author} {\bibinfo {author} {\bibfnamefont {J.}~\bibnamefont {Nasu}}, \bibinfo {author} {\bibfnamefont {J.}~\bibnamefont {Knolle}}, \bibinfo {author} {\bibfnamefont {D.~L.}\ \bibnamefont {Kovrizhin}}, \bibinfo {author} {\bibfnamefont {Y.}~\bibnamefont {Motome}},\ and\ \bibinfo {author} {\bibfnamefont {R.}~\bibnamefont {Moessner}},\ }\bibfield  {title} {\bibinfo {title} {{Fermionic response from fractionalization in an insulating two-dimensional magnet}},\ }\href {https://doi.org/10.1038/nphys3809} {\bibfield  {journal} {\bibinfo  {journal} {Nat. Phys.}\ }\textbf {\bibinfo {volume} {12}},\ \bibinfo {pages} {912} (\bibinfo {year} {2016})}\BibitemShut {NoStop}%
\bibitem [{\citenamefont {Nasu}\ \emph {et~al.}(2017)\citenamefont {Nasu}, \citenamefont {Yoshitake},\ and\ \citenamefont {Motome}}]{nasu2017PRL}%
  \BibitemOpen
  \bibfield  {author} {\bibinfo {author} {\bibfnamefont {J.}~\bibnamefont {Nasu}}, \bibinfo {author} {\bibfnamefont {J.}~\bibnamefont {Yoshitake}},\ and\ \bibinfo {author} {\bibfnamefont {Y.}~\bibnamefont {Motome}},\ }\bibfield  {title} {\bibinfo {title} {{Thermal Transport in the Kitaev Model}},\ }\href {https://doi.org/10.1103/PhysRevLett.119.127204} {\bibfield  {journal} {\bibinfo  {journal} {Phys. Rev. Lett.}\ }\textbf {\bibinfo {volume} {119}},\ \bibinfo {pages} {127204} (\bibinfo {year} {2017})}\BibitemShut {NoStop}%
\bibitem [{\citenamefont {Gohlke}\ \emph {et~al.}(2017)\citenamefont {Gohlke}, \citenamefont {Verresen}, \citenamefont {Moessner},\ and\ \citenamefont {Pollmann}}]{gohlke2017PRL}%
  \BibitemOpen
  \bibfield  {author} {\bibinfo {author} {\bibfnamefont {M.}~\bibnamefont {Gohlke}}, \bibinfo {author} {\bibfnamefont {R.}~\bibnamefont {Verresen}}, \bibinfo {author} {\bibfnamefont {R.}~\bibnamefont {Moessner}},\ and\ \bibinfo {author} {\bibfnamefont {F.}~\bibnamefont {Pollmann}},\ }\bibfield  {title} {\bibinfo {title} {{Dynamics of the Kitaev-Heisenberg Model}},\ }\href {https://doi.org/10.1103/PhysRevLett.119.157203} {\bibfield  {journal} {\bibinfo  {journal} {Phys. Rev. Lett.}\ }\textbf {\bibinfo {volume} {119}},\ \bibinfo {pages} {157203} (\bibinfo {year} {2017})}\BibitemShut {NoStop}%
\bibitem [{\citenamefont {Yao}\ and\ \citenamefont {Kivelson}(2007)}]{yao2007PRL}%
  \BibitemOpen
  \bibfield  {author} {\bibinfo {author} {\bibfnamefont {H.}~\bibnamefont {Yao}}\ and\ \bibinfo {author} {\bibfnamefont {S.~A.}\ \bibnamefont {Kivelson}},\ }\bibfield  {title} {\bibinfo {title} {{Exact Chiral Spin Liquid with Non-Abelian Anyons}},\ }\href {https://doi.org/10.1103/PhysRevLett.99.247203} {\bibfield  {journal} {\bibinfo  {journal} {Phys. Rev. Lett.}\ }\textbf {\bibinfo {volume} {99}},\ \bibinfo {pages} {247203} (\bibinfo {year} {2007})}\BibitemShut {NoStop}%
\bibitem [{\citenamefont {Yao}\ \emph {et~al.}(2009)\citenamefont {Yao}, \citenamefont {Zhang},\ and\ \citenamefont {Kivelson}}]{yao2009PRL}%
  \BibitemOpen
  \bibfield  {author} {\bibinfo {author} {\bibfnamefont {H.}~\bibnamefont {Yao}}, \bibinfo {author} {\bibfnamefont {S.-C.}\ \bibnamefont {Zhang}},\ and\ \bibinfo {author} {\bibfnamefont {S.~A.}\ \bibnamefont {Kivelson}},\ }\bibfield  {title} {\bibinfo {title} {{Algebraic Spin Liquid in an Exactly Solvable Spin Model}},\ }\href {https://doi.org/10.1103/PhysRevLett.102.217202} {\bibfield  {journal} {\bibinfo  {journal} {Phys. Rev. Lett.}\ }\textbf {\bibinfo {volume} {102}},\ \bibinfo {pages} {217202} (\bibinfo {year} {2009})}\BibitemShut {NoStop}%
\bibitem [{\citenamefont {Chua}\ \emph {et~al.}(2011)\citenamefont {Chua}, \citenamefont {Yao},\ and\ \citenamefont {Fiete}}]{chua2011PRB}%
  \BibitemOpen
  \bibfield  {author} {\bibinfo {author} {\bibfnamefont {V.}~\bibnamefont {Chua}}, \bibinfo {author} {\bibfnamefont {H.}~\bibnamefont {Yao}},\ and\ \bibinfo {author} {\bibfnamefont {G.~A.}\ \bibnamefont {Fiete}},\ }\bibfield  {title} {\bibinfo {title} {Exact chiral spin liquid with stable spin fermi surface on the kagome lattice},\ }\href {https://doi.org/10.1103/PhysRevB.83.180412} {\bibfield  {journal} {\bibinfo  {journal} {Phys. Rev. B}\ }\textbf {\bibinfo {volume} {83}},\ \bibinfo {pages} {180412} (\bibinfo {year} {2011})}\BibitemShut {NoStop}%
\bibitem [{\citenamefont {Wu}\ \emph {et~al.}(2009)\citenamefont {Wu}, \citenamefont {Arovas},\ and\ \citenamefont {Hung}}]{wu2009PRB}%
  \BibitemOpen
  \bibfield  {author} {\bibinfo {author} {\bibfnamefont {C.}~\bibnamefont {Wu}}, \bibinfo {author} {\bibfnamefont {D.}~\bibnamefont {Arovas}},\ and\ \bibinfo {author} {\bibfnamefont {H.-H.}\ \bibnamefont {Hung}},\ }\bibfield  {title} {\bibinfo {title} {{$\Gamma$-matrix generalization of the Kitaev model}},\ }\href {https://doi.org/10.1103/PhysRevB.79.134427} {\bibfield  {journal} {\bibinfo  {journal} {Phys. Rev. B}\ }\textbf {\bibinfo {volume} {79}},\ \bibinfo {pages} {134427} (\bibinfo {year} {2009})}\BibitemShut {NoStop}%
\bibitem [{\citenamefont {Willans}\ \emph {et~al.}(2011)\citenamefont {Willans}, \citenamefont {Chalker},\ and\ \citenamefont {Moessner}}]{willans2011PRB}%
  \BibitemOpen
  \bibfield  {author} {\bibinfo {author} {\bibfnamefont {A.~J.}\ \bibnamefont {Willans}}, \bibinfo {author} {\bibfnamefont {J.~T.}\ \bibnamefont {Chalker}},\ and\ \bibinfo {author} {\bibfnamefont {R.}~\bibnamefont {Moessner}},\ }\bibfield  {title} {\bibinfo {title} {{Site dilution in the Kitaev honeycomb model}},\ }\href {https://doi.org/10.1103/PhysRevB.84.115146} {\bibfield  {journal} {\bibinfo  {journal} {Phys. Rev. B}\ }\textbf {\bibinfo {volume} {84}},\ \bibinfo {pages} {115146} (\bibinfo {year} {2011})}\BibitemShut {NoStop}%
\bibitem [{\citenamefont {Shackleton}\ and\ \citenamefont {Scheurer}(2024)}]{shackleton2024}%
  \BibitemOpen
  \bibfield  {author} {\bibinfo {author} {\bibfnamefont {H.}~\bibnamefont {Shackleton}}\ and\ \bibinfo {author} {\bibfnamefont {M.~S.}\ \bibnamefont {Scheurer}},\ }\bibfield  {title} {\bibinfo {title} {{Exactly solvable dissipative spin liquid}},\ }\href {https://doi.org/10.1103/PhysRevB.109.085115} {\bibfield  {journal} {\bibinfo  {journal} {Phys. Rev. B}\ }\textbf {\bibinfo {volume} {109}},\ \bibinfo {pages} {085115} (\bibinfo {year} {2024})}\BibitemShut {NoStop}%
\bibitem [{\citenamefont {Gidugu}\ and\ \citenamefont {Arovas}(2024)}]{gidugu2024}%
  \BibitemOpen
  \bibfield  {author} {\bibinfo {author} {\bibfnamefont {J.}~\bibnamefont {Gidugu}}\ and\ \bibinfo {author} {\bibfnamefont {D.~P.}\ \bibnamefont {Arovas}},\ }\bibfield  {title} {\bibinfo {title} {{Dissipative Dirac matrix spin model in two dimensions}},\ }\href {https://doi.org/10.1103/PhysRevA.109.022212} {\bibfield  {journal} {\bibinfo  {journal} {Phys. Rev. A}\ }\textbf {\bibinfo {volume} {109}},\ \bibinfo {pages} {022212} (\bibinfo {year} {2024})}\BibitemShut {NoStop}%
\bibitem [{\citenamefont {Dai}\ \emph {et~al.}(2023)\citenamefont {Dai}, \citenamefont {Song},\ and\ \citenamefont {Wang}}]{dai2023PRB}%
  \BibitemOpen
  \bibfield  {author} {\bibinfo {author} {\bibfnamefont {X.-D.}\ \bibnamefont {Dai}}, \bibinfo {author} {\bibfnamefont {F.}~\bibnamefont {Song}},\ and\ \bibinfo {author} {\bibfnamefont {Z.}~\bibnamefont {Wang}},\ }\bibfield  {title} {\bibinfo {title} {{Solvable BCS-Hubbard Liouvillians in arbitrary dimensions}},\ }\href {https://doi.org/10.1103/PhysRevB.108.115127} {\bibfield  {journal} {\bibinfo  {journal} {Phys. Rev. B}\ }\textbf {\bibinfo {volume} {108}},\ \bibinfo {pages} {115127} (\bibinfo {year} {2023})}\BibitemShut {NoStop}%
\bibitem [{\citenamefont {Hwang}(2024)}]{hwang2024Quantum}%
  \BibitemOpen
  \bibfield  {author} {\bibinfo {author} {\bibfnamefont {K.}~\bibnamefont {Hwang}},\ }\bibfield  {title} {\bibinfo {title} {Mixed-{S}tate {Q}uantum {S}pin {L}iquids and {D}ynamical {A}nyon {C}ondensations in {K}itaev {L}indbladians},\ }\href {https://doi.org/10.22331/q-2024-07-17-1412} {\bibfield  {journal} {\bibinfo  {journal} {{Quantum}}\ }\textbf {\bibinfo {volume} {8}},\ \bibinfo {pages} {1412} (\bibinfo {year} {2024})}\BibitemShut {NoStop}%
\bibitem [{\citenamefont {Pocklington}\ and\ \citenamefont {Clerk}(2025)}]{pocklington2025}%
  \BibitemOpen
  \bibfield  {author} {\bibinfo {author} {\bibfnamefont {A.}~\bibnamefont {Pocklington}}\ and\ \bibinfo {author} {\bibfnamefont {A.~A.}\ \bibnamefont {Clerk}},\ }\bibfield  {title} {\bibinfo {title} {{Efficient Simulation of Nontrivial Dissipative Spin Chains via Stochastic Unraveling}},\ }\href {https://doi.org/10.1103/vptq-xy6h} {\bibfield  {journal} {\bibinfo  {journal} {PRX Quantum}\ }\textbf {\bibinfo {volume} {6}},\ \bibinfo {pages} {030349} (\bibinfo {year} {2025})}\BibitemShut {NoStop}%
\bibitem [{\citenamefont {Lee}\ \emph {et~al.}(2025)\citenamefont {Lee}, \citenamefont {You},\ and\ \citenamefont {Xu}}]{lee2025quantum}%
  \BibitemOpen
  \bibfield  {author} {\bibinfo {author} {\bibfnamefont {J.~Y.}\ \bibnamefont {Lee}}, \bibinfo {author} {\bibfnamefont {Y.-Z.}\ \bibnamefont {You}},\ and\ \bibinfo {author} {\bibfnamefont {C.}~\bibnamefont {Xu}},\ }\bibfield  {title} {\bibinfo {title} {{Symmetry protected topological phases under decoherence}},\ }\href {https://doi.org/10.22331/q-2025-01-23-1607} {\bibfield  {journal} {\bibinfo  {journal} {{Quantum}}\ }\textbf {\bibinfo {volume} {9}},\ \bibinfo {pages} {1607} (\bibinfo {year} {2025})}\BibitemShut {NoStop}%
\bibitem [{\citenamefont {Wang}\ \emph {et~al.}(2024)\citenamefont {Wang}, \citenamefont {Dai}, \citenamefont {Wang},\ and\ \citenamefont {Wang}}]{wang2024SciPost}%
  \BibitemOpen
  \bibfield  {author} {\bibinfo {author} {\bibfnamefont {Z.}~\bibnamefont {Wang}}, \bibinfo {author} {\bibfnamefont {X.-D.}\ \bibnamefont {Dai}}, \bibinfo {author} {\bibfnamefont {H.-R.}\ \bibnamefont {Wang}},\ and\ \bibinfo {author} {\bibfnamefont {Z.}~\bibnamefont {Wang}},\ }\bibfield  {title} {\bibinfo {title} {{Topologically ordered steady states in open quantum systems}},\ }\href {https://doi.org/10.21468/SciPostPhys.17.6.167} {\bibfield  {journal} {\bibinfo  {journal} {SciPost Phys.}\ }\textbf {\bibinfo {volume} {17}},\ \bibinfo {pages} {167} (\bibinfo {year} {2024})}\BibitemShut {NoStop}%
\bibitem [{\citenamefont {Wang}\ \emph {et~al.}(2025{\natexlab{a}})\citenamefont {Wang}, \citenamefont {Wu},\ and\ \citenamefont {Wang}}]{wang2025PRXQ}%
  \BibitemOpen
  \bibfield  {author} {\bibinfo {author} {\bibfnamefont {Z.}~\bibnamefont {Wang}}, \bibinfo {author} {\bibfnamefont {Z.}~\bibnamefont {Wu}},\ and\ \bibinfo {author} {\bibfnamefont {Z.}~\bibnamefont {Wang}},\ }\bibfield  {title} {\bibinfo {title} {{Intrinsic Mixed-State Topological Order}},\ }\href {https://doi.org/10.1103/PRXQuantum.6.010314} {\bibfield  {journal} {\bibinfo  {journal} {PRX Quantum}\ }\textbf {\bibinfo {volume} {6}},\ \bibinfo {pages} {010314} (\bibinfo {year} {2025}{\natexlab{a}})}\BibitemShut {NoStop}%
\bibitem [{\citenamefont {Chen}\ and\ \citenamefont {Grover}(2024)}]{chen2024PRL}%
  \BibitemOpen
  \bibfield  {author} {\bibinfo {author} {\bibfnamefont {Y.-H.}\ \bibnamefont {Chen}}\ and\ \bibinfo {author} {\bibfnamefont {T.}~\bibnamefont {Grover}},\ }\bibfield  {title} {\bibinfo {title} {{Separability Transitions in Topological States Induced by Local Decoherence}},\ }\href {https://doi.org/10.1103/PhysRevLett.132.170602} {\bibfield  {journal} {\bibinfo  {journal} {Phys. Rev. Lett.}\ }\textbf {\bibinfo {volume} {132}},\ \bibinfo {pages} {170602} (\bibinfo {year} {2024})}\BibitemShut {NoStop}%
\bibitem [{\citenamefont {Dai}\ \emph {et~al.}(2025)\citenamefont {Dai}, \citenamefont {Wang}, \citenamefont {Wang},\ and\ \citenamefont {Wang}}]{dai2025PRB}%
  \BibitemOpen
  \bibfield  {author} {\bibinfo {author} {\bibfnamefont {X.-D.}\ \bibnamefont {Dai}}, \bibinfo {author} {\bibfnamefont {Z.}~\bibnamefont {Wang}}, \bibinfo {author} {\bibfnamefont {H.-R.}\ \bibnamefont {Wang}},\ and\ \bibinfo {author} {\bibfnamefont {Z.}~\bibnamefont {Wang}},\ }\bibfield  {title} {\bibinfo {title} {{Steady-state topological order}},\ }\href {https://doi.org/10.1103/PhysRevB.111.115142} {\bibfield  {journal} {\bibinfo  {journal} {Phys. Rev. B}\ }\textbf {\bibinfo {volume} {111}},\ \bibinfo {pages} {115142} (\bibinfo {year} {2025})}\BibitemShut {NoStop}%
\bibitem [{\citenamefont {Sohal}\ and\ \citenamefont {Prem}(2025)}]{sohal2025PRXQ}%
  \BibitemOpen
  \bibfield  {author} {\bibinfo {author} {\bibfnamefont {R.}~\bibnamefont {Sohal}}\ and\ \bibinfo {author} {\bibfnamefont {A.}~\bibnamefont {Prem}},\ }\bibfield  {title} {\bibinfo {title} {{Noisy Approach to Intrinsically Mixed-State Topological Order}},\ }\href {https://doi.org/10.1103/PRXQuantum.6.010313} {\bibfield  {journal} {\bibinfo  {journal} {PRX Quantum}\ }\textbf {\bibinfo {volume} {6}},\ \bibinfo {pages} {010313} (\bibinfo {year} {2025})}\BibitemShut {NoStop}%
\bibitem [{\citenamefont {Sang}\ and\ \citenamefont {Hsieh}(2025)}]{sang2024PRL}%
  \BibitemOpen
  \bibfield  {author} {\bibinfo {author} {\bibfnamefont {S.}~\bibnamefont {Sang}}\ and\ \bibinfo {author} {\bibfnamefont {T.~H.}\ \bibnamefont {Hsieh}},\ }\bibfield  {title} {\bibinfo {title} {{Stability of Mixed-State Quantum Phases via Finite Markov Length}},\ }\href {https://doi.org/10.1103/PhysRevLett.134.070403} {\bibfield  {journal} {\bibinfo  {journal} {Phys. Rev. Lett.}\ }\textbf {\bibinfo {volume} {134}},\ \bibinfo {pages} {070403} (\bibinfo {year} {2025})}\BibitemShut {NoStop}%
\bibitem [{\citenamefont {Li}\ \emph {et~al.}(2025)\citenamefont {Li}, \citenamefont {Lee},\ and\ \citenamefont {Yoshida}}]{li2024arXiv}%
  \BibitemOpen
  \bibfield  {author} {\bibinfo {author} {\bibfnamefont {Z.}~\bibnamefont {Li}}, \bibinfo {author} {\bibfnamefont {D.}~\bibnamefont {Lee}},\ and\ \bibinfo {author} {\bibfnamefont {B.}~\bibnamefont {Yoshida}},\ }\bibfield  {title} {\bibinfo {title} {{How Much Entanglement Is Needed for Topological Codes and Mixed States with Anomalous Symmetry?}},\ }\href {https://doi.org/10.1103/pw12-kdjx} {\bibfield  {journal} {\bibinfo  {journal} {Phys. Rev. X}\ }\textbf {\bibinfo {volume} {15}},\ \bibinfo {pages} {021090} (\bibinfo {year} {2025})}\BibitemShut {NoStop}%
\bibitem [{\citenamefont {Wang}\ \emph {et~al.}(2025{\natexlab{b}})\citenamefont {Wang}, \citenamefont {Song}, \citenamefont {Meng},\ and\ \citenamefont {Grover}}]{wang2025PRXQanalogEE}%
  \BibitemOpen
  \bibfield  {author} {\bibinfo {author} {\bibfnamefont {T.-T.}\ \bibnamefont {Wang}}, \bibinfo {author} {\bibfnamefont {M.}~\bibnamefont {Song}}, \bibinfo {author} {\bibfnamefont {Z.~Y.}\ \bibnamefont {Meng}},\ and\ \bibinfo {author} {\bibfnamefont {T.}~\bibnamefont {Grover}},\ }\bibfield  {title} {\bibinfo {title} {{Analog of Topological Entanglement Entropy for Mixed States}},\ }\href {https://doi.org/10.1103/PRXQuantum.6.010358} {\bibfield  {journal} {\bibinfo  {journal} {PRX Quantum}\ }\textbf {\bibinfo {volume} {6}},\ \bibinfo {pages} {010358} (\bibinfo {year} {2025}{\natexlab{b}})}\BibitemShut {NoStop}%
\bibitem [{\citenamefont {Lessa}\ \emph {et~al.}(2025{\natexlab{a}})\citenamefont {Lessa}, \citenamefont {Sang}, \citenamefont {Lu}, \citenamefont {Hsieh},\ and\ \citenamefont {Wang}}]{lessa2025arXiv}%
  \BibitemOpen
  \bibfield  {author} {\bibinfo {author} {\bibfnamefont {L.~A.}\ \bibnamefont {Lessa}}, \bibinfo {author} {\bibfnamefont {S.}~\bibnamefont {Sang}}, \bibinfo {author} {\bibfnamefont {T.-C.}\ \bibnamefont {Lu}}, \bibinfo {author} {\bibfnamefont {T.~H.}\ \bibnamefont {Hsieh}},\ and\ \bibinfo {author} {\bibfnamefont {C.}~\bibnamefont {Wang}},\ }\bibfield  {title} {\bibinfo {title} {{Higher-form anomaly and long-range entanglement of mixed states}},\ }\href {https://arxiv.org/abs/2503.12792} {\bibfield  {journal} {\bibinfo  {journal} {arXiv:2503.12792}\ } (\bibinfo {year} {2025}{\natexlab{a}})}\BibitemShut {NoStop}%
\bibitem [{\citenamefont {Sang}\ \emph {et~al.}(2024)\citenamefont {Sang}, \citenamefont {Zou},\ and\ \citenamefont {Hsieh}}]{sang2024PRX}%
  \BibitemOpen
  \bibfield  {author} {\bibinfo {author} {\bibfnamefont {S.}~\bibnamefont {Sang}}, \bibinfo {author} {\bibfnamefont {Y.}~\bibnamefont {Zou}},\ and\ \bibinfo {author} {\bibfnamefont {T.~H.}\ \bibnamefont {Hsieh}},\ }\bibfield  {title} {\bibinfo {title} {{Mixed-State Quantum Phases: Renormalization and Quantum Error Correction}},\ }\href {https://doi.org/10.1103/PhysRevX.14.031044} {\bibfield  {journal} {\bibinfo  {journal} {Phys. Rev. X}\ }\textbf {\bibinfo {volume} {14}},\ \bibinfo {pages} {031044} (\bibinfo {year} {2024})}\BibitemShut {NoStop}%
\bibitem [{\citenamefont {Li}\ and\ \citenamefont {Mong}(2025)}]{li2025PRB}%
  \BibitemOpen
  \bibfield  {author} {\bibinfo {author} {\bibfnamefont {Z.}~\bibnamefont {Li}}\ and\ \bibinfo {author} {\bibfnamefont {R.~S.~K.}\ \bibnamefont {Mong}},\ }\bibfield  {title} {\bibinfo {title} {{Replica topological order in quantum mixed states and quantum error correction}},\ }\href {https://doi.org/10.1103/PhysRevB.111.125106} {\bibfield  {journal} {\bibinfo  {journal} {Phys. Rev. B}\ }\textbf {\bibinfo {volume} {111}},\ \bibinfo {pages} {125106} (\bibinfo {year} {2025})}\BibitemShut {NoStop}%
\bibitem [{\citenamefont {Ellison}\ and\ \citenamefont {Cheng}(2025)}]{ellison2025PRXQ}%
  \BibitemOpen
  \bibfield  {author} {\bibinfo {author} {\bibfnamefont {T.~D.}\ \bibnamefont {Ellison}}\ and\ \bibinfo {author} {\bibfnamefont {M.}~\bibnamefont {Cheng}},\ }\bibfield  {title} {\bibinfo {title} {{Toward a Classification of Mixed-State Topological Orders in Two Dimensions}},\ }\href {https://doi.org/10.1103/PRXQuantum.6.010315} {\bibfield  {journal} {\bibinfo  {journal} {PRX Quantum}\ }\textbf {\bibinfo {volume} {6}},\ \bibinfo {pages} {010315} (\bibinfo {year} {2025})}\BibitemShut {NoStop}%
\bibitem [{\citenamefont {Bao}\ \emph {et~al.}(2023)\citenamefont {Bao}, \citenamefont {Fan}, \citenamefont {Vishwanath},\ and\ \citenamefont {Altman}}]{bao2023arXiv}%
  \BibitemOpen
  \bibfield  {author} {\bibinfo {author} {\bibfnamefont {Y.}~\bibnamefont {Bao}}, \bibinfo {author} {\bibfnamefont {R.}~\bibnamefont {Fan}}, \bibinfo {author} {\bibfnamefont {A.}~\bibnamefont {Vishwanath}},\ and\ \bibinfo {author} {\bibfnamefont {E.}~\bibnamefont {Altman}},\ }\bibfield  {title} {\bibinfo {title} {{Mixed-state topological order and the errorfield double formulation of decoherence-induced transitions}},\ }\href {https://arxiv.org/abs/2301.05687} {\bibfield  {journal} {\bibinfo  {journal} {arXiv:2301.05687}\ } (\bibinfo {year} {2023})}\BibitemShut {NoStop}%
\bibitem [{\citenamefont {Fan}\ \emph {et~al.}(2024)\citenamefont {Fan}, \citenamefont {Bao}, \citenamefont {Altman},\ and\ \citenamefont {Vishwanath}}]{fan2024PRXQ}%
  \BibitemOpen
  \bibfield  {author} {\bibinfo {author} {\bibfnamefont {R.}~\bibnamefont {Fan}}, \bibinfo {author} {\bibfnamefont {Y.}~\bibnamefont {Bao}}, \bibinfo {author} {\bibfnamefont {E.}~\bibnamefont {Altman}},\ and\ \bibinfo {author} {\bibfnamefont {A.}~\bibnamefont {Vishwanath}},\ }\bibfield  {title} {\bibinfo {title} {{Diagnostics of Mixed-State Topological Order and Breakdown of Quantum Memory}},\ }\href {https://doi.org/10.1103/PRXQuantum.5.020343} {\bibfield  {journal} {\bibinfo  {journal} {PRX Quantum}\ }\textbf {\bibinfo {volume} {5}},\ \bibinfo {pages} {020343} (\bibinfo {year} {2024})}\BibitemShut {NoStop}%
\bibitem [{\citenamefont {Kim}\ \emph {et~al.}(2024)\citenamefont {Kim}, \citenamefont {Altman},\ and\ \citenamefont {Lee}}]{kim2024arXiv}%
  \BibitemOpen
  \bibfield  {author} {\bibinfo {author} {\bibfnamefont {J.}~\bibnamefont {Kim}}, \bibinfo {author} {\bibfnamefont {E.}~\bibnamefont {Altman}},\ and\ \bibinfo {author} {\bibfnamefont {J.~Y.}\ \bibnamefont {Lee}},\ }\bibfield  {title} {\bibinfo {title} {{Error Threshold of SYK Codes from Strong-to-Weak Parity Symmetry Breaking}},\ }\href {https://arxiv.org/abs/2410.24225} {\bibfield  {journal} {\bibinfo  {journal} {arXiv:2410.24225}\ } (\bibinfo {year} {2024})}\BibitemShut {NoStop}%
\bibitem [{\citenamefont {Lee}\ \emph {et~al.}(2023)\citenamefont {Lee}, \citenamefont {Jian},\ and\ \citenamefont {Xu}}]{lee2023PRXQ}%
  \BibitemOpen
  \bibfield  {author} {\bibinfo {author} {\bibfnamefont {J.~Y.}\ \bibnamefont {Lee}}, \bibinfo {author} {\bibfnamefont {C.-M.}\ \bibnamefont {Jian}},\ and\ \bibinfo {author} {\bibfnamefont {C.}~\bibnamefont {Xu}},\ }\bibfield  {title} {\bibinfo {title} {{Quantum Criticality Under Decoherence or Weak Measurement}},\ }\href {https://doi.org/10.1103/PRXQuantum.4.030317} {\bibfield  {journal} {\bibinfo  {journal} {PRX Quantum}\ }\textbf {\bibinfo {volume} {4}},\ \bibinfo {pages} {030317} (\bibinfo {year} {2023})}\BibitemShut {NoStop}%
\bibitem [{\citenamefont {Ma}\ \emph {et~al.}(2025)\citenamefont {Ma}, \citenamefont {Zhang}, \citenamefont {Bi}, \citenamefont {Cheng},\ and\ \citenamefont {Wang}}]{ma2023arXiv}%
  \BibitemOpen
  \bibfield  {author} {\bibinfo {author} {\bibfnamefont {R.}~\bibnamefont {Ma}}, \bibinfo {author} {\bibfnamefont {J.-H.}\ \bibnamefont {Zhang}}, \bibinfo {author} {\bibfnamefont {Z.}~\bibnamefont {Bi}}, \bibinfo {author} {\bibfnamefont {M.}~\bibnamefont {Cheng}},\ and\ \bibinfo {author} {\bibfnamefont {C.}~\bibnamefont {Wang}},\ }\bibfield  {title} {\bibinfo {title} {{Topological Phases with Average Symmetries: The Decohered, the Disordered, and the Intrinsic}},\ }\href {https://doi.org/10.1103/PhysRevX.15.021062} {\bibfield  {journal} {\bibinfo  {journal} {Phys. Rev. X}\ }\textbf {\bibinfo {volume} {15}},\ \bibinfo {pages} {021062} (\bibinfo {year} {2025})}\BibitemShut {NoStop}%
\bibitem [{\citenamefont {Sala}\ \emph {et~al.}(2024)\citenamefont {Sala}, \citenamefont {Gopalakrishnan}, \citenamefont {Oshikawa},\ and\ \citenamefont {You}}]{sala2024PRB}%
  \BibitemOpen
  \bibfield  {author} {\bibinfo {author} {\bibfnamefont {P.}~\bibnamefont {Sala}}, \bibinfo {author} {\bibfnamefont {S.}~\bibnamefont {Gopalakrishnan}}, \bibinfo {author} {\bibfnamefont {M.}~\bibnamefont {Oshikawa}},\ and\ \bibinfo {author} {\bibfnamefont {Y.}~\bibnamefont {You}},\ }\bibfield  {title} {\bibinfo {title} {{Spontaneous strong symmetry breaking in open systems: Purification perspective}},\ }\href {https://doi.org/10.1103/PhysRevB.110.155150} {\bibfield  {journal} {\bibinfo  {journal} {Phys. Rev. B}\ }\textbf {\bibinfo {volume} {110}},\ \bibinfo {pages} {155150} (\bibinfo {year} {2024})}\BibitemShut {NoStop}%
\bibitem [{\citenamefont {Lessa}\ \emph {et~al.}(2025{\natexlab{b}})\citenamefont {Lessa}, \citenamefont {Ma}, \citenamefont {Zhang}, \citenamefont {Bi}, \citenamefont {Cheng},\ and\ \citenamefont {Wang}}]{lessa2025PRXQ}%
  \BibitemOpen
  \bibfield  {author} {\bibinfo {author} {\bibfnamefont {L.~A.}\ \bibnamefont {Lessa}}, \bibinfo {author} {\bibfnamefont {R.}~\bibnamefont {Ma}}, \bibinfo {author} {\bibfnamefont {J.-H.}\ \bibnamefont {Zhang}}, \bibinfo {author} {\bibfnamefont {Z.}~\bibnamefont {Bi}}, \bibinfo {author} {\bibfnamefont {M.}~\bibnamefont {Cheng}},\ and\ \bibinfo {author} {\bibfnamefont {C.}~\bibnamefont {Wang}},\ }\bibfield  {title} {\bibinfo {title} {{Strong-to-Weak Spontaneous Symmetry Breaking in Mixed Quantum States}},\ }\href {https://doi.org/10.1103/PRXQuantum.6.010344} {\bibfield  {journal} {\bibinfo  {journal} {PRX Quantum}\ }\textbf {\bibinfo {volume} {6}},\ \bibinfo {pages} {010344} (\bibinfo {year} {2025}{\natexlab{b}})}\BibitemShut {NoStop}%
\bibitem [{\citenamefont {Zhang}\ \emph {et~al.}(2025)\citenamefont {Zhang}, \citenamefont {Xu}, \citenamefont {Zhang}, \citenamefont {Xu}, \citenamefont {Bi},\ and\ \citenamefont {Luo}}]{zhang2025PRB}%
  \BibitemOpen
  \bibfield  {author} {\bibinfo {author} {\bibfnamefont {C.}~\bibnamefont {Zhang}}, \bibinfo {author} {\bibfnamefont {Y.}~\bibnamefont {Xu}}, \bibinfo {author} {\bibfnamefont {J.-H.}\ \bibnamefont {Zhang}}, \bibinfo {author} {\bibfnamefont {C.}~\bibnamefont {Xu}}, \bibinfo {author} {\bibfnamefont {Z.}~\bibnamefont {Bi}},\ and\ \bibinfo {author} {\bibfnamefont {Z.-X.}\ \bibnamefont {Luo}},\ }\bibfield  {title} {\bibinfo {title} {{Strong-to-weak spontaneous breaking of 1-form symmetry and intrinsically mixed topological order}},\ }\href {https://doi.org/10.1103/PhysRevB.111.115137} {\bibfield  {journal} {\bibinfo  {journal} {Phys. Rev. B}\ }\textbf {\bibinfo {volume} {111}},\ \bibinfo {pages} {115137} (\bibinfo {year} {2025})}\BibitemShut {NoStop}%
\bibitem [{\citenamefont {Ma}\ and\ \citenamefont {Turzillo}(2025)}]{ma2025PRXQ}%
  \BibitemOpen
  \bibfield  {author} {\bibinfo {author} {\bibfnamefont {R.}~\bibnamefont {Ma}}\ and\ \bibinfo {author} {\bibfnamefont {A.}~\bibnamefont {Turzillo}},\ }\bibfield  {title} {\bibinfo {title} {{Symmetry-Protected Topological Phases of Mixed States in the Doubled Space}},\ }\href {https://doi.org/10.1103/PRXQuantum.6.010348} {\bibfield  {journal} {\bibinfo  {journal} {PRX Quantum}\ }\textbf {\bibinfo {volume} {6}},\ \bibinfo {pages} {010348} (\bibinfo {year} {2025})}\BibitemShut {NoStop}%
\bibitem [{\citenamefont {Moharramipour}\ \emph {et~al.}(2024)\citenamefont {Moharramipour}, \citenamefont {Lessa}, \citenamefont {Wang}, \citenamefont {Hsieh},\ and\ \citenamefont {Sahu}}]{moharramipour2024PRXQ}%
  \BibitemOpen
  \bibfield  {author} {\bibinfo {author} {\bibfnamefont {A.}~\bibnamefont {Moharramipour}}, \bibinfo {author} {\bibfnamefont {L.~A.}\ \bibnamefont {Lessa}}, \bibinfo {author} {\bibfnamefont {C.}~\bibnamefont {Wang}}, \bibinfo {author} {\bibfnamefont {T.~H.}\ \bibnamefont {Hsieh}},\ and\ \bibinfo {author} {\bibfnamefont {S.}~\bibnamefont {Sahu}},\ }\bibfield  {title} {\bibinfo {title} {{Symmetry-Enforced Entanglement in Maximally Mixed States}},\ }\href {https://doi.org/10.1103/PRXQuantum.5.040336} {\bibfield  {journal} {\bibinfo  {journal} {PRX Quantum}\ }\textbf {\bibinfo {volume} {5}},\ \bibinfo {pages} {040336} (\bibinfo {year} {2024})}\BibitemShut {NoStop}%
\bibitem [{\citenamefont {Gu}\ \emph {et~al.}(2024)\citenamefont {Gu}, \citenamefont {Wang},\ and\ \citenamefont {Wang}}]{gu2024arXiv}%
  \BibitemOpen
  \bibfield  {author} {\bibinfo {author} {\bibfnamefont {D.}~\bibnamefont {Gu}}, \bibinfo {author} {\bibfnamefont {Z.}~\bibnamefont {Wang}},\ and\ \bibinfo {author} {\bibfnamefont {Z.}~\bibnamefont {Wang}},\ }\bibfield  {title} {\bibinfo {title} {{Spontaneous symmetry breaking in open quantum systems: strong, weak, and strong-to-weak}},\ }\href {https://arxiv.org/abs/2406.19381} {\bibfield  {journal} {\bibinfo  {journal} {arXiv:2406.19381}\ } (\bibinfo {year} {2024})}\BibitemShut {NoStop}%
\bibitem [{\citenamefont {Liu}\ \emph {et~al.}(2025)\citenamefont {Liu}, \citenamefont {Chen}, \citenamefont {Zhang}, \citenamefont {Zhou},\ and\ \citenamefont {Zhang}}]{liu2024arXiv}%
  \BibitemOpen
  \bibfield  {author} {\bibinfo {author} {\bibfnamefont {Z.}~\bibnamefont {Liu}}, \bibinfo {author} {\bibfnamefont {L.}~\bibnamefont {Chen}}, \bibinfo {author} {\bibfnamefont {Y.}~\bibnamefont {Zhang}}, \bibinfo {author} {\bibfnamefont {S.}~\bibnamefont {Zhou}},\ and\ \bibinfo {author} {\bibfnamefont {P.}~\bibnamefont {Zhang}},\ }\bibfield  {title} {\bibinfo {title} {{Diagnosing strong-to-weak symmetry breaking via Wightman correlators}},\ }\href {https://doi.org/https://doi.org/10.1038/s42005-025-02199-7} {\bibfield  {journal} {\bibinfo  {journal} {Commun. Phys.}\ }\textbf {\bibinfo {volume} {8}},\ \bibinfo {pages} {274} (\bibinfo {year} {2025})}\BibitemShut {NoStop}%
\bibitem [{\citenamefont {Weinstein}(2025)}]{weinstein2025PRL}%
  \BibitemOpen
  \bibfield  {author} {\bibinfo {author} {\bibfnamefont {Z.}~\bibnamefont {Weinstein}},\ }\bibfield  {title} {\bibinfo {title} {{Efficient Detection of Strong-to-Weak Spontaneous Symmetry Breaking via the R\'enyi-1 Correlator}},\ }\href {https://doi.org/10.1103/PhysRevLett.134.150405} {\bibfield  {journal} {\bibinfo  {journal} {Phys. Rev. Lett.}\ }\textbf {\bibinfo {volume} {134}},\ \bibinfo {pages} {150405} (\bibinfo {year} {2025})}\BibitemShut {NoStop}%
\bibitem [{\citenamefont {Ogunnaike}\ \emph {et~al.}(2023)\citenamefont {Ogunnaike}, \citenamefont {Feldmeier},\ and\ \citenamefont {Lee}}]{ogunnaike2023PRL}%
  \BibitemOpen
  \bibfield  {author} {\bibinfo {author} {\bibfnamefont {O.}~\bibnamefont {Ogunnaike}}, \bibinfo {author} {\bibfnamefont {J.}~\bibnamefont {Feldmeier}},\ and\ \bibinfo {author} {\bibfnamefont {J.~Y.}\ \bibnamefont {Lee}},\ }\bibfield  {title} {\bibinfo {title} {{Unifying Emergent Hydrodynamics and Lindbladian Low-Energy Spectra across Symmetries, Constraints, and Long-Range Interactions}},\ }\href {https://doi.org/10.1103/PhysRevLett.131.220403} {\bibfield  {journal} {\bibinfo  {journal} {Phys. Rev. Lett.}\ }\textbf {\bibinfo {volume} {131}},\ \bibinfo {pages} {220403} (\bibinfo {year} {2023})}\BibitemShut {NoStop}%
\bibitem [{\citenamefont {Xu}\ and\ \citenamefont {Jian}(2025)}]{xu2025PRB}%
  \BibitemOpen
  \bibfield  {author} {\bibinfo {author} {\bibfnamefont {Y.}~\bibnamefont {Xu}}\ and\ \bibinfo {author} {\bibfnamefont {C.-M.}\ \bibnamefont {Jian}},\ }\bibfield  {title} {\bibinfo {title} {{Average-exact mixed anomalies and compatible phases}},\ }\href {https://doi.org/10.1103/PhysRevB.111.125128} {\bibfield  {journal} {\bibinfo  {journal} {Phys. Rev. B}\ }\textbf {\bibinfo {volume} {111}},\ \bibinfo {pages} {125128} (\bibinfo {year} {2025})}\BibitemShut {NoStop}%
\bibitem [{\citenamefont {Huang}\ \emph {et~al.}(2025)\citenamefont {Huang}, \citenamefont {Qi}, \citenamefont {Zhang},\ and\ \citenamefont {Lucas}}]{huang2025PRB}%
  \BibitemOpen
  \bibfield  {author} {\bibinfo {author} {\bibfnamefont {X.}~\bibnamefont {Huang}}, \bibinfo {author} {\bibfnamefont {M.}~\bibnamefont {Qi}}, \bibinfo {author} {\bibfnamefont {J.-H.}\ \bibnamefont {Zhang}},\ and\ \bibinfo {author} {\bibfnamefont {A.}~\bibnamefont {Lucas}},\ }\bibfield  {title} {\bibinfo {title} {{Hydrodynamics as the effective field theory of strong-to-weak spontaneous symmetry breaking}},\ }\href {https://doi.org/10.1103/PhysRevB.111.125147} {\bibfield  {journal} {\bibinfo  {journal} {Phys. Rev. B}\ }\textbf {\bibinfo {volume} {111}},\ \bibinfo {pages} {125147} (\bibinfo {year} {2025})}\BibitemShut {NoStop}%
\bibitem [{\citenamefont {Zhang}\ \emph {et~al.}(2024{\natexlab{a}})\citenamefont {Zhang}, \citenamefont {Xu},\ and\ \citenamefont {Xu}}]{zhang2024arXiv}%
  \BibitemOpen
  \bibfield  {author} {\bibinfo {author} {\bibfnamefont {J.-H.}\ \bibnamefont {Zhang}}, \bibinfo {author} {\bibfnamefont {C.}~\bibnamefont {Xu}},\ and\ \bibinfo {author} {\bibfnamefont {Y.}~\bibnamefont {Xu}},\ }\bibfield  {title} {\bibinfo {title} {{Fluctuation-dissipation theorem and information geometry in open quantum systems}},\ }\href {https://arxiv.org/abs/2409.18944} {\bibfield  {journal} {\bibinfo  {journal} {arXiv:2409.18944}\ } (\bibinfo {year} {2024}{\natexlab{a}})}\BibitemShut {NoStop}%
\bibitem [{\citenamefont {Guo}\ and\ \citenamefont {Yang}(2025)}]{guo2024arXiv}%
  \BibitemOpen
  \bibfield  {author} {\bibinfo {author} {\bibfnamefont {Y.}~\bibnamefont {Guo}}\ and\ \bibinfo {author} {\bibfnamefont {S.}~\bibnamefont {Yang}},\ }\bibfield  {title} {\bibinfo {title} {{Strong-to-weak spontaneous symmetry breaking meets average symmetry-protected topological order}},\ }\href {https://doi.org/10.1103/PhysRevB.111.L201108} {\bibfield  {journal} {\bibinfo  {journal} {Phys. Rev. B}\ }\textbf {\bibinfo {volume} {111}},\ \bibinfo {pages} {L201108} (\bibinfo {year} {2025})}\BibitemShut {NoStop}%
\bibitem [{\citenamefont {Ando}\ \emph {et~al.}(2024)\citenamefont {Ando}, \citenamefont {Ryu},\ and\ \citenamefont {Watanabe}}]{ando2024arXiv}%
  \BibitemOpen
  \bibfield  {author} {\bibinfo {author} {\bibfnamefont {T.}~\bibnamefont {Ando}}, \bibinfo {author} {\bibfnamefont {S.}~\bibnamefont {Ryu}},\ and\ \bibinfo {author} {\bibfnamefont {M.}~\bibnamefont {Watanabe}},\ }\bibfield  {title} {\bibinfo {title} {{Gauge theory and mixed state criticality}},\ }\href {https://arxiv.org/abs/2411.04360} {\bibfield  {journal} {\bibinfo  {journal} {arXiv:2411.04360}\ } (\bibinfo {year} {2024})}\BibitemShut {NoStop}%
\bibitem [{\citenamefont {Chen}\ \emph {et~al.}(2025)\citenamefont {Chen}, \citenamefont {Sun},\ and\ \citenamefont {Zhang}}]{chen2025PRB}%
  \BibitemOpen
  \bibfield  {author} {\bibinfo {author} {\bibfnamefont {L.}~\bibnamefont {Chen}}, \bibinfo {author} {\bibfnamefont {N.}~\bibnamefont {Sun}},\ and\ \bibinfo {author} {\bibfnamefont {P.}~\bibnamefont {Zhang}},\ }\bibfield  {title} {\bibinfo {title} {{Strong-to-weak symmetry breaking and entanglement transitions}},\ }\href {https://doi.org/10.1103/PhysRevB.111.L060304} {\bibfield  {journal} {\bibinfo  {journal} {Phys. Rev. B}\ }\textbf {\bibinfo {volume} {111}},\ \bibinfo {pages} {L060304} (\bibinfo {year} {2025})}\BibitemShut {NoStop}%
\bibitem [{\citenamefont {Sun}\ \emph {et~al.}(2025)\citenamefont {Sun}, \citenamefont {Zhang},\ and\ \citenamefont {Feng}}]{sun2024arXiv}%
  \BibitemOpen
  \bibfield  {author} {\bibinfo {author} {\bibfnamefont {N.}~\bibnamefont {Sun}}, \bibinfo {author} {\bibfnamefont {P.}~\bibnamefont {Zhang}},\ and\ \bibinfo {author} {\bibfnamefont {L.}~\bibnamefont {Feng}},\ }\bibfield  {title} {\bibinfo {title} {{Scheme to Detect the Strong-to-Weak Symmetry Breaking via Randomized Measurements}},\ }\href {https://doi.org/10.1103/7p5x-7yqb} {\bibfield  {journal} {\bibinfo  {journal} {Phys. Rev. Lett.}\ }\textbf {\bibinfo {volume} {135}},\ \bibinfo {pages} {090403} (\bibinfo {year} {2025})}\BibitemShut {NoStop}%
\bibitem [{\citenamefont {Orito}\ \emph {et~al.}(2025)\citenamefont {Orito}, \citenamefont {Kuno},\ and\ \citenamefont {Ichinose}}]{orito2025PRB}%
  \BibitemOpen
  \bibfield  {author} {\bibinfo {author} {\bibfnamefont {T.}~\bibnamefont {Orito}}, \bibinfo {author} {\bibfnamefont {Y.}~\bibnamefont {Kuno}},\ and\ \bibinfo {author} {\bibfnamefont {I.}~\bibnamefont {Ichinose}},\ }\bibfield  {title} {\bibinfo {title} {{Strong and weak symmetries and their spontaneous symmetry breaking in mixed states emerging from the quantum Ising model under multiple decoherence}},\ }\href {https://doi.org/10.1103/PhysRevB.111.054106} {\bibfield  {journal} {\bibinfo  {journal} {Phys. Rev. B}\ }\textbf {\bibinfo {volume} {111}},\ \bibinfo {pages} {054106} (\bibinfo {year} {2025})}\BibitemShut {NoStop}%
\bibitem [{\citenamefont {Lu}\ \emph {et~al.}(2024)\citenamefont {Lu}, \citenamefont {Zhu},\ and\ \citenamefont {Lu}}]{lu2024arXiv}%
  \BibitemOpen
  \bibfield  {author} {\bibinfo {author} {\bibfnamefont {S.}~\bibnamefont {Lu}}, \bibinfo {author} {\bibfnamefont {P.}~\bibnamefont {Zhu}},\ and\ \bibinfo {author} {\bibfnamefont {Y.-M.}\ \bibnamefont {Lu}},\ }\bibfield  {title} {\bibinfo {title} {{Bilayer construction for mixed state phenomena with strong, weak symmetries and symmetry breakings}},\ }\href {https://arxiv.org/abs/2411.07174} {\bibfield  {journal} {\bibinfo  {journal} {arXiv:2411.07174}\ } (\bibinfo {year} {2024})}\BibitemShut {NoStop}%
\bibitem [{\citenamefont {Feng}\ \emph {et~al.}(2025)\citenamefont {Feng}, \citenamefont {Cheng},\ and\ \citenamefont {Ippoliti}}]{feng2025arXiv}%
  \BibitemOpen
  \bibfield  {author} {\bibinfo {author} {\bibfnamefont {X.}~\bibnamefont {Feng}}, \bibinfo {author} {\bibfnamefont {Z.}~\bibnamefont {Cheng}},\ and\ \bibinfo {author} {\bibfnamefont {M.}~\bibnamefont {Ippoliti}},\ }\bibfield  {title} {\bibinfo {title} {{Hardness of observing strong-to-weak symmetry breaking}},\ }\href {https://arxiv.org/abs/2504.12233} {\bibfield  {journal} {\bibinfo  {journal} {arXiv:2504.12233}\ } (\bibinfo {year} {2025})}\BibitemShut {NoStop}%
\bibitem [{\citenamefont {Prosen}(2008)}]{prosen2008}%
  \BibitemOpen
  \bibfield  {author} {\bibinfo {author} {\bibfnamefont {T.}~\bibnamefont {Prosen}},\ }\bibfield  {title} {\bibinfo {title} {{Third quantization: a general method to solve master equations for quadratic open Fermi systems}},\ }\href {https://doi.org/10.1088/1367-2630/10/4/043026} {\bibfield  {journal} {\bibinfo  {journal} {New J. Phys.}\ }\textbf {\bibinfo {volume} {10}},\ \bibinfo {pages} {043026} (\bibinfo {year} {2008})}\BibitemShut {NoStop}%
\bibitem [{\citenamefont {Prosen}\ and\ \citenamefont {Pi{\v{z}}orn}(2008)}]{prosen2008prl}%
  \BibitemOpen
  \bibfield  {author} {\bibinfo {author} {\bibfnamefont {T.}~\bibnamefont {Prosen}}\ and\ \bibinfo {author} {\bibfnamefont {I.}~\bibnamefont {Pi{\v{z}}orn}},\ }\bibfield  {title} {\bibinfo {title} {{Quantum Phase Transition in a Far-from-Equilibrium Steady State of an $XY$ Spin Chain}},\ }\href {https://doi.org/10.1103/PhysRevLett.101.105701} {\bibfield  {journal} {\bibinfo  {journal} {Phys. Rev. Lett.}\ }\textbf {\bibinfo {volume} {101}},\ \bibinfo {pages} {105701} (\bibinfo {year} {2008})}\BibitemShut {NoStop}%
\bibitem [{\citenamefont {Prosen}(2010)}]{prosen2010jstat}%
  \BibitemOpen
  \bibfield  {author} {\bibinfo {author} {\bibfnamefont {T.}~\bibnamefont {Prosen}},\ }\bibfield  {title} {\bibinfo {title} {{Spectral theorem for the Lindblad equation for quadratic open fermionic systems}},\ }\href {https://doi.org/10.1088/1742-5468/2010/07/P07020} {\bibfield  {journal} {\bibinfo  {journal} {J. Stat. Mech.}\ }\textbf {\bibinfo {volume} {2010}},\ \bibinfo {pages} {P07020} (\bibinfo {year} {2010})}\BibitemShut {NoStop}%
\bibitem [{\citenamefont {Prosen}\ and\ \citenamefont {{\v{Z}}unkovi{\v{c}}}(2010)}]{prosen2010njp}%
  \BibitemOpen
  \bibfield  {author} {\bibinfo {author} {\bibfnamefont {T.}~\bibnamefont {Prosen}}\ and\ \bibinfo {author} {\bibfnamefont {B.}~\bibnamefont {{\v{Z}}unkovi{\v{c}}}},\ }\bibfield  {title} {\bibinfo {title} {{Exact solution of Markovian master equations for quadratic Fermi systems: thermal baths, open {XY} spin chains and non-equilibrium phase transition}},\ }\href {https://doi.org/10.1088/1367-2630/12/2/025016} {\bibfield  {journal} {\bibinfo  {journal} {New J. Phys.}\ }\textbf {\bibinfo {volume} {12}},\ \bibinfo {pages} {025016} (\bibinfo {year} {2010})}\BibitemShut {NoStop}%
\bibitem [{\citenamefont {{\v{Z}}unkovi{\v{c}}}\ and\ \citenamefont {Prosen}(2010)}]{zunkovivc2010}%
  \BibitemOpen
  \bibfield  {author} {\bibinfo {author} {\bibfnamefont {B.}~\bibnamefont {{\v{Z}}unkovi{\v{c}}}}\ and\ \bibinfo {author} {\bibfnamefont {T.}~\bibnamefont {Prosen}},\ }\bibfield  {title} {\bibinfo {title} {{Explicit solution of the Lindblad equation for nearly isotropic boundary driven XY spin 1/2 chain}},\ }\href {https://doi.org/10.1088/1742-5468/2010/08/P08016} {\bibfield  {journal} {\bibinfo  {journal} {J. Stat. Mech.}\ }\textbf {\bibinfo {volume} {2010}},\ \bibinfo {pages} {P08016} (\bibinfo {year} {2010})}\BibitemShut {NoStop}%
\bibitem [{\citenamefont {Prosen}\ and\ \citenamefont {Seligman}(2010)}]{prosen2010jphysa}%
  \BibitemOpen
  \bibfield  {author} {\bibinfo {author} {\bibfnamefont {T.}~\bibnamefont {Prosen}}\ and\ \bibinfo {author} {\bibfnamefont {T.~H.}\ \bibnamefont {Seligman}},\ }\bibfield  {title} {\bibinfo {title} {{Quantization over boson operator spaces}},\ }\href {https://doi.org/10.1088/1751-8113/43/39/392004} {\bibfield  {journal} {\bibinfo  {journal} {J. Phys. A}\ }\textbf {\bibinfo {volume} {43}},\ \bibinfo {pages} {392004} (\bibinfo {year} {2010})}\BibitemShut {NoStop}%
\bibitem [{\citenamefont {{\v{Z}}nidari{\v{c}}}(2010{\natexlab{a}})}]{znidaric2010jphysa}%
  \BibitemOpen
  \bibfield  {author} {\bibinfo {author} {\bibfnamefont {M.}~\bibnamefont {{\v{Z}}nidari{\v{c}}}},\ }\bibfield  {title} {\bibinfo {title} {{A matrix product solution for a nonequilibrium steady state of an XX chain}},\ }\href {https://doi.org/10.1088/1751-8113/43/41/415004} {\bibfield  {journal} {\bibinfo  {journal} {J. Phys. A}\ }\textbf {\bibinfo {volume} {43}},\ \bibinfo {pages} {415004} (\bibinfo {year} {2010}{\natexlab{a}})}\BibitemShut {NoStop}%
\bibitem [{\citenamefont {Prosen}\ and\ \citenamefont {Ilievski}(2011)}]{prosen2011ilievski}%
  \BibitemOpen
  \bibfield  {author} {\bibinfo {author} {\bibfnamefont {T.}~\bibnamefont {Prosen}}\ and\ \bibinfo {author} {\bibfnamefont {E.}~\bibnamefont {Ilievski}},\ }\bibfield  {title} {\bibinfo {title} {{Nonequilibrium Phase Transition in a Periodically Driven $XY$ Spin Chain}},\ }\href {https://doi.org/10.1103/PhysRevLett.107.060403} {\bibfield  {journal} {\bibinfo  {journal} {Phys. Rev. Lett.}\ }\textbf {\bibinfo {volume} {107}},\ \bibinfo {pages} {060403} (\bibinfo {year} {2011})}\BibitemShut {NoStop}%
\bibitem [{\citenamefont {{\v{Z}}unkovi{\v{c}}}(2014)}]{zunkovic2014}%
  \BibitemOpen
  \bibfield  {author} {\bibinfo {author} {\bibfnamefont {B.}~\bibnamefont {{\v{Z}}unkovi{\v{c}}}},\ }\bibfield  {title} {\bibinfo {title} {{Closed hierarchy of correlations in Markovian open quantum systems}},\ }\href {https://doi.org/10.1088/1367-2630/16/1/013042} {\bibfield  {journal} {\bibinfo  {journal} {New J. Phys.}\ }\textbf {\bibinfo {volume} {16}},\ \bibinfo {pages} {013042} (\bibinfo {year} {2014})}\BibitemShut {NoStop}%
\bibitem [{\citenamefont {{\v{Z}}nidari{\v{c}}}(2010{\natexlab{b}})}]{znidaric2010}%
  \BibitemOpen
  \bibfield  {author} {\bibinfo {author} {\bibfnamefont {M.}~\bibnamefont {{\v{Z}}nidari{\v{c}}}},\ }\bibfield  {title} {\bibinfo {title} {{Exact solution for a diffusive nonequilibrium steady state of an open quantum chain}},\ }\href {https://doi.org/10.1088/1742-5468/2010/05/L05002} {\bibfield  {journal} {\bibinfo  {journal} {J. Stat. Mech.}\ }\textbf {\bibinfo {volume} {2010}},\ \bibinfo {pages} {L05002} (\bibinfo {year} {2010}{\natexlab{b}})}\BibitemShut {NoStop}%
\bibitem [{\citenamefont {{\v{Z}}nidari{\v{c}}}(2011)}]{znidaric2011}%
  \BibitemOpen
  \bibfield  {author} {\bibinfo {author} {\bibfnamefont {M.}~\bibnamefont {{\v{Z}}nidari{\v{c}}}},\ }\bibfield  {title} {\bibinfo {title} {{Solvable quantum nonequilibrium model exhibiting a phase transition and a matrix product representation}},\ }\href {https://doi.org/10.1103/PhysRevE.83.011108} {\bibfield  {journal} {\bibinfo  {journal} {Phys. Rev. E}\ }\textbf {\bibinfo {volume} {83}},\ \bibinfo {pages} {011108} (\bibinfo {year} {2011})}\BibitemShut {NoStop}%
\bibitem [{\citenamefont {Temme}\ \emph {et~al.}(2012)\citenamefont {Temme}, \citenamefont {Wolf},\ and\ \citenamefont {Verstraete}}]{temme2012}%
  \BibitemOpen
  \bibfield  {author} {\bibinfo {author} {\bibfnamefont {K.}~\bibnamefont {Temme}}, \bibinfo {author} {\bibfnamefont {M.~M.}\ \bibnamefont {Wolf}},\ and\ \bibinfo {author} {\bibfnamefont {F.}~\bibnamefont {Verstraete}},\ }\bibfield  {title} {\bibinfo {title} {{Stochastic exclusion processes versus coherent transport}},\ }\href {https://doi.org/10.1088/1367-2630/14/7/075004} {\bibfield  {journal} {\bibinfo  {journal} {New J. Phys.}\ }\textbf {\bibinfo {volume} {14}},\ \bibinfo {pages} {075004} (\bibinfo {year} {2012})}\BibitemShut {NoStop}%
\bibitem [{\citenamefont {Eisler}(2011)}]{eisler2011}%
  \BibitemOpen
  \bibfield  {author} {\bibinfo {author} {\bibfnamefont {V.}~\bibnamefont {Eisler}},\ }\bibfield  {title} {\bibinfo {title} {{Crossover between ballistic and diffusive transport: the quantum exclusion process}},\ }\href {https://doi.org/10.1088/1742-5468/2011/06/P06007} {\bibfield  {journal} {\bibinfo  {journal} {J. Stat. Mech.}\ }\textbf {\bibinfo {volume} {2011}},\ \bibinfo {pages} {P06007} (\bibinfo {year} {2011})}\BibitemShut {NoStop}%
\bibitem [{\citenamefont {Prosen}(2011{\natexlab{a}})}]{prosen2011a}%
  \BibitemOpen
  \bibfield  {author} {\bibinfo {author} {\bibfnamefont {T.}~\bibnamefont {Prosen}},\ }\bibfield  {title} {\bibinfo {title} {{Open $XXZ$ Spin Chain: Nonequilibrium Steady State and a Strict Bound on Ballistic Transport}},\ }\href {https://doi.org/10.1103/PhysRevLett.106.217206} {\bibfield  {journal} {\bibinfo  {journal} {Phys. Rev. Lett.}\ }\textbf {\bibinfo {volume} {106}},\ \bibinfo {pages} {217206} (\bibinfo {year} {2011}{\natexlab{a}})}\BibitemShut {NoStop}%
\bibitem [{\citenamefont {Prosen}(2011{\natexlab{b}})}]{prosen2011b}%
  \BibitemOpen
  \bibfield  {author} {\bibinfo {author} {\bibfnamefont {T.}~\bibnamefont {Prosen}},\ }\bibfield  {title} {\bibinfo {title} {{Exact Nonequilibrium Steady State of a Strongly Driven Open XXZ Chain}},\ }\href {http://dx.doi.org/10.1103/PhysRevLett.107.137201} {\bibfield  {journal} {\bibinfo  {journal} {Phys. Rev. Lett.}\ }\textbf {\bibinfo {volume} {107}},\ \bibinfo {pages} {137201} (\bibinfo {year} {2011}{\natexlab{b}})}\BibitemShut {NoStop}%
\bibitem [{\citenamefont {Prosen}\ and\ \citenamefont {Ilievski}(2013)}]{prosen2013prl}%
  \BibitemOpen
  \bibfield  {author} {\bibinfo {author} {\bibfnamefont {T.}~\bibnamefont {Prosen}}\ and\ \bibinfo {author} {\bibfnamefont {E.}~\bibnamefont {Ilievski}},\ }\bibfield  {title} {\bibinfo {title} {{Families of Quasilocal Conservation Laws and Quantum Spin Transport}},\ }\href {https://doi.org/10.1103/PhysRevLett.111.057203} {\bibfield  {journal} {\bibinfo  {journal} {Phys. Rev. Lett.}\ }\textbf {\bibinfo {volume} {111}},\ \bibinfo {pages} {057203} (\bibinfo {year} {2013})}\BibitemShut {NoStop}%
\bibitem [{\citenamefont {Karevski}\ \emph {et~al.}(2013)\citenamefont {Karevski}, \citenamefont {Popkov},\ and\ \citenamefont {Sch{\"u}tz}}]{karevski2013}%
  \BibitemOpen
  \bibfield  {author} {\bibinfo {author} {\bibfnamefont {D.}~\bibnamefont {Karevski}}, \bibinfo {author} {\bibfnamefont {V.}~\bibnamefont {Popkov}},\ and\ \bibinfo {author} {\bibfnamefont {G.~M.}\ \bibnamefont {Sch{\"u}tz}},\ }\bibfield  {title} {\bibinfo {title} {{Exact Matrix Product Solution for the Boundary-Driven Lindblad XXZ Chain}},\ }\href {http://dx.doi.org/10.1103/PhysRevLett.110.047201} {\bibfield  {journal} {\bibinfo  {journal} {Phys. Rev. Lett.}\ }\textbf {\bibinfo {volume} {110}},\ \bibinfo {pages} {047201} (\bibinfo {year} {2013})}\BibitemShut {NoStop}%
\bibitem [{\citenamefont {Prosen}(2014)}]{prosen2014}%
  \BibitemOpen
  \bibfield  {author} {\bibinfo {author} {\bibfnamefont {T.}~\bibnamefont {Prosen}},\ }\bibfield  {title} {\bibinfo {title} {{Exact nonequilibrium steady state of an open Hubbard chain}},\ }\href {http://dx.doi.org/10.1103/PhysRevLett.112.030603} {\bibfield  {journal} {\bibinfo  {journal} {Phys. Rev. Lett.}\ }\textbf {\bibinfo {volume} {112}},\ \bibinfo {pages} {030603} (\bibinfo {year} {2014})}\BibitemShut {NoStop}%
\bibitem [{\citenamefont {Popkov}\ and\ \citenamefont {Prosen}(2015)}]{popkov2015}%
  \BibitemOpen
  \bibfield  {author} {\bibinfo {author} {\bibfnamefont {V.}~\bibnamefont {Popkov}}\ and\ \bibinfo {author} {\bibfnamefont {T.}~\bibnamefont {Prosen}},\ }\bibfield  {title} {\bibinfo {title} {{Infinitely Dimensional Lax Structure for the One-Dimensional Hubbard Model}},\ }\href {https://doi.org/10.1103/PhysRevLett.114.127201} {\bibfield  {journal} {\bibinfo  {journal} {Phys. Rev. Lett.}\ }\textbf {\bibinfo {volume} {114}},\ \bibinfo {pages} {127201} (\bibinfo {year} {2015})}\BibitemShut {NoStop}%
\bibitem [{\citenamefont {Prosen}(2015)}]{prosen2015REVIEW}%
  \BibitemOpen
  \bibfield  {author} {\bibinfo {author} {\bibfnamefont {T.}~\bibnamefont {Prosen}},\ }\bibfield  {title} {\bibinfo {title} {{Matrix product solutions of boundary driven quantum chains}},\ }\href {https://doi.org/10.1088/1751-8113/48/37/373001} {\bibfield  {journal} {\bibinfo  {journal} {J. Phys. A}\ }\textbf {\bibinfo {volume} {48}},\ \bibinfo {pages} {373001} (\bibinfo {year} {2015})}\BibitemShut {NoStop}%
\bibitem [{\citenamefont {Bu{\v{c}}a}\ and\ \citenamefont {Prosen}(2018)}]{buca2018}%
  \BibitemOpen
  \bibfield  {author} {\bibinfo {author} {\bibfnamefont {B.}~\bibnamefont {Bu{\v{c}}a}}\ and\ \bibinfo {author} {\bibfnamefont {T.}~\bibnamefont {Prosen}},\ }\bibfield  {title} {\bibinfo {title} {{Strongly correlated non-equilibrium steady states with currents--quantum and classical picture}},\ }\href {https://doi.org/10.1140/epjst/e2018-00100-9} {\bibfield  {journal} {\bibinfo  {journal} {Eur. Phys. J. Spec. Top.}\ }\textbf {\bibinfo {volume} {227}},\ \bibinfo {pages} {421} (\bibinfo {year} {2018})}\BibitemShut {NoStop}%
\bibitem [{\citenamefont {Popkov}\ \emph {et~al.}(2020)\citenamefont {Popkov}, \citenamefont {Prosen},\ and\ \citenamefont {Zadnik}}]{popkov2020prl}%
  \BibitemOpen
  \bibfield  {author} {\bibinfo {author} {\bibfnamefont {V.}~\bibnamefont {Popkov}}, \bibinfo {author} {\bibfnamefont {T.}~\bibnamefont {Prosen}},\ and\ \bibinfo {author} {\bibfnamefont {L.}~\bibnamefont {Zadnik}},\ }\bibfield  {title} {\bibinfo {title} {{Exact Nonequilibrium Steady State of Open $XXZ/XYZ$ Spin-$1/2$ Chain with Dirichlet Boundary Conditions}},\ }\href {https://doi.org/10.1103/PhysRevLett.124.160403} {\bibfield  {journal} {\bibinfo  {journal} {Phys. Rev. Lett.}\ }\textbf {\bibinfo {volume} {124}},\ \bibinfo {pages} {160403} (\bibinfo {year} {2020})}\BibitemShut {NoStop}%
\bibitem [{\citenamefont {Medvedyeva}\ \emph {et~al.}(2016)\citenamefont {Medvedyeva}, \citenamefont {Essler},\ and\ \citenamefont {Prosen}}]{medvedyeva2016}%
  \BibitemOpen
  \bibfield  {author} {\bibinfo {author} {\bibfnamefont {M.~V.}\ \bibnamefont {Medvedyeva}}, \bibinfo {author} {\bibfnamefont {F.~H.~L.}\ \bibnamefont {Essler}},\ and\ \bibinfo {author} {\bibfnamefont {T.}~\bibnamefont {Prosen}},\ }\bibfield  {title} {\bibinfo {title} {{Exact Bethe ansatz spectrum of a tight-binding chain with dephasing noise}},\ }\href {http://dx.doi.org/10.1103/PhysRevLett.117.137202} {\bibfield  {journal} {\bibinfo  {journal} {Phys. Rev. Lett.}\ }\textbf {\bibinfo {volume} {117}},\ \bibinfo {pages} {137202} (\bibinfo {year} {2016})}\BibitemShut {NoStop}%
\bibitem [{\citenamefont {Rowlands}\ and\ \citenamefont {Lamacraft}(2018)}]{rowlands2018}%
  \BibitemOpen
  \bibfield  {author} {\bibinfo {author} {\bibfnamefont {D.~A.}\ \bibnamefont {Rowlands}}\ and\ \bibinfo {author} {\bibfnamefont {A.}~\bibnamefont {Lamacraft}},\ }\bibfield  {title} {\bibinfo {title} {{Noisy Spins and the Richardson-Gaudin Model}},\ }\href {http://dx.doi.org/10.1103/PhysRevLett.120.090401} {\bibfield  {journal} {\bibinfo  {journal} {Phys. Rev. Lett.}\ }\textbf {\bibinfo {volume} {120}},\ \bibinfo {pages} {090401} (\bibinfo {year} {2018})}\BibitemShut {NoStop}%
\bibitem [{\citenamefont {Shibata}\ and\ \citenamefont {Katsura}(2019{\natexlab{a}})}]{shibata2019a}%
  \BibitemOpen
  \bibfield  {author} {\bibinfo {author} {\bibfnamefont {N.}~\bibnamefont {Shibata}}\ and\ \bibinfo {author} {\bibfnamefont {H.}~\bibnamefont {Katsura}},\ }\bibfield  {title} {\bibinfo {title} {{Dissipative spin chain as a non-Hermitian Kitaev ladder}},\ }\href {https://doi.org/10.1103/PhysRevB.99.174303} {\bibfield  {journal} {\bibinfo  {journal} {Phys. Rev. B}\ }\textbf {\bibinfo {volume} {99}},\ \bibinfo {pages} {174303} (\bibinfo {year} {2019}{\natexlab{a}})}\BibitemShut {NoStop}%
\bibitem [{\citenamefont {Shibata}\ and\ \citenamefont {Katsura}(2019{\natexlab{b}})}]{shibata2019b}%
  \BibitemOpen
  \bibfield  {author} {\bibinfo {author} {\bibfnamefont {N.}~\bibnamefont {Shibata}}\ and\ \bibinfo {author} {\bibfnamefont {H.}~\bibnamefont {Katsura}},\ }\bibfield  {title} {\bibinfo {title} {{Dissipative quantum Ising chain as a non-Hermitian Ashkin-Teller model}},\ }\href {https://doi.org/10.1103/PhysRevB.99.224432} {\bibfield  {journal} {\bibinfo  {journal} {Phys. Rev. B}\ }\textbf {\bibinfo {volume} {99}},\ \bibinfo {pages} {224432} (\bibinfo {year} {2019}{\natexlab{b}})}\BibitemShut {NoStop}%
\bibitem [{\citenamefont {Ziolkowska}\ and\ \citenamefont {Essler}(2020)}]{ziolkowska2020SciPost}%
  \BibitemOpen
  \bibfield  {author} {\bibinfo {author} {\bibfnamefont {A.~A.}\ \bibnamefont {Ziolkowska}}\ and\ \bibinfo {author} {\bibfnamefont {F.~H.~L.}\ \bibnamefont {Essler}},\ }\bibfield  {title} {\bibinfo {title} {{Yang-Baxter integrable Lindblad equations}},\ }\href {https://doi.org/10.21468/SciPostPhys.8.3.044} {\bibfield  {journal} {\bibinfo  {journal} {SciPost Phys.}\ }\textbf {\bibinfo {volume} {8}},\ \bibinfo {pages} {044} (\bibinfo {year} {2020})}\BibitemShut {NoStop}%
\bibitem [{\citenamefont {Ribeiro}\ and\ \citenamefont {Prosen}(2019)}]{ribeiro2019}%
  \BibitemOpen
  \bibfield  {author} {\bibinfo {author} {\bibfnamefont {P.}~\bibnamefont {Ribeiro}}\ and\ \bibinfo {author} {\bibfnamefont {T.}~\bibnamefont {Prosen}},\ }\bibfield  {title} {\bibinfo {title} {{Integrable Quantum Dynamics of Open Collective Spin Models}},\ }\href {http://dx.doi.org/10.1103/PhysRevLett.122.010401} {\bibfield  {journal} {\bibinfo  {journal} {Phys. Rev. Lett.}\ }\textbf {\bibinfo {volume} {122}},\ \bibinfo {pages} {010401} (\bibinfo {year} {2019})}\BibitemShut {NoStop}%
\bibitem [{\citenamefont {Torres}(2014)}]{torres2014PRA}%
  \BibitemOpen
  \bibfield  {author} {\bibinfo {author} {\bibfnamefont {J.~M.}\ \bibnamefont {Torres}},\ }\bibfield  {title} {\bibinfo {title} {{Closed-form solution of Lindblad master equations without gain}},\ }\href {https://doi.org/10.1103/PhysRevA.89.052133} {\bibfield  {journal} {\bibinfo  {journal} {Phys. Rev. A}\ }\textbf {\bibinfo {volume} {89}},\ \bibinfo {pages} {052133} (\bibinfo {year} {2014})}\BibitemShut {NoStop}%
\bibitem [{\citenamefont {Nakagawa}\ \emph {et~al.}(2021)\citenamefont {Nakagawa}, \citenamefont {Kawakami},\ and\ \citenamefont {Ueda}}]{nakagawa2020PRL}%
  \BibitemOpen
  \bibfield  {author} {\bibinfo {author} {\bibfnamefont {M.}~\bibnamefont {Nakagawa}}, \bibinfo {author} {\bibfnamefont {N.}~\bibnamefont {Kawakami}},\ and\ \bibinfo {author} {\bibfnamefont {M.}~\bibnamefont {Ueda}},\ }\bibfield  {title} {\bibinfo {title} {{Exact Liouvillian Spectrum of a One-Dimensional Dissipative Hubbard Model}},\ }\href {https://doi.org/10.1103/PhysRevLett.126.110404} {\bibfield  {journal} {\bibinfo  {journal} {Phys. Rev. Lett.}\ }\textbf {\bibinfo {volume} {126}},\ \bibinfo {pages} {110404} (\bibinfo {year} {2021})}\BibitemShut {NoStop}%
\bibitem [{\citenamefont {Bu\v{c}a}\ \emph {et~al.}(2020)\citenamefont {Bu\v{c}a}, \citenamefont {Booker}, \citenamefont {Medenjak},\ and\ \citenamefont {Jaksch}}]{buca2020NJP}%
  \BibitemOpen
  \bibfield  {author} {\bibinfo {author} {\bibfnamefont {B.}~\bibnamefont {Bu\v{c}a}}, \bibinfo {author} {\bibfnamefont {C.}~\bibnamefont {Booker}}, \bibinfo {author} {\bibfnamefont {M.}~\bibnamefont {Medenjak}},\ and\ \bibinfo {author} {\bibfnamefont {D.}~\bibnamefont {Jaksch}},\ }\bibfield  {title} {\bibinfo {title} {{Bethe ansatz approach for dissipation: exact solutions of quantum many-body dynamics under loss}},\ }\href {https://doi.org/10.1088/1367-2630/abd124} {\bibfield  {journal} {\bibinfo  {journal} {New J. Phys.}\ }\textbf {\bibinfo {volume} {22}},\ \bibinfo {pages} {123040} (\bibinfo {year} {2020})}\BibitemShut {NoStop}%
\bibitem [{\citenamefont {S\'a}\ \emph {et~al.}(2021)\citenamefont {S\'a}, \citenamefont {Ribeiro},\ and\ \citenamefont {Prosen}}]{sa2021PRB}%
  \BibitemOpen
  \bibfield  {author} {\bibinfo {author} {\bibfnamefont {L.}~\bibnamefont {S\'a}}, \bibinfo {author} {\bibfnamefont {P.}~\bibnamefont {Ribeiro}},\ and\ \bibinfo {author} {\bibfnamefont {T.}~\bibnamefont {Prosen}},\ }\bibfield  {title} {\bibinfo {title} {{Integrable nonunitary open quantum circuits}},\ }\href {https://doi.org/10.1103/PhysRevB.103.115132} {\bibfield  {journal} {\bibinfo  {journal} {Phys. Rev. B}\ }\textbf {\bibinfo {volume} {103}},\ \bibinfo {pages} {115132} (\bibinfo {year} {2021})}\BibitemShut {NoStop}%
\bibitem [{\citenamefont {de~Leeuw}\ \emph {et~al.}(2021)\citenamefont {de~Leeuw}, \citenamefont {Paletta},\ and\ \citenamefont {Pozsgay}}]{deleeuw2021PRL}%
  \BibitemOpen
  \bibfield  {author} {\bibinfo {author} {\bibfnamefont {M.}~\bibnamefont {de~Leeuw}}, \bibinfo {author} {\bibfnamefont {C.}~\bibnamefont {Paletta}},\ and\ \bibinfo {author} {\bibfnamefont {B.}~\bibnamefont {Pozsgay}},\ }\bibfield  {title} {\bibinfo {title} {{Constructing Integrable Lindblad Superoperators}},\ }\href {https://doi.org/10.1103/PhysRevLett.126.240403} {\bibfield  {journal} {\bibinfo  {journal} {Phys. Rev. Lett.}\ }\textbf {\bibinfo {volume} {126}},\ \bibinfo {pages} {240403} (\bibinfo {year} {2021})}\BibitemShut {NoStop}%
\bibitem [{\citenamefont {de~Leeuw}\ and\ \citenamefont {Paletta}(2022)}]{deleeuw2022ARXIV}%
  \BibitemOpen
  \bibfield  {author} {\bibinfo {author} {\bibfnamefont {M.}~\bibnamefont {de~Leeuw}}\ and\ \bibinfo {author} {\bibfnamefont {C.}~\bibnamefont {Paletta}},\ }\bibfield  {title} {\bibinfo {title} {{The Bethe ansatz for a new integrable open quantum system}},\ }\href {https://arxiv.org/abs/2207.14193} {\bibfield  {journal} {\bibinfo  {journal} {arXiv:2207.14193}\ } (\bibinfo {year} {2022})}\BibitemShut {NoStop}%
\bibitem [{\citenamefont {Su}\ and\ \citenamefont {Martin}(2022)}]{su2022PRB}%
  \BibitemOpen
  \bibfield  {author} {\bibinfo {author} {\bibfnamefont {L.}~\bibnamefont {Su}}\ and\ \bibinfo {author} {\bibfnamefont {I.}~\bibnamefont {Martin}},\ }\bibfield  {title} {\bibinfo {title} {{Integrable nonunitary quantum circuits}},\ }\href {https://doi.org/10.1103/PhysRevB.106.134312} {\bibfield  {journal} {\bibinfo  {journal} {Phys. Rev. B}\ }\textbf {\bibinfo {volume} {106}},\ \bibinfo {pages} {134312} (\bibinfo {year} {2022})}\BibitemShut {NoStop}%
\bibitem [{\citenamefont {\v{Z}nidari\v{c}}\ \emph {et~al.}(2025)\citenamefont {\v{Z}nidari\v{c}}, \citenamefont {Duh},\ and\ \citenamefont {Zadnik}}]{znidaric2024b}%
  \BibitemOpen
  \bibfield  {author} {\bibinfo {author} {\bibfnamefont {M.}~\bibnamefont {\v{Z}nidari\v{c}}}, \bibinfo {author} {\bibfnamefont {U.}~\bibnamefont {Duh}},\ and\ \bibinfo {author} {\bibfnamefont {L.}~\bibnamefont {Zadnik}},\ }\bibfield  {title} {\bibinfo {title} {{Integrability is generic in homogeneous U(1)-invariant nearest-neighbor qubit circuits}},\ }\href {https://doi.org/10.1103/tqy8-ynpd} {\bibfield  {journal} {\bibinfo  {journal} {Phys. Rev. B}\ }\textbf {\bibinfo {volume} {112}},\ \bibinfo {pages} {L020302} (\bibinfo {year} {2025})}\BibitemShut {NoStop}%
\bibitem [{\citenamefont {Wang}\ \emph {et~al.}(2020)\citenamefont {Wang}, \citenamefont {Piazza},\ and\ \citenamefont {Luitz}}]{wang2020PRL}%
  \BibitemOpen
  \bibfield  {author} {\bibinfo {author} {\bibfnamefont {K.}~\bibnamefont {Wang}}, \bibinfo {author} {\bibfnamefont {F.}~\bibnamefont {Piazza}},\ and\ \bibinfo {author} {\bibfnamefont {D.~J.}\ \bibnamefont {Luitz}},\ }\bibfield  {title} {\bibinfo {title} {{Hierarchy of Relaxation Timescales in Local Random Liouvillians}},\ }\href {https://doi.org/10.1103/PhysRevLett.124.100604} {\bibfield  {journal} {\bibinfo  {journal} {Phys. Rev. Lett.}\ }\textbf {\bibinfo {volume} {124}},\ \bibinfo {pages} {100604} (\bibinfo {year} {2020})}\BibitemShut {NoStop}%
\bibitem [{\citenamefont {Sommer}\ \emph {et~al.}(2021)\citenamefont {Sommer}, \citenamefont {Piazza},\ and\ \citenamefont {Luitz}}]{sommer2021PRR}%
  \BibitemOpen
  \bibfield  {author} {\bibinfo {author} {\bibfnamefont {O.~E.}\ \bibnamefont {Sommer}}, \bibinfo {author} {\bibfnamefont {F.}~\bibnamefont {Piazza}},\ and\ \bibinfo {author} {\bibfnamefont {D.~J.}\ \bibnamefont {Luitz}},\ }\bibfield  {title} {\bibinfo {title} {{Many-body hierarchy of dissipative timescales in a quantum computer}},\ }\href {https://doi.org/10.1103/PhysRevResearch.3.023190} {\bibfield  {journal} {\bibinfo  {journal} {Phys. Rev. Res.}\ }\textbf {\bibinfo {volume} {3}},\ \bibinfo {pages} {023190} (\bibinfo {year} {2021})}\BibitemShut {NoStop}%
\bibitem [{\citenamefont {Li}\ \emph {et~al.}(2022)\citenamefont {Li}, \citenamefont {Rose}, \citenamefont {Garrahan},\ and\ \citenamefont {Luitz}}]{li2022PRB}%
  \BibitemOpen
  \bibfield  {author} {\bibinfo {author} {\bibfnamefont {J.~L.}\ \bibnamefont {Li}}, \bibinfo {author} {\bibfnamefont {D.~C.}\ \bibnamefont {Rose}}, \bibinfo {author} {\bibfnamefont {J.~P.}\ \bibnamefont {Garrahan}},\ and\ \bibinfo {author} {\bibfnamefont {D.~J.}\ \bibnamefont {Luitz}},\ }\bibfield  {title} {\bibinfo {title} {{Random matrix theory for quantum and classical metastability in local Liouvillians}},\ }\href {https://doi.org/10.1103/PhysRevB.105.L180201} {\bibfield  {journal} {\bibinfo  {journal} {Phys. Rev. B}\ }\textbf {\bibinfo {volume} {105}},\ \bibinfo {pages} {L180201} (\bibinfo {year} {2022})}\BibitemShut {NoStop}%
\bibitem [{\citenamefont {Hartmann}\ \emph {et~al.}(2024)\citenamefont {Hartmann}, \citenamefont {Li},\ and\ \citenamefont {Luitz}}]{hartmann2024}%
  \BibitemOpen
  \bibfield  {author} {\bibinfo {author} {\bibfnamefont {N.~D.}\ \bibnamefont {Hartmann}}, \bibinfo {author} {\bibfnamefont {J.~L.}\ \bibnamefont {Li}},\ and\ \bibinfo {author} {\bibfnamefont {D.~J.}\ \bibnamefont {Luitz}},\ }\bibfield  {title} {\bibinfo {title} {{Fate of dissipative hierarchy of timescales in the presence of unitary dynamics}},\ }\href {https://doi.org/10.1103/PhysRevB.109.054203} {\bibfield  {journal} {\bibinfo  {journal} {Phys. Rev. B}\ }\textbf {\bibinfo {volume} {109}},\ \bibinfo {pages} {054203} (\bibinfo {year} {2024})}\BibitemShut {NoStop}%
\bibitem [{\citenamefont {Popkov}\ and\ \citenamefont {Presilla}(2021)}]{popkov2021PRL}%
  \BibitemOpen
  \bibfield  {author} {\bibinfo {author} {\bibfnamefont {V.}~\bibnamefont {Popkov}}\ and\ \bibinfo {author} {\bibfnamefont {C.}~\bibnamefont {Presilla}},\ }\bibfield  {title} {\bibinfo {title} {{Full Spectrum of the Liouvillian of Open Dissipative Quantum Systems in the Zeno Limit}},\ }\href {https://doi.org/10.1103/PhysRevLett.126.190402} {\bibfield  {journal} {\bibinfo  {journal} {Phys. Rev. Lett.}\ }\textbf {\bibinfo {volume} {126}},\ \bibinfo {pages} {190402} (\bibinfo {year} {2021})}\BibitemShut {NoStop}%
\bibitem [{\citenamefont {Zhou}\ \emph {et~al.}(2021)\citenamefont {Zhou}, \citenamefont {Mao},\ and\ \citenamefont {Zhai}}]{Zhou2021PRR}%
  \BibitemOpen
  \bibfield  {author} {\bibinfo {author} {\bibfnamefont {Y.-N.}\ \bibnamefont {Zhou}}, \bibinfo {author} {\bibfnamefont {L.}~\bibnamefont {Mao}},\ and\ \bibinfo {author} {\bibfnamefont {H.}~\bibnamefont {Zhai}},\ }\bibfield  {title} {\bibinfo {title} {{R\'enyi entropy dynamics and Lindblad spectrum for open quantum systems}},\ }\href {https://doi.org/10.1103/PhysRevResearch.3.043060} {\bibfield  {journal} {\bibinfo  {journal} {Phys. Rev. Res.}\ }\textbf {\bibinfo {volume} {3}},\ \bibinfo {pages} {043060} (\bibinfo {year} {2021})}\BibitemShut {NoStop}%
\bibitem [{\citenamefont {S\'a}\ \emph {et~al.}(2022)\citenamefont {S\'a}, \citenamefont {Ribeiro},\ and\ \citenamefont {Prosen}}]{sa2022PRR}%
  \BibitemOpen
  \bibfield  {author} {\bibinfo {author} {\bibfnamefont {L.}~\bibnamefont {S\'a}}, \bibinfo {author} {\bibfnamefont {P.}~\bibnamefont {Ribeiro}},\ and\ \bibinfo {author} {\bibfnamefont {T.}~\bibnamefont {Prosen}},\ }\bibfield  {title} {\bibinfo {title} {{Lindbladian dissipation of strongly-correlated quantum matter}},\ }\href {https://doi.org/10.1103/PhysRevResearch.4.L022068} {\bibfield  {journal} {\bibinfo  {journal} {Phys. Rev. Res.}\ }\textbf {\bibinfo {volume} {4}},\ \bibinfo {pages} {L022068} (\bibinfo {year} {2022})}\BibitemShut {NoStop}%
\bibitem [{\citenamefont {Garc\'{\i}a-Garc\'{\i}a}\ \emph {et~al.}(2023)\citenamefont {Garc\'{\i}a-Garc\'{\i}a}, \citenamefont {S\'a}, \citenamefont {Verbaarschot},\ and\ \citenamefont {Zheng}}]{garcia2023PRD2}%
  \BibitemOpen
  \bibfield  {author} {\bibinfo {author} {\bibfnamefont {A.~M.}\ \bibnamefont {Garc\'{\i}a-Garc\'{\i}a}}, \bibinfo {author} {\bibfnamefont {L.}~\bibnamefont {S\'a}}, \bibinfo {author} {\bibfnamefont {J.~J.~M.}\ \bibnamefont {Verbaarschot}},\ and\ \bibinfo {author} {\bibfnamefont {J.~P.}\ \bibnamefont {Zheng}},\ }\bibfield  {title} {\bibinfo {title} {{Keldysh wormholes and anomalous relaxation in the dissipative Sachdev-Ye-Kitaev model}},\ }\href {https://doi.org/10.1103/PhysRevD.107.106006} {\bibfield  {journal} {\bibinfo  {journal} {Phys. Rev. D}\ }\textbf {\bibinfo {volume} {107}},\ \bibinfo {pages} {106006} (\bibinfo {year} {2023})}\BibitemShut {NoStop}%
\bibitem [{\citenamefont {Mori}(2024)}]{mori2024}%
  \BibitemOpen
  \bibfield  {author} {\bibinfo {author} {\bibfnamefont {T.}~\bibnamefont {Mori}},\ }\bibfield  {title} {\bibinfo {title} {{Liouvillian-gap analysis of open quantum many-body systems in the weak dissipation limit}},\ }\href {https://doi.org/10.1103/PhysRevB.109.064311} {\bibfield  {journal} {\bibinfo  {journal} {Phys. Rev. B}\ }\textbf {\bibinfo {volume} {109}},\ \bibinfo {pages} {064311} (\bibinfo {year} {2024})}\BibitemShut {NoStop}%
\bibitem [{\citenamefont {Yoshimura}\ and\ \citenamefont {Sá}(2024)}]{yoshimura2024}%
  \BibitemOpen
  \bibfield  {author} {\bibinfo {author} {\bibfnamefont {T.}~\bibnamefont {Yoshimura}}\ and\ \bibinfo {author} {\bibfnamefont {L.}~\bibnamefont {Sá}},\ }\bibfield  {title} {\bibinfo {title} {{Robustness of quantum chaos and anomalous relaxation in open quantum circuits}},\ }\href {http://dx.doi.org/10.1038/s41467-024-54164-7} {\bibfield  {journal} {\bibinfo  {journal} {Nat. Commun.}\ }\textbf {\bibinfo {volume} {15}},\ \bibinfo {pages} {9808} (\bibinfo {year} {2024})}\BibitemShut {NoStop}%
\bibitem [{\citenamefont {Yoshimura}\ and\ \citenamefont {S\'a}(2025)}]{yoshimura2025}%
  \BibitemOpen
  \bibfield  {author} {\bibinfo {author} {\bibfnamefont {T.}~\bibnamefont {Yoshimura}}\ and\ \bibinfo {author} {\bibfnamefont {L.}~\bibnamefont {S\'a}},\ }\bibfield  {title} {\bibinfo {title} {{Theory of irreversibility in quantum many-body systems}},\ }\href {https://doi.org/10.1103/82f6-qdyd} {\bibfield  {journal} {\bibinfo  {journal} {Phys. Rev. E}\ }\textbf {\bibinfo {volume} {111}},\ \bibinfo {pages} {064135} (\bibinfo {year} {2025})}\BibitemShut {NoStop}%
\bibitem [{\citenamefont {Bu{\v{c}}a}\ and\ \citenamefont {Prosen}(2012)}]{buca2012}%
  \BibitemOpen
  \bibfield  {author} {\bibinfo {author} {\bibfnamefont {B.}~\bibnamefont {Bu{\v{c}}a}}\ and\ \bibinfo {author} {\bibfnamefont {T.}~\bibnamefont {Prosen}},\ }\bibfield  {title} {\bibinfo {title} {{A note on symmetry reductions of the Lindblad equation: transport in constrained open spin chains}},\ }\href {https://doi.org/10.1088/1367-2630/14/7/073007} {\bibfield  {journal} {\bibinfo  {journal} {New J. Phys.}\ }\textbf {\bibinfo {volume} {14}},\ \bibinfo {pages} {073007} (\bibinfo {year} {2012})}\BibitemShut {NoStop}%
\bibitem [{\citenamefont {Gamayun}\ and\ \citenamefont {Lychkovskiy}(2022)}]{gamayun2022}%
  \BibitemOpen
  \bibfield  {author} {\bibinfo {author} {\bibfnamefont {O.}~\bibnamefont {Gamayun}}\ and\ \bibinfo {author} {\bibfnamefont {O.}~\bibnamefont {Lychkovskiy}},\ }\bibfield  {title} {\bibinfo {title} {{Out-of-equilibrium dynamics of the Kitaev model on the Bethe lattice via coupled Heisenberg equations}},\ }\href {https://doi.org/10.21468/SciPostPhys.12.5.175} {\bibfield  {journal} {\bibinfo  {journal} {SciPost Phys.}\ }\textbf {\bibinfo {volume} {12}},\ \bibinfo {pages} {175} (\bibinfo {year} {2022})}\BibitemShut {NoStop}%
\bibitem [{\citenamefont {Teretenkov}\ and\ \citenamefont {Lychkovskiy}(2024)}]{terentenkov2024}%
  \BibitemOpen
  \bibfield  {author} {\bibinfo {author} {\bibfnamefont {A.}~\bibnamefont {Teretenkov}}\ and\ \bibinfo {author} {\bibfnamefont {O.}~\bibnamefont {Lychkovskiy}},\ }\bibfield  {title} {\bibinfo {title} {{Exact dynamics of quantum dissipative XX models: Wannier-Stark localization in the fragmented operator space}},\ }\href {https://doi.org/10.1103/PhysRevB.109.L140302} {\bibfield  {journal} {\bibinfo  {journal} {Phys. Rev. B}\ }\textbf {\bibinfo {volume} {109}},\ \bibinfo {pages} {L140302} (\bibinfo {year} {2024})}\BibitemShut {NoStop}%
\bibitem [{\citenamefont {Belavin}\ \emph {et~al.}(1969)\citenamefont {Belavin}, \citenamefont {Zeldovich}, \citenamefont {Perelomov},\ and\ \citenamefont {Popov}}]{belavin1969}%
  \BibitemOpen
  \bibfield  {author} {\bibinfo {author} {\bibfnamefont {A.~A.}\ \bibnamefont {Belavin}}, \bibinfo {author} {\bibfnamefont {B.~Y.}\ \bibnamefont {Zeldovich}}, \bibinfo {author} {\bibfnamefont {A.~M.}\ \bibnamefont {Perelomov}},\ and\ \bibinfo {author} {\bibfnamefont {V.~S.}\ \bibnamefont {Popov}},\ }\bibfield  {title} {\bibinfo {title} {{Relaxation of quantum systems with equidistant spectra}},\ }\href {http://jetp.ras.ru/cgi-bin/e/index/e/29/1/p145?a=list} {\bibfield  {journal} {\bibinfo  {journal} {Sov. Phys. JETP}\ }\textbf {\bibinfo {volume} {29}},\ \bibinfo {pages} {145} (\bibinfo {year} {1969})}\BibitemShut {NoStop}%
\bibitem [{\citenamefont {Lindblad}(1976)}]{lindblad1976}%
  \BibitemOpen
  \bibfield  {author} {\bibinfo {author} {\bibfnamefont {G.}~\bibnamefont {Lindblad}},\ }\bibfield  {title} {\bibinfo {title} {{On the generators of quantum dynamical semigroups}},\ }\href {https://doi.org/https://doi.org/10.1007/BF01608499} {\bibfield  {journal} {\bibinfo  {journal} {Commun. Math. Phys.}\ }\textbf {\bibinfo {volume} {48}},\ \bibinfo {pages} {119} (\bibinfo {year} {1976})}\BibitemShut {NoStop}%
\bibitem [{\citenamefont {Gorini}\ \emph {et~al.}(1976)\citenamefont {Gorini}, \citenamefont {Kossakowski},\ and\ \citenamefont {Sudarshan}}]{gorini1976}%
  \BibitemOpen
  \bibfield  {author} {\bibinfo {author} {\bibfnamefont {V.}~\bibnamefont {Gorini}}, \bibinfo {author} {\bibfnamefont {A.}~\bibnamefont {Kossakowski}},\ and\ \bibinfo {author} {\bibfnamefont {E.~C.~G.}\ \bibnamefont {Sudarshan}},\ }\bibfield  {title} {\bibinfo {title} {{Completely positive dynamical semigroups of $N$-level systems}},\ }\href {https://doi.org/https://doi.org/10.1063/1.522979} {\bibfield  {journal} {\bibinfo  {journal} {J. Math. Phys.}\ }\textbf {\bibinfo {volume} {17}},\ \bibinfo {pages} {821} (\bibinfo {year} {1976})}\BibitemShut {NoStop}%
\bibitem [{\citenamefont {Wen}(1990)}]{wen1990}%
  \BibitemOpen
  \bibfield  {author} {\bibinfo {author} {\bibfnamefont {X.-G.}\ \bibnamefont {Wen}},\ }\bibfield  {title} {\bibinfo {title} {{Topological Orders in Rigid States}},\ }\href {https://www.worldscientific.com/doi/abs/10.1142/s0217979290000139} {\bibfield  {journal} {\bibinfo  {journal} {Int. J. Mod. Phys. B}\ }\textbf {\bibinfo {volume} {4}},\ \bibinfo {pages} {239} (\bibinfo {year} {1990})}\BibitemShut {NoStop}%
\bibitem [{\citenamefont {Wen}(1995)}]{wen1995}%
  \BibitemOpen
  \bibfield  {author} {\bibinfo {author} {\bibfnamefont {X.-G.}\ \bibnamefont {Wen}},\ }\bibfield  {title} {\bibinfo {title} {{Topological orders and edge excitations in fractional quantum Hall states}},\ }\href {https://doi.org/10.1080/00018739500101566} {\bibfield  {journal} {\bibinfo  {journal} {Adv. Phys.}\ }\textbf {\bibinfo {volume} {44}},\ \bibinfo {pages} {405} (\bibinfo {year} {1995})}\BibitemShut {NoStop}%
\bibitem [{\citenamefont {Wen}(2003)}]{wen2003PRL}%
  \BibitemOpen
  \bibfield  {author} {\bibinfo {author} {\bibfnamefont {X.-G.}\ \bibnamefont {Wen}},\ }\bibfield  {title} {\bibinfo {title} {{Quantum Orders in an Exact Soluble Model}},\ }\href {https://doi.org/10.1103/PhysRevLett.90.016803} {\bibfield  {journal} {\bibinfo  {journal} {Phys. Rev. Lett.}\ }\textbf {\bibinfo {volume} {90}},\ \bibinfo {pages} {016803} (\bibinfo {year} {2003})}\BibitemShut {NoStop}%
\bibitem [{\citenamefont {Levin}\ and\ \citenamefont {Wen}(2005)}]{levin2005PRB}%
  \BibitemOpen
  \bibfield  {author} {\bibinfo {author} {\bibfnamefont {M.~A.}\ \bibnamefont {Levin}}\ and\ \bibinfo {author} {\bibfnamefont {X.-G.}\ \bibnamefont {Wen}},\ }\bibfield  {title} {\bibinfo {title} {{String-net condensation: A physical mechanism for topological phases}},\ }\href {https://doi.org/10.1103/PhysRevB.71.045110} {\bibfield  {journal} {\bibinfo  {journal} {Phys. Rev. B}\ }\textbf {\bibinfo {volume} {71}},\ \bibinfo {pages} {045110} (\bibinfo {year} {2005})}\BibitemShut {NoStop}%
\bibitem [{\citenamefont {Wen}(2017)}]{wen2017RMP}%
  \BibitemOpen
  \bibfield  {author} {\bibinfo {author} {\bibfnamefont {X.-G.}\ \bibnamefont {Wen}},\ }\bibfield  {title} {\bibinfo {title} {{Colloquium: Zoo of quantum-topological phases of matter}},\ }\href {https://doi.org/10.1103/RevModPhys.89.041004} {\bibfield  {journal} {\bibinfo  {journal} {Rev. Mod. Phys.}\ }\textbf {\bibinfo {volume} {89}},\ \bibinfo {pages} {041004} (\bibinfo {year} {2017})}\BibitemShut {NoStop}%
\bibitem [{\citenamefont {Simon}(2023)}]{simon2023book}%
  \BibitemOpen
  \bibfield  {author} {\bibinfo {author} {\bibfnamefont {S.~H.}\ \bibnamefont {Simon}},\ }\href@noop {} {\emph {\bibinfo {title} {Topological Quantum}}}\ (\bibinfo  {publisher} {Oxford University Press},\ \bibinfo {address} {Oxford},\ \bibinfo {year} {2023})\BibitemShut {NoStop}%
\bibitem [{\citenamefont {Hastings}\ and\ \citenamefont {Wen}(2005)}]{hastings2005PRB}%
  \BibitemOpen
  \bibfield  {author} {\bibinfo {author} {\bibfnamefont {M.~B.}\ \bibnamefont {Hastings}}\ and\ \bibinfo {author} {\bibfnamefont {X.-G.}\ \bibnamefont {Wen}},\ }\bibfield  {title} {\bibinfo {title} {{Quasiadiabatic continuation of quantum states: The stability of topological ground-state degeneracy and emergent gauge invariance}},\ }\href {https://doi.org/10.1103/PhysRevB.72.045141} {\bibfield  {journal} {\bibinfo  {journal} {Phys. Rev. B}\ }\textbf {\bibinfo {volume} {72}},\ \bibinfo {pages} {045141} (\bibinfo {year} {2005})}\BibitemShut {NoStop}%
\bibitem [{\citenamefont {Nussinov}\ and\ \citenamefont {Ortiz}(2009)}]{nussinov2009}%
  \BibitemOpen
  \bibfield  {author} {\bibinfo {author} {\bibfnamefont {Z.}~\bibnamefont {Nussinov}}\ and\ \bibinfo {author} {\bibfnamefont {G.}~\bibnamefont {Ortiz}},\ }\bibfield  {title} {\bibinfo {title} {{A symmetry principle for topological quantum order}},\ }\href {https://doi.org/https://doi.org/10.1016/j.aop.2008.11.002} {\bibfield  {journal} {\bibinfo  {journal} {Ann. Phys.}\ }\textbf {\bibinfo {volume} {324}},\ \bibinfo {pages} {977} (\bibinfo {year} {2009})}\BibitemShut {NoStop}%
\bibitem [{\citenamefont {Gaiotto}\ \emph {et~al.}(2015)\citenamefont {Gaiotto}, \citenamefont {Kapustin}, \citenamefont {Seiberg},\ and\ \citenamefont {Willett}}]{gaiotto2015JHEP}%
  \BibitemOpen
  \bibfield  {author} {\bibinfo {author} {\bibfnamefont {D.}~\bibnamefont {Gaiotto}}, \bibinfo {author} {\bibfnamefont {A.}~\bibnamefont {Kapustin}}, \bibinfo {author} {\bibfnamefont {N.}~\bibnamefont {Seiberg}},\ and\ \bibinfo {author} {\bibfnamefont {B.}~\bibnamefont {Willett}},\ }\bibfield  {title} {\bibinfo {title} {{Generalized global symmetries}},\ }\href {https://link.springer.com/article/10.1007/JHEP02(2015)172} {\bibfield  {journal} {\bibinfo  {journal} {J. High Energy Phys{.}}\ }\textbf {\bibinfo {volume} {2015}},\ \bibinfo {pages} {172} (\bibinfo {year} {2015})}\BibitemShut {NoStop}%
\bibitem [{\citenamefont {McGreevy}(2023)}]{mcgreevy2023Annu}%
  \BibitemOpen
  \bibfield  {author} {\bibinfo {author} {\bibfnamefont {J.}~\bibnamefont {McGreevy}},\ }\bibfield  {title} {\bibinfo {title} {{Generalized Symmetries in Condensed Matter}},\ }\href {https://doi.org/https://doi.org/10.1146/annurev-conmatphys-040721-021029} {\bibfield  {journal} {\bibinfo  {journal} {Annu. Rev. Condens. Matter Phys.}\ }\textbf {\bibinfo {volume} {14}},\ \bibinfo {pages} {57} (\bibinfo {year} {2023})}\BibitemShut {NoStop}%
\bibitem [{\citenamefont {Bravyi}\ and\ \citenamefont {Terhal}(2009)}]{bravyi2009NJP}%
  \BibitemOpen
  \bibfield  {author} {\bibinfo {author} {\bibfnamefont {S.}~\bibnamefont {Bravyi}}\ and\ \bibinfo {author} {\bibfnamefont {B.}~\bibnamefont {Terhal}},\ }\bibfield  {title} {\bibinfo {title} {{A no-go theorem for a two-dimensional self-correcting quantum memory based on stabilizer codes}},\ }\href {https://doi.org/10.1088/1367-2630/11/4/043029} {\bibfield  {journal} {\bibinfo  {journal} {New J. Phys.}\ }\textbf {\bibinfo {volume} {11}},\ \bibinfo {pages} {043029} (\bibinfo {year} {2009})}\BibitemShut {NoStop}%
\bibitem [{\citenamefont {Kitaev}(2001)}]{kitaev2001Wire}%
  \BibitemOpen
  \bibfield  {author} {\bibinfo {author} {\bibfnamefont {A.~Y.}\ \bibnamefont {Kitaev}},\ }\bibfield  {title} {\bibinfo {title} {{Unpaired Majorana fermions in quantum wires}},\ }\href {https://doi.org/10.1070/1063-7869/44/10S/S29} {\bibfield  {journal} {\bibinfo  {journal} {Phys.-Uspekhi}\ }\textbf {\bibinfo {volume} {44}},\ \bibinfo {pages} {131} (\bibinfo {year} {2001})}\BibitemShut {NoStop}%
\bibitem [{\citenamefont {Prosen}(2002)}]{prosen2002}%
  \BibitemOpen
  \bibfield  {author} {\bibinfo {author} {\bibfnamefont {T.}~\bibnamefont {Prosen}},\ }\bibfield  {title} {\bibinfo {title} {{Ruelle resonances in quantum many-body dynamics}},\ }\href {https://doi.org/10.1088/0305-4470/35/48/102} {\bibfield  {journal} {\bibinfo  {journal} {J. Phys. A}\ }\textbf {\bibinfo {volume} {35}},\ \bibinfo {pages} {L737} (\bibinfo {year} {2002})}\BibitemShut {NoStop}%
\bibitem [{\citenamefont {Prosen}(2004)}]{prosen2004}%
  \BibitemOpen
  \bibfield  {author} {\bibinfo {author} {\bibfnamefont {T.}~\bibnamefont {Prosen}},\ }\bibfield  {title} {\bibinfo {title} {{Ruelle resonances in kicked quantum spin chain}},\ }\href {https://doi.org/https://doi.org/10.1016/j.physd.2003.09.017} {\bibfield  {journal} {\bibinfo  {journal} {Physica D}\ }\textbf {\bibinfo {volume} {187}},\ \bibinfo {pages} {244} (\bibinfo {year} {2004})}\BibitemShut {NoStop}%
\bibitem [{\citenamefont {Prosen}(2007)}]{prosen2007}%
  \BibitemOpen
  \bibfield  {author} {\bibinfo {author} {\bibfnamefont {T.}~\bibnamefont {Prosen}},\ }\bibfield  {title} {\bibinfo {title} {{Chaos and complexity of quantum motion}},\ }\href {https://doi.org/10.1088/1751-8113/40/28/S02} {\bibfield  {journal} {\bibinfo  {journal} {J. Phys. A}\ }\textbf {\bibinfo {volume} {40}},\ \bibinfo {pages} {7881} (\bibinfo {year} {2007})}\BibitemShut {NoStop}%
\bibitem [{\citenamefont {\v{Z}nidari\v{c}}(2024)}]{znidaric2024}%
  \BibitemOpen
  \bibfield  {author} {\bibinfo {author} {\bibfnamefont {M.}~\bibnamefont {\v{Z}nidari\v{c}}},\ }\bibfield  {title} {\bibinfo {title} {{Momentum-dependent quantum Ruelle-Pollicott resonances in translationally invariant many-body systems}},\ }\href {https://doi.org/10.1103/PhysRevE.110.054204} {\bibfield  {journal} {\bibinfo  {journal} {Phys. Rev. E}\ }\textbf {\bibinfo {volume} {110}},\ \bibinfo {pages} {054204} (\bibinfo {year} {2024})}\BibitemShut {NoStop}%
\bibitem [{\citenamefont {Jacoby}\ \emph {et~al.}(2025)\citenamefont {Jacoby}, \citenamefont {Huse},\ and\ \citenamefont {Gopalakrishnan}}]{jacoby2024}%
  \BibitemOpen
  \bibfield  {author} {\bibinfo {author} {\bibfnamefont {J.~A.}\ \bibnamefont {Jacoby}}, \bibinfo {author} {\bibfnamefont {D.~A.}\ \bibnamefont {Huse}},\ and\ \bibinfo {author} {\bibfnamefont {S.}~\bibnamefont {Gopalakrishnan}},\ }\bibfield  {title} {\bibinfo {title} {{Spectral gaps of local quantum channels in the weak-dissipation limit}},\ }\href {https://doi.org/10.1103/PhysRevB.111.104303} {\bibfield  {journal} {\bibinfo  {journal} {Phys. Rev. B}\ }\textbf {\bibinfo {volume} {111}},\ \bibinfo {pages} {104303} (\bibinfo {year} {2025})}\BibitemShut {NoStop}%
\bibitem [{\citenamefont {Zhang}\ \emph {et~al.}(2024{\natexlab{b}})\citenamefont {Zhang}, \citenamefont {Nie},\ and\ \citenamefont {von Keyserlingk}}]{zhang2024}%
  \BibitemOpen
  \bibfield  {author} {\bibinfo {author} {\bibfnamefont {C.}~\bibnamefont {Zhang}}, \bibinfo {author} {\bibfnamefont {L.}~\bibnamefont {Nie}},\ and\ \bibinfo {author} {\bibfnamefont {C.}~\bibnamefont {von Keyserlingk}},\ }\bibfield  {title} {\bibinfo {title} {{Thermalization rates and quantum Ruelle-Pollicott resonances: insights from operator hydrodynamics}},\ }\href {https://arxiv.org/abs/2409.17251} {\bibfield  {journal} {\bibinfo  {journal} {arXiv:2409.17251}\ } (\bibinfo {year} {2024}{\natexlab{b}})}\BibitemShut {NoStop}%
\bibitem [{\citenamefont {Bendixson}(1902)}]{bendixson1902}%
  \BibitemOpen
  \bibfield  {author} {\bibinfo {author} {\bibfnamefont {I.}~\bibnamefont {Bendixson}},\ }\bibfield  {title} {\bibinfo {title} {{Sur les racines d'une équation fondamentale}},\ }\href {https://doi.org/10.1007/BF02419030} {\bibfield  {journal} {\bibinfo  {journal} {Acta Math.}\ }\textbf {\bibinfo {volume} {25}},\ \bibinfo {pages} {359 } (\bibinfo {year} {1902})}\BibitemShut {NoStop}%
\bibitem [{\citenamefont {Mirsky}(1955)}]{mirskylinear}%
  \BibitemOpen
  \bibfield  {author} {\bibinfo {author} {\bibfnamefont {L.}~\bibnamefont {Mirsky}},\ }\href@noop {} {\emph {\bibinfo {title} {{An Introduction to Linear Algebra}}}}\ (\bibinfo  {publisher} {Clarendon Press},\ \bibinfo {address} {Oxford},\ \bibinfo {year} {1955})\BibitemShut {NoStop}%
\bibitem [{\citenamefont {Lieb}(1994)}]{lieb1994}%
  \BibitemOpen
  \bibfield  {author} {\bibinfo {author} {\bibfnamefont {E.~H.}\ \bibnamefont {Lieb}},\ }\bibfield  {title} {\bibinfo {title} {{Flux Phase of the Half-Filled Band}},\ }\href {https://doi.org/10.1103/PhysRevLett.73.2158} {\bibfield  {journal} {\bibinfo  {journal} {Phys. Rev. Lett.}\ }\textbf {\bibinfo {volume} {73}},\ \bibinfo {pages} {2158} (\bibinfo {year} {1994})}\BibitemShut {NoStop}%
\bibitem [{\citenamefont {Klocke}\ \emph {et~al.}(2024)\citenamefont {Klocke}, \citenamefont {Simm}, \citenamefont {Zhu}, \citenamefont {Trebst},\ and\ \citenamefont {Buchhold}}]{klocke2024arXiv}%
  \BibitemOpen
  \bibfield  {author} {\bibinfo {author} {\bibfnamefont {K.}~\bibnamefont {Klocke}}, \bibinfo {author} {\bibfnamefont {D.}~\bibnamefont {Simm}}, \bibinfo {author} {\bibfnamefont {G.-Y.}\ \bibnamefont {Zhu}}, \bibinfo {author} {\bibfnamefont {S.}~\bibnamefont {Trebst}},\ and\ \bibinfo {author} {\bibfnamefont {M.}~\bibnamefont {Buchhold}},\ }\bibfield  {title} {\bibinfo {title} {{Entanglement dynamics in monitored Kitaev circuits: Loop models, symmetry classification, and quantum Lifshitz scaling}},\ }\href {https://arxiv.org/abs/2409.02171} {\bibfield  {journal} {\bibinfo  {journal} {arXiv:2409.02171}\ } (\bibinfo {year} {2024})}\BibitemShut {NoStop}%
\bibitem [{\citenamefont {Bernard}\ and\ \citenamefont {LeClair}(2002)}]{bernard2002}%
  \BibitemOpen
  \bibfield  {author} {\bibinfo {author} {\bibfnamefont {D.}~\bibnamefont {Bernard}}\ and\ \bibinfo {author} {\bibfnamefont {A.}~\bibnamefont {LeClair}},\ }\bibfield  {title} {\bibinfo {title} {{A classification of non-Hermitian random matrices}},\ }in\ \href {https://link.springer.com/chapter/10.1007/978-94-010-0514-2_19} {\emph {\bibinfo {booktitle} {{Statistical Field Theories}}}},\ \bibinfo {editor} {edited by\ \bibinfo {editor} {\bibfnamefont {A.}~\bibnamefont {Cappelli}}\ and\ \bibinfo {editor} {\bibfnamefont {G.}~\bibnamefont {Mussardo}}}\ (\bibinfo  {publisher} {Springer},\ \bibinfo {address} {Dordrecht},\ \bibinfo {year} {2002})\ pp.\ \bibinfo {pages} {207--214}\BibitemShut {NoStop}%
\bibitem [{\citenamefont {Kawabata}\ \emph {et~al.}(2019)\citenamefont {Kawabata}, \citenamefont {Shiozaki}, \citenamefont {Ueda},\ and\ \citenamefont {Sato}}]{kawabata2019PRX}%
  \BibitemOpen
  \bibfield  {author} {\bibinfo {author} {\bibfnamefont {K.}~\bibnamefont {Kawabata}}, \bibinfo {author} {\bibfnamefont {K.}~\bibnamefont {Shiozaki}}, \bibinfo {author} {\bibfnamefont {M.}~\bibnamefont {Ueda}},\ and\ \bibinfo {author} {\bibfnamefont {M.}~\bibnamefont {Sato}},\ }\bibfield  {title} {\bibinfo {title} {Symmetry and topology in {Non-Hermitian} physics},\ }\href {https://doi.org/10.1103/PhysRevX.9.041015} {\bibfield  {journal} {\bibinfo  {journal} {Phys. Rev. X}\ }\textbf {\bibinfo {volume} {9}},\ \bibinfo {pages} {041015} (\bibinfo {year} {2019})}\BibitemShut {NoStop}%
\bibitem [{\citenamefont {Zhou}\ and\ \citenamefont {Lee}(2019)}]{zhou2019PRB}%
  \BibitemOpen
  \bibfield  {author} {\bibinfo {author} {\bibfnamefont {H.}~\bibnamefont {Zhou}}\ and\ \bibinfo {author} {\bibfnamefont {J.~Y.}\ \bibnamefont {Lee}},\ }\bibfield  {title} {\bibinfo {title} {{Periodic table for topological bands with non-Hermitian symmetries}},\ }\href {https://doi.org/10.1103/PhysRevB.99.235112} {\bibfield  {journal} {\bibinfo  {journal} {Phys. Rev. B}\ }\textbf {\bibinfo {volume} {99}},\ \bibinfo {pages} {235112} (\bibinfo {year} {2019})}\BibitemShut {NoStop}%
\bibitem [{\citenamefont {S\'a}\ \emph {et~al.}(2023)\citenamefont {S\'a}, \citenamefont {Ribeiro},\ and\ \citenamefont {Prosen}}]{sa2023PRX}%
  \BibitemOpen
  \bibfield  {author} {\bibinfo {author} {\bibfnamefont {L.}~\bibnamefont {S\'a}}, \bibinfo {author} {\bibfnamefont {P.}~\bibnamefont {Ribeiro}},\ and\ \bibinfo {author} {\bibfnamefont {T.}~\bibnamefont {Prosen}},\ }\bibfield  {title} {\bibinfo {title} {{Symmetry Classification of Many-Body Lindbladians: Tenfold Way and Beyond}},\ }\href {https://doi.org/10.1103/PhysRevX.13.031019} {\bibfield  {journal} {\bibinfo  {journal} {Phys. Rev. X}\ }\textbf {\bibinfo {volume} {13}},\ \bibinfo {pages} {031019} (\bibinfo {year} {2023})}\BibitemShut {NoStop}%
\bibitem [{\citenamefont {Kawabata}\ \emph {et~al.}(2023)\citenamefont {Kawabata}, \citenamefont {Kulkarni}, \citenamefont {Li}, \citenamefont {Numasawa},\ and\ \citenamefont {Ryu}}]{kawabata2023PRXQ}%
  \BibitemOpen
  \bibfield  {author} {\bibinfo {author} {\bibfnamefont {K.}~\bibnamefont {Kawabata}}, \bibinfo {author} {\bibfnamefont {A.}~\bibnamefont {Kulkarni}}, \bibinfo {author} {\bibfnamefont {J.}~\bibnamefont {Li}}, \bibinfo {author} {\bibfnamefont {T.}~\bibnamefont {Numasawa}},\ and\ \bibinfo {author} {\bibfnamefont {S.}~\bibnamefont {Ryu}},\ }\bibfield  {title} {\bibinfo {title} {{Symmetry of Open Quantum Systems: Classification of Dissipative Quantum Chaos}},\ }\href {https://doi.org/10.1103/PRXQuantum.4.030328} {\bibfield  {journal} {\bibinfo  {journal} {PRX Quantum}\ }\textbf {\bibinfo {volume} {4}},\ \bibinfo {pages} {030328} (\bibinfo {year} {2023})}\BibitemShut {NoStop}%
\bibitem [{\citenamefont {Garc\'{\i}a-Garc\'{\i}a}\ \emph {et~al.}(2024)\citenamefont {Garc\'{\i}a-Garc\'{\i}a}, \citenamefont {S\'a}, \citenamefont {Verbaarschot},\ and\ \citenamefont {Yin}}]{garcia2024PRD}%
  \BibitemOpen
  \bibfield  {author} {\bibinfo {author} {\bibfnamefont {A.~M.}\ \bibnamefont {Garc\'{\i}a-Garc\'{\i}a}}, \bibinfo {author} {\bibfnamefont {L.}~\bibnamefont {S\'a}}, \bibinfo {author} {\bibfnamefont {J.~J.~M.}\ \bibnamefont {Verbaarschot}},\ and\ \bibinfo {author} {\bibfnamefont {C.}~\bibnamefont {Yin}},\ }\bibfield  {title} {\bibinfo {title} {{Toward a classification of PT-symmetric quantum systems: From dissipative dynamics to topology and wormholes}},\ }\href {https://doi.org/10.1103/PhysRevD.109.105017} {\bibfield  {journal} {\bibinfo  {journal} {Phys. Rev. D}\ }\textbf {\bibinfo {volume} {109}},\ \bibinfo {pages} {105017} (\bibinfo {year} {2024})}\BibitemShut {NoStop}%
\bibitem [{\citenamefont {Lieu}\ \emph {et~al.}(2020)\citenamefont {Lieu}, \citenamefont {McGinley},\ and\ \citenamefont {Cooper}}]{lieu2020PRL}%
  \BibitemOpen
  \bibfield  {author} {\bibinfo {author} {\bibfnamefont {S.}~\bibnamefont {Lieu}}, \bibinfo {author} {\bibfnamefont {M.}~\bibnamefont {McGinley}},\ and\ \bibinfo {author} {\bibfnamefont {N.~R.}\ \bibnamefont {Cooper}},\ }\bibfield  {title} {\bibinfo {title} {{Tenfold Way for Quadratic Lindbladians}},\ }\href {https://doi.org/10.1103/PhysRevLett.124.040401} {\bibfield  {journal} {\bibinfo  {journal} {Phys. Rev. Lett.}\ }\textbf {\bibinfo {volume} {124}},\ \bibinfo {pages} {040401} (\bibinfo {year} {2020})}\BibitemShut {NoStop}%
\bibitem [{\citenamefont {Altland}\ \emph {et~al.}(2021)\citenamefont {Altland}, \citenamefont {Fleischhauer},\ and\ \citenamefont {Diehl}}]{altland2021PRX}%
  \BibitemOpen
  \bibfield  {author} {\bibinfo {author} {\bibfnamefont {A.}~\bibnamefont {Altland}}, \bibinfo {author} {\bibfnamefont {M.}~\bibnamefont {Fleischhauer}},\ and\ \bibinfo {author} {\bibfnamefont {S.}~\bibnamefont {Diehl}},\ }\bibfield  {title} {\bibinfo {title} {{Symmetry Classes of Open Fermionic Quantum Matter}},\ }\href {https://doi.org/10.1103/PhysRevX.11.021037} {\bibfield  {journal} {\bibinfo  {journal} {Phys. Rev. X}\ }\textbf {\bibinfo {volume} {11}},\ \bibinfo {pages} {021037} (\bibinfo {year} {2021})}\BibitemShut {NoStop}%
\bibitem [{\citenamefont {Kawasaki}\ \emph {et~al.}(2022)\citenamefont {Kawasaki}, \citenamefont {Mochizuki},\ and\ \citenamefont {Obuse}}]{kawasaki2022PRB}%
  \BibitemOpen
  \bibfield  {author} {\bibinfo {author} {\bibfnamefont {M.}~\bibnamefont {Kawasaki}}, \bibinfo {author} {\bibfnamefont {K.}~\bibnamefont {Mochizuki}},\ and\ \bibinfo {author} {\bibfnamefont {H.}~\bibnamefont {Obuse}},\ }\bibfield  {title} {\bibinfo {title} {{Topological phases protected by shifted sublattice symmetry in dissipative quantum systems}},\ }\href {https://doi.org/10.1103/PhysRevB.106.035408} {\bibfield  {journal} {\bibinfo  {journal} {Phys. Rev. B}\ }\textbf {\bibinfo {volume} {106}},\ \bibinfo {pages} {035408} (\bibinfo {year} {2022})}\BibitemShut {NoStop}%
\end{thebibliography}%

\end{document}